\documentclass[draftcls,onecolumn]{IEEEtran}

\usepackage{amsmath,amsfonts,amsthm, amssymb, enumerate}
\usepackage{cite}
\usepackage{lipsum}
\usepackage[justification=centering]{caption}
\usepackage{amssymb,graphicx,subfigure}
\usepackage[bottom]{footmisc}

\newtheorem{defn}{Definition}[section]
\newtheorem{thm}{Theorem}[section]
\newtheorem{lem}[thm]{Lemma}
\newtheorem{prop}[thm]{Proposition}
\newtheorem{rem}{Remark}[thm]

\newtheorem{cor}{Corollary}[thm]

\usepackage{algorithm,algpseudocode}
\algnewcommand\Input{\item[{\textbf{Input:}}]}
\algnewcommand\Output{\item[{\textbf{Output:}}]}
\algnewcommand\Initialize{\item[{\textbf{Initialize:}}]}
\algnewcommand{\return}[1]{
  \State \textbf{return:}
  \Statex \hspace*{\algorithmicindent}\parbox[t]{.8\linewidth}{\raggedright #1}
}

\newcommand{\rank}{\operatorname{rank}}
\newcommand{\krank}{\operatorname{krank}}

\newcommand{\supp}{\operatorname{supp}}
\newcommand{\nulli}{\operatorname{null}}
\begin{document}
\title{Greedy Subspace Pursuit for Joint Sparse Recovery}
\author{Kyung Su~Kim,~\IEEEmembership{Student Member,~IEEE,}
        Sae-Young Chung,~\IEEEmembership{Senior Member,~IEEE}
        
        \thanks{The material in this paper was in part submitted to the 2016 IEEE International Symposium on Information Theory.}
    }

\markboth{ }%
{Shell \MakeLowercase{\textit{et al.}}: Bare Demo of IEEEtran.cls for Journals}

\maketitle

\begin{abstract} 
In this paper, we address the sparse multiple measurement vector (MMV) problem where the objective is to recover a set of sparse nonzero row vectors or indices of a signal matrix from incomplete measurements. Ideally, regardless of the number of columns in the signal matrix, the sparsity ($k$) plus one measurements is sufficient for the uniform recovery of signal vectors for almost all signals, i.e., excluding a set of Lebesgue measure zero. To approach the ``$k+1$'' lower bound with computational efficiency even when the rank of signal matrix is smaller than $k$, we propose a greedy algorithm called Two-stage orthogonal Subspace Matching Pursuit (TSMP) whose theoretical results approach the lower bound with less restriction than the Orthogonal Subspace Matching Pursuit (OSMP) and Subspace-Augmented MUltiple SIgnal Classification (SA-MUSIC) algorithms. We provide non-asymptotical performance guarantees of OSMP and TSMP by covering both noiseless and noisy cases. Variants of restricted isometry property and mutual coherence are used to improve the performance guarantees. Numerical simulations demonstrate that the proposed scheme has low complexity and outperforms most existing greedy methods. This shows that the minimum number of measurements for the success of TSMP converges more rapidly to the lower bound than the existing methods as the number of columns of the signal matrix increases. 
\end{abstract}


\IEEEpeerreviewmaketitle

\section{Introduction}
\label{intr}
In recent years, the compressive sensing (CS) theory \cite{donoho2006compressed}, \cite{candes2006compressive} and its extension has received much attention as means to solve the underdetermined inverse problem to estimate the sparse signal matrix given a multiple measurement matrix. The subject has been studied in many fields of science \cite{gorodnitsky1997sparse,gorodnitsky1995neuromagnetic,jeffs1998sparse,duarte2008single,wright2009robust,cabrera1991extrapolation,jin2010algorithms,chu2004speech,cotter2002sparse,bajwa2008compressed,duttweiler2000proportionate,rao2003adaptive,garudadri2010low,cotter2002sparse,guo2009neighbor}. 

The basic principle of CS is as follows: when the signal matrix is sparse (i.e., when most rows of the matrix  are zeros), the signal matrix can be uniquely determined through the identification of its support -- a set of indices extracted from rows of the signal matrix that include nonzero elements. Once the support is determined, the problem of estimating the signal matrix reduces to a standard overdetermined linear inverse problem, which can be easily solved.

\subsection{Multiple measurement vector problem}
CS can be formulated by the linear structure $Y =\Phi X_0 + W$ given a measurement matrix $Y=\Phi X_0 \in \mathbb{K}^{m \times l}$ and a sensing matrix $\Phi \in \mathbb{K}^{m \times n}$ where $X_0 \in \mathbb{K}^{n \times l}$ is a signal matrix and $W \in \mathbb{K}^{m \times n}$ is a measurement noise. Most compressive sensing theories were developed to address the single measurement vector (SMV) problem (i.e., the case when $l=1$). \cite{mallat1993matching,pati1993orthogonal,tibshirani1996regression,chen1998atomic,gorodnitsky1997sparse,candes2008enhancing,harikumar1996new,delaney1998globally,chartrand2008iteratively,tipping2001sparse,wipf2004sparse,vetterli2002sampling,needell2009cosamp,dai2009subspace,donoho2006compressed,donoho2006stable,bruckstein2009sparse,candes2005decoding,candes2006stable,tropp2004greed,tropp2007signal,donoho2012sparse,wainwright2009sharp ,zhao2006model,wainwright2009information,fletcher2009necessary,wang2010information,akccakaya2010shannon,jin2011limits}. Sparse signal recovery with multiple measurement vectors (MMV) refers to the case when $l\geq 1$, which is also known as the joint sparse recovery problem \cite{cotter2005sparse}, \cite{zdunek2008improved}. Joint sparse recovery has many important applications such as the sub-Nyquist sampling of multiband signals\cite{feng1996spectrum,bresler1996spectrum,feng1998universal,gastpar2000necessary,bresler2008spectrum,mishali2009blind,mishali2010theory},  magnetoencephalography (MEG) and electroencephalography (EEG) \cite{zdunek2008improved,wipf2009unified}, blind source separation \cite{cichocki2005blind}, multivariate regression \cite{obozinski2011support}, and source localization \cite{malioutov2005sparse}.

Compared to the SMV case, the MMV approach is known to have a greatly improved recovery rate \cite{cotter2005sparse,eldar2009robust,eldar2010average,jin2008insights,foucart2011recovering,zhang2011sparse} and yields computational advantages \cite{foucart2011recovering}, \cite{tropp2006algorithms} over the approach of running multiple independent instances of an SMV algorithm.

\subsection{$l_0$ bound}
In the noiseless case $W=0$, an ideal approach (\ref{l0minf1}) to recover $X_0 \in \mathbb{K}^{n \times l} $ in the MMV problem is to minimize the $l_0$ norm of $X \in \mathbb{K}^{n \times l} $ as follows:
\begin{align}\nonumber
\underset{X}{\arg \min}& \left \| X \right \|_0 \\\label{l0minf1} 
\textup{  subject to  }& Y=\Phi X
\end{align}
Davies and Eldar \cite{davies2012rank} extended the works of Chen, Huo \cite{chen2006theoretical} and Feng, Bresler \cite{feng1998universal} to show that the following $l_0$ bound (i.e., $2k-\rank(X_0)+1$) is the minimum $m$ for (\ref{l0minf1}) to ensure the exact recovery of $X_0$. To be more concrete, they showed that (\ref{l0c}) is a sufficient and necessary condition for the solution of (\ref{l0minf1}) to be unique and equal to $X_0$ for any $X_0$.
\begin{align}\label{l0c}
\krank(\Phi)&\geq 2k-r+1,
\end{align}
where $k$ is the number of nonzero rows in $X_0$ and $r$ is the rank of $X_0$.
In the worst case, the MMV problem is not any easier than the SMV problem since they become identical when $X_0$ comprises of a single repetitive vector \cite{eldar2010average}. (\ref{l0c}), however, informs us that depending on the ranks of matrices $Y$ or $X_0$, the required number of measurements can be reduced to less than $2k$, which is known to be the smallest required $m$ in the SMV problem even in the worst case. 
When $\rank(X_0)=k$, the right-hand side of (\ref{l0c}) has a minimum value of $k+1$. 

\subsection{``$k+1$'' bound}
Foucart and Rauhut \cite{foucart2013mathematical} motivated by Wakin's work \cite{wakin2006geometry} showed another condition on the minimum required $m$ for the ideal approach (\ref{l0minf1}) in the noiseless case to recover $X_0$. They showed that the following is a sufficient condition for the solution of (\ref{l0minf1}) to be unique and equal to $X_0$:
\begin{align}\label{l0k1}
\krank(\Phi)&\geq k+1 
\end{align}
and
\begin{align}\label{l0k12}
X_0 \notin  \underset{\underset{\textup{s.t. }|J|=k,J \neq \Omega}{J \in \{1,...,n\}}}{\bigcup} e(J),
\end{align}
where $\Omega$ is the support of $X_0$ and $e(J)$ is defined in (\ref{unq_3}). Since the right hand side of (\ref{l0k12}) has Lebesgue measure zero, it is sufficient that $m$ is simply $k+1$ irrespective of $l$ for almost all $X_0$. Therefore, in the practical case, $k+1$ measurements are ideally sufficient even for the SMV case (i.e., $l=1$). Based on the fact, this number is defined in the rest of paper as the ``$k+1$'' bound which is the minimum $m$ to ensure the exact recovery for almost all $X_0$. 

\subsection{Practical schemes for approaching the ``$k+1$'' bound}
Further work has to be done to determine whether there is a tractable way to achieve (or closely approach) the ``$k+1$'' bound. Before the advancement of compressive sensing, MUltiple SIgnal Classification (MUSIC) \cite{schmidt1986multiple} was proposed to solve the direction-of-arrival (DOA) or the bearing estimation problem with high computation efficiency \cite{malioutov2005sparse,krim1996two}. Bresler and Feng \cite{feng1996spectrum},\cite{feng1998universal} demonstrated that the application of MUSIC to the joint sparse recovery problem can achieve the ``$k+1$'' bound when $\rank(X_0)=k$. Based on the theoretical guarantee, this is one of the most popular and successful DOA estimation algorithms providing both high empirical performance and computational efficiency when $\rank(X_0)=k$. 

One of the main limitations of the MUSIC algorithm, however, is its failure when $\rank(X_0) < k$. This rank defective case is common in the field of CS since most problems in the field face situations where a correlation between signal vectors exists or the number of the common sparsity of signal vectors is larger than the number of measurement vectors. 

Inspired by the MUSIC algorithm, Davies and Eldar proposed a greedy method called the rank aware algorithm (RA-ORMP) \cite{davies2012rank} to overcome the limitations of MUSIC. They showed that the behavior of RA-ORMP is improved when the rank of $X_0$ is increased and proved that the ``$k+1$'' bound can be achieved when $\rank(X_0)=k$. Its empirical performance was significantly better than MUSIC in dealing with multiple measurements even when $\rank(X_0)<k$. 
Similarly to Kim, et al.'s work \cite{kim2012compressive}, Lee, et al. \cite{lee2012subspace} supplemented MUSIC and developed a Subspace-Augmented MUSIC (SA-MUSIC) algorithm which had better performance than MUSIC. They showed that it theorectially and empirically outperformed MUSIC and provided restrictive conditions in approaching the ``$k+1$'' bound in the $\rank(X_0)<k$ case. To recover the partial support before operating SA-MUSIC, Lee, et al. \cite{lee2012subspace} proposed a new greedy method called the Orthogonal Subspace Matching Pursuit (OSMP) which is an extended version of RA-ORMP and is robust to noise. By combining OSMP and SA-MUSIC, they proposed a greedy algorithm called SA-MUSIC+OSMP, which provided better empirical performances at all rank conditions of $X_0$ than most of the existing methods for the MMV problem especially when the number of measurement vectors is relatively large.

\subsection{Comparison to other methods for MMV problem}

Practical algorithms have been developed to address the new challenges in the joint sparse recovery problem. One class of algorithms for solving the MMV problem includes M-OMP \cite{cotter2005sparse}, \cite{gilbert2005simultaneous}, M-FOCUSS \cite{cotter2005sparse}, $l_1$ and $l_2$ minimization method \cite{eldar2009robust}, simultaneous recovery variants of NIHT, NHTP, CoSaMP \cite{blanchard2014greedy}, multivariate group Lasso \cite{cichocki2005blind}, and MSBL \cite{wipf2007empirical} where all can be viewed as direct extensions of their one dimensional counterparts. Another class of algorithms utilized the correlation, stochastic behavior and the subspace structure of $X$ to achieve better performance in sparse signal recovery. The improved M-FOCUSS algorithms \cite{zdunek2008improved}, variants of MSBL such as AR-SBL \cite{zhang2010sparse} or TMSBL \cite{zhang2011sparse}, the correlation-aware framework of LASSO \cite{pal2015pushing}, the approximate message passing scheme exploiting temporal correlation of $X$ \cite{ziniel2013efficient} and MUSIC-like subspace methods \cite{davies2012rank}, \cite{kim2012compressive}, \cite{lee2012subspace} can all be viewed as such examples. Methods other than MUSIC-like methods (i.e., MUSIC, SA-MUSIC, CS-MUSIC, RA-ORMP, OSMP, etc.) \cite{davies2012rank}, \cite{kim2012compressive}, \cite{lee2012subspace} and MSBLs \cite{wipf2006bayesian}, \cite{wipf2007empirical}, however, are not proved to approach the ``$k+1$'' bound even when $\rank(X_0)=k$. Comparison to MSBL and its variants are discussed in detail in Section \ref{sec6}.

\subsection{Our contributions}

The main contributions of this paper are summarized below.

\begin{itemize}
\item A sufficient condition for the success of OSMP (i.e., RA-ORMP in the noiseless case) is theoretically derived which is not stronger than that of SA-MUSIC.    
\item  An improved scheme of RA-ORMP and OSMP called a Two-stage orthogonal Subspace Matching Pursuit (TSMP) is proposed to enhance the efficiency of reconstructing sparse signals in MMV. TSMP requires less restrictive conditions in approaching the ``$k+1$'' bound than other methods. The TSMP consists of the following procedure: 1. Subspace estimation from a signal space $\mathcal{R}(Y)$ (subspace estimation), 2. Iterative selection of $m-1$ multiple candidate indices through OSMP's selection rule (identification), 3. Recovery of the signal matrix $X_0$ and its support from the set of candidate indices (support and signal matrix estimation). Since the last two steps are the main steps, we refer to TSMP as a ``two-stage'' process.

\item Sufficient conditions for $\Phi$ or the minimal requirements for $m$ in TSMP or OSMP to recover the true support and signal matrix $X_0$ are theoretically derived. The analysis is based on the worst-case scenario where the rank of the signal matrix is considered. Under the rank deficient case $\rank(X_0)\leq k$, it is shown in both theoretical and empirical perspectives that the performances of OSMP or the proposed scheme, TSMP, improve as $r$ (i.e., $l$ in most cases) increases. 
The performances are analyzed in terms of fundamental measures such as WRIP \cite{candes2011probabilistic}, a weaker version of the restricted isometry property (RIP) \cite{candes2005decoding}, and a variant of mutual coherence \cite{needell2009uniform} for a submatrix $\Phi$ to make the results more reliable and applicable to a wider class of sensing matrices for real applications.
A different measure expressed by a singular value of the submatrix in $\Phi$ is also introduced to mitigate the successful conditions in terms of WRIP. The results presented in this paper are valid for both noiseless and noisy cases and are non-asymptotic for parameters such as $(m,n,l,k,\rank(X_0))$.

\item In terms of empirical performance, TSMP mostly outperforms previous greedy algorithms and convex relaxation methods as SNR increased in both SMV and MMV cases. The minimum $m$ required for TSMP to recover the support decreases below the $l_0$ bound and more rapidly converges to the ``$k+1$'' bound than most of the existing algorithms as $l$ increased. More is discussed in detail in Section \ref{sec6}.

\end{itemize}

In the SMV case, there have been recent efforts to modify the popular OMP rule with an aim to enhance the recovery performance and computational efficiency by considering more than sparsity level of the $X_0$ indices in the process of estimating the true support. Special treatments such as thresholding, regularization, or pruning are used. Well known examples of such efforts include Stage wise OMP (StOMP) \cite{donoho2012sparse}, Regularized OMP (ROMP) \cite{needell2009uniform}, CoSaMP \cite{needell2009cosamp}, Subspace Pursuit (SP) \cite{dai2009subspace}, and Generalized OMP (GOMP) \cite{wang2012generalized}. Our approach lies on similar grounds with these approaches but extends to the MMV problem. Our proposed scheme provides better empirical and theoretical performances with low complexity and requires only milder conditions for support identification compared to most SMV or MMV algorithms.

\subsection{Organization of this paper}
The remainder of this paper is organized as follows. Notations, the problem statement, and some definitions are introduced in Sections \ref{sec_not}, \ref{sec_pf}, and IIII, respectively. Previous work on OSMP and TSMP for joint sparse recovery are described in Section \ref{sec_ad}. Conditions for joint sparse recovery using an ideal approach and its relation to OSMP and TSMP are discussed in Section \ref{sec_ocmmv}. The performances of OSMP and TSMP measured by variants of RIP and mutual coherence in noiseless and noisy cases are analyzed in Section \ref{sec_pg}. The empirical performances of OSMP and TSMP are compared to other methods in Section \ref{sec_ne} and their relations to relevant works are discussed in Section \ref{sec_d}. Appendices are dedicated to the proofs of our results. 

\section{Notation}\label{sec_not}
Symbol $\mathbb{N}$ denotes the set of natural numbers and $\Sigma$ denotes the set $\{1,...,n\}$ for $n \in \mathbb{N}$. $[i]$ denotes the subset $\{1,...,i\}$ of $\Sigma$. $\Delta_{s} \subseteq \Sigma$ denotes a subset of $\{1,...,n\}$ whose cardinality is $s \in \mathbb{N}$. Symbol $\mathbb{K}$ denotes a scalar field which is either the real field $\mathbb{R}$ or $\mathbb{C}$. The vector space of $d$-tuples over $\mathbb{K}$ is denoted as $\mathbb{K}^d$ for $d \in \mathbb{N}$. Similarly, for $d,n \in \mathbb{N}$, the vector space of $d \times n$ matrices over $\mathbb{K}$ is denoted by $\mathbb{K}^{d \times n}$. We will use some notations for the matrix $A:=[a_1,...,a_n] \in \mathbb{K}^{d \times n}$ whose $i$ th column is $a_i$. The range space spanned by the columns of $A$ is denoted by $\mathcal{R}(A)$. $\supp(A)$ is the support of $A$ and is defined as a set of nonzero row indices of $A$. The Hermitian transpose (transpose) of $A$ are denoted by $A^*$ ($A^\top$), respectively. $A^{\dagger}$ denotes the Moore-Penrose pseudoinverse of $A$. The $i$th column of $A$ is denoted by $a_i$ and the submatrix of $A$ with columns indexed by $J \subseteq \Sigma$ is denoted by $A_J$. The $i$th row of $A$ is denoted by $A^{\{i\}}$ and the submatrix of $A$ with rows indexed by $K \subseteq [m]$ is denoted by $A^{K}$.  The $i$th largest singular value of $A$ is denoted by $\sigma_i(A)$. The Frobenius norm and the spectral norm of $A$ are denoted by $\left \| A \right \|_F$ and $\left \| A \right \|_2$, respectively. For $p,q \in [1, \infty]$, the mixed $l_{p,q}$ norm of $A$ is defined by $\left \| A \right \|_{p,q}:=(\sum\limits^{m}_{k=1}\left \| a^k \right \|^q_p)^{\frac{1}{q}}$ for $q < \infty$ and $\left \| A \right \|_{p,\infty}:=\max\limits_{k \in [m]} \left \| a^k \right \|_p$. The inner product is denoted by $\langle \cdot,\cdot \rangle$. 
For a subspace $S$ of $\mathbb{K}^d$, $\dim(S)$ denotes the dimension of $S$. Matrices $P_{S} \in \mathbb{K}^{d \times d}$ and $P^{\perp}_{S} \in \mathbb{K}^{d \times d}$ denote the orthogonal projection onto $S$ and its orthogonal complement $S^{\perp}$, respectively. Symbols $\mathbb{P}$ and $\mathbb{E}$ denote the probability and the expectation with respect to a certain distribution. For a set $\Gamma \subseteq \Sigma$ and a subspace $\mathcal{R}(A_{\Gamma})$ of $\mathbb{K}^d$, $\dot a_i:= \frac{P^{\perp}_{\mathcal{R}(A_{\Gamma})}a_i}{\left \| P^{\perp}_{\mathcal{R}(A_{\Gamma})}a_i \right \|_2}$ and $\dot A:= [\dot a_1,...,\dot a_n]$ denote the scaled $a_i$ vector with the orthonormal projection onto $\mathcal{R}(A_{\Gamma})^{\perp}$ and the matrix whose $i$th column is $\dot a_i$, respectively. If the denominator of $\dot a_i$ is zero, $\dot a_i$ is defined as a zero vector. 

\section{Fomulation: MMV problem}\label{sec_pf}
A matrix $X_0 \in \mathbb{K}^{n \times l}$ is called row $k$-sparse if it has at most $k$ nonzero rows. $\Omega \subseteq \Sigma$ denotes the support of $X_0$, i.e., $\supp(X_0)$, and its sparsity level denotes the cardinality of $\Omega$. The joint sparse recovery problem is to find the support $\Omega$ and a row $k$-sparse signal matrix $X_0$ from the matrix $Y \in \mathbb{K}^{m \times l}$ using multiple measurement vectors (columns of $Y$) given by
\begin{align}\nonumber
Y = \Phi X_0 + W,
\end{align}
where $\Phi:=[\phi_1,...,\phi_n] \in \mathbb{K}^{m \times n}$ is a common and known sensing matrix whose $i$th column is $\phi_i$ and $W \in \mathbb{K}^{m \times l}$ is a perturbation. We will refer to the case when $\rank(X_0^\Omega)$ has its maximum value $k$ as the full row rank case. Otherwise, the case when $\rank(X_0^{\Omega}) < k$ will be called the rank-defective case \cite{lee2012subspace}.

\section{Some definitions of measure and their properties}
\subsection{Measures and their properties}
\subsubsection{Restricted Isometry Property}
One approximate way to specify which matrices ($\Phi$) the sparse recovery is applicable to is to use the restricted isometry property. The RIP provides upper and lower bounds on the singular values for all submatrices of $\Phi$ by retaining no more than $k$ columns of $\Phi$.
\begin{defn}[RIP, Restricted Isometry Property \cite{candes2005decoding}]
Matrix $A \in \mathbb{K}^{m \times n}$ satisfies the restricted isometry property with parameters $(c,\delta)$ where $c \in \mathbb{R^+}$ and $\delta \in (0,1)$ if there exist constants $(c,\delta)$ such that for $\forall x \in \mathbb{K}^n$, 
\begin{align}\nonumber
c(1-\delta)\left \|  x \right \|^2_2 \leq  \left \| A  x \right \|^2_2 \leq c(1+\delta)\left \|  x \right \|^2_2.
\end{align}
The RIP constant $\delta_{k} $ is defined as the smallest value of $\delta$ that satisfies the restricted isometry property with some positive constant $c$. 
\end{defn}
The RIP of order $k$ implies that all sets of $k$ columns in $\Phi$ are uniformly well conditioned. It, however, requires a strong condition on $\Phi$. A weaker version for the definition of RIP is therefore used:
\begin{defn}[WRIP, Week Restricted Isometry Property \cite{candes2011probabilistic}]
Matrix $A \in \mathbb{K}^{m \times n}$ satisfies the week restricted isometry property (WRIP) with parameters $(J,a,b,c,\delta)$ where $(a,b) \in \mathbb{N}^2$, $J \subseteq \Sigma$ with $|J|=a$, $c \in \mathbb{R^+}$ and $\delta \in (0,1)$ if there exist constants $(c,\delta)$ such that for $\forall x \in \mathbb{K}^{a+b}$ and $\forall K \supseteq J$ with $ |K|=a+b$,
\begin{align}\nonumber
c(1-\delta)\left \|  x \right \|^2_2 \leq  \left \| A_K  x \right \|^2_2 \leq c(1+\delta)\left \|  x \right \|^2_2.
\end{align}
The RIP constant $\delta_{a}(A_J;b) $ is defined as the smallest value of $\delta$ that satisfies the local restricted isometry property with some positive constant $c$. The corresponding WRIP constant is given by
\begin{align}\nonumber
\delta_{a}(A_J;b) =\max \limits_{\underset{|K|=a+b}{K \supseteq J}} \frac{\left \| A^*_KA_K -c I_{a+b} \right \|}{c}= \frac{1-\kappa(J)}{1+\kappa(J)},
\end{align}
where $\kappa(J):= \min \limits_{\underset{|K|=a+b}{K \supseteq J}} \sigma^2_{a+b}(A_K)  / \max \limits_{\underset{|K|=a+b}{K \supseteq J}} \sigma^2_{1}(A_K)$. 
\end{defn}
Since $\delta_{a}(A_J;b) \leq \delta_{a+b}$, having a more mild condition on $A$ is attainable with WRIP than with RIP having the same order. The special case of WRIP with $l_2$-normalized columns of $\Phi$ has been previously proposed  \cite{eldar2010average,candes2011probabilistic,lee2012subspace}. \cite{lee2012subspace} shows that compared to RIP, the required number of $l_2$-normalized columns to guarantee the success of sparse recovery is largely reduced when using WRIP.

\subsubsection{Coherence}
Another concept is used to specify which matrices ($\Phi$) the sparse recovery is applicable. The coherence is defined as follows:
\begin{defn}[LCP, Locally mutual Coherence with the orthogonal complemental Projection] \label{newcoh} Let $\Delta, \Gamma$ be proper subsets of $\Sigma$. The LCP with $(\Delta, \Gamma)$ is defined by  
\begin{align}
\mu(\Delta,\Gamma)= \underset{\{i,j\} \subseteq \Delta \setminus \Gamma}{\max} \,\left| \left \langle \frac{P^{\perp}_{\mathcal{R}(\Phi_{\Gamma})}\phi_i}{\left \| P^{\perp}_{\mathcal{R}(\Phi_{\Gamma})}\phi_i \right \|_2},\frac{P^{\perp}_{\mathcal{R}(\Phi_{\Gamma})}\phi_j}{\left \| P^{\perp}_{\mathcal{R}(\Phi_{\Gamma})}\phi_j \right \|_2} \right \rangle  \right|.
\end{align}
\end{defn}
The special case when $(\Delta,\Gamma)=(\Sigma,\o)$ corresponds to the worst-case coherence \cite{needell2009uniform} \cite{donoho2001uncertainty} \cite{bajwa2012two} which is computed with lower computational complexity than RIP. This has been widely used to analyze a sufficient number of measurements for the success of many CS algorithms.
\subsection{Some definitions}
The following definitions are used throughout this paper.
\begin{defn} \label{krank}
The Kruskal rank of a matrix $A$, denoted by $\krank(A)$, is the maximal number $q$ in which any $q$ columns of $A$ are independent.
\end{defn} 

\begin{defn} \label{rownond}
Matrix $X$ is row-nondegenerate if 
\begin{align}\label{rown}
\krank(X^*)=\rank(X).
\end{align}
\end{defn}
(\ref{rown}) implies that every $i$ row vectors of $X$ are linearly independent for
$i \leq \rank(X)$. This is satisfied by $X$ whose row vectors are in general position \cite{donoho2005sparse}. This is a property of the subspace of $\mathcal{R}(X)$ since it is equivalent to $\krank(B^*)=\rank(X)$ for any orthonormal basis $B$ of $\mathcal{R}(X)$ \cite{lee2012subspace}. This condition holds if each row of $X^{\Omega}_0$ is independently and identically sampled from any probability measure which is non-singular with respect to the Lebesgue measure. Since most probability distributions defined in continuous fields such as the random Gaussian matrix follow this property and the elements of a signal matrix are statistically assumed in continuous fields in most applications of joint sparse recovery, the above condition could therefore be satisfied without major restrictions.

\section{Algorithm describtion}
\label{sec_ad}
\subsection{Existing scheme: OSMP}

The OSMP algorithm is designated as Algorithm 1. The OSMP comprises of two steps which are described below.

\subsubsection{Subspace estimation from $\mathcal{R}(Y)$} Step 1 of OSMP is to estimate an $r$-dimensional subspace $\hat S$ from $\mathcal{R}(Y)$ via an arbitrary subspace estimator. The estimator is not usually needed, i.e., $\hat S = \mathcal{R}(Y)$, since $\rank(\Phi_{\Omega}X_0^{\Omega}) \geq l$ in the majority of cases for joint sparse recovery. The consideration of an estimator could provide a better solution in the noisy and $\rank(\Phi_{\Omega}X_0^{\Omega}) < l$ cases. We will discuss this in detail in Section \ref{sec_d}.

\subsubsection{Index selection method for estimating support} Step 2 of OSMP is to extract a set of $k$ indices $\hat \Omega$ to estimate the true support $\Omega$ through Algorithm 4 (submp). During each iteration step in submp, a set of indices selected before preceding steps denoted by $\Gamma$ is given and an index $i \in \Sigma \setminus \Gamma$ is selected such that the angle between spaces $P^{\perp}_{\mathcal{R}(\Gamma)}\hat S$ and $P^{\perp}_{\mathcal{R}(\Gamma)}\mathcal{R}(\phi_i)$ is minimized.  

\subsection{Proposed scheme: TSMP}
TSMP algorithms are designated as Algorithms 2 or 3 (TSMP$_1$ or TSMP$_2$) depending on whether the sparsity $k$ is known or unknown. Each algorithm consists of three steps:

\begin{algorithm}
  \caption{OSMP($k$)}
    \begin{algorithmic}[1]
    \Input{$Y \in \mathbb{K}^{m \times l}, \Phi \in \mathbb{K}^{m \times n}$, $k\in \mathbb{N}.$}
    \Output{$\hat \Omega\subseteq \Sigma$}
        \State $\hat S \in \mathbb{K}^{m \times r}\gets \textup{estimate a signal subspace from $\mathcal{R}(Y)$}$ 
		\State $\hat \Omega \gets $ submp($\hat S,\o,k$) 
      \State \Return{$\hat \Omega$}
  \end{algorithmic}
\end{algorithm}

\begin{algorithm}
  \caption{TSMP$_1$($k$): $k$ is known}
    \begin{algorithmic}[1]
    \Input{$Y \in \mathbb{K}^{m \times l}, \Phi \in \mathbb{K}^{m \times n}$, $k\in \mathbb{N}.$}
    \Output{$\Omega_c,\hat \Omega\subseteq \Sigma$, $\hat X \in  \mathbb{K}^{n \times l}$}
        \State $\hat S \in \mathbb{K}^{m \times r}\gets \textup{estimate a signal subspace from $\mathcal{R}(Y)$}$ 
		\State $\Omega_c \gets $ submp($\hat S,\o,m-1$) 
        \State $\hat \Omega,\hat X \gets$ ESMS$_1$($\Omega_c,k$) 
      \State \Return{$\Omega_c,\hat \Omega,\hat X$}
  \end{algorithmic}
\end{algorithm}

\begin{algorithm}
  \caption{TSMP$_2$($\kappa$): $k$ is unknown}
    \begin{algorithmic}[1]
    \Input{$Y \in \mathbb{K}^{m \times l}, \Phi \in \mathbb{K}^{m \times n}$, $\kappa \in \mathbb{R}.$}
    \Output{$\Omega_c,\hat \Omega\subseteq \Sigma$, $\hat X \in  \mathbb{K}^{n \times l}$}
        \State $\hat S \in \mathbb{K}^{m \times r}\gets \textup{estimate a signal subspace from $\mathcal{R}(Y)$}$ 
		\State $\Omega_c \gets $ submp($\hat S,\o,m-1$) 
        \State $\hat \Omega,\hat X \gets$ ESMS$_2$($\Omega_c,\kappa$) 
      \State \Return{$\Omega_c,\hat \Omega,\hat X$}
  \end{algorithmic}
\end{algorithm}
\subsubsection{Subspace estimation from $\mathcal{R}(Y)$} Step 1 of TSMP$_1$ or TSMP$_2$ is to estimate an $r$-dimensional subspace $\hat S$ from $\mathcal{R}(Y)$. It is the same as that of OSMP.
\subsubsection{Index selection for support candidates} Step 2 of Algorithms 2 or 3 is to extract a set $ \Omega_c$ consisting of $m-1$ indices for the support candidates through Algorithm 4 (submp, sub-algorithm of matching pursuit). The subroutine of submp($\hat S,\o,m-1$) in TSMP has the same structure with that of submp($\hat S,\o,k$) in OSMP (i.e., RA-ORMP in the noiseless case). Compared to OSMP which only selects indices with sparsity level $k$, TSMP selects indices with sparsity level $m-1$ using submp. Since $m$ is not related to the sparsity level, it is not necessary to know the sparsity level to operate submp.  Constructing the projection operators in OSMP or TSMP could be performed by QR decomposition which reduces the complexity \cite{lee2012subspace}. Algorithm 7 shows an example of TSMP$_1$ via QR decomposition. 

\begin{algorithm}
  \caption{submp($\hat S, \Gamma_0,s$)}
  \begin{algorithmic}[1]
    \Input{$\hat S \in \mathbb{K}^{m \times r}, \Phi \in \mathbb{K}^{m \times n}, \Gamma_{0} \subseteq \Sigma$, $s  \in \mathbb{N} $.} 
    \Output{$\Gamma \subseteq \Sigma$}
                     \State $\Gamma \gets  \Gamma_0$    
    \For{i = 1 to $s$}
                     \State $a_i \gets \underset{l \in \Sigma\setminus \Gamma}{\arg \max} \frac{\left \|(P_{\mathcal{R}(P^{\perp}_{\mathcal{R}(\Phi_\Gamma)}\hat S)})\phi_l\right \|_2} {\left \|P^{\perp}_{\mathcal{R}(\Phi_\Gamma)}\phi_l\right \|_2} $
                     \State $\Gamma \gets \Gamma \cup \{ a_i\}$
      \EndFor
      \State \Return{$\Gamma$}
  \end{algorithmic}
\end{algorithm}

\subsubsection{Estimation of signal matrix and support using support candidates} Step 3 of Algorithms 2 or 3 is to estimate $\Omega$ and $X_0$ as $\hat \Omega$ and $\hat X$, respectively through Algorithms 5 or 6 (ESMS, Estimation of Signal Matrix and Support) using the triplet $(\Omega_c,\Phi,Y)$ as their inputs. Algorithms 5 (ESMS$_2$) and 6 (ESMS$_3$) are used either when the sparsity $k$ is known or unknown, respectively. If the sparsity is unknown, ESMS$_2$ additionally detects the sparsity by thresholding using a fixed parameter $\kappa$. Conditions for $\Phi$, $X$, and $W$ such that each algorithm recovers $\Omega$ will be shown in Section \ref{sec_pg}. Though the explicit method on how to set up $\kappa$ in ESMS$_2$ when $k$ is unknown will not be discussed, an actual implementation might use the following methods: 1. Detect the largest gap between two $\zeta_l$'s and set up $\kappa$ to distinguish the two $\zeta_l$'s. 2. Set up $\kappa$ as an estimation of the weighted noise level (i.e., the expected value of $\left \| ( \Phi_{J}^{\dagger} W) ^{\{i\}}\right \|_2$ with respect to $W$ and $i \in \Sigma$).

\begin{algorithm}\label{algo1}
  \caption{ESMS$_1$($J,k$)}
  \begin{algorithmic}[1]
    \Input{$Y \in \mathbb{K}^{m \times l}, \Phi \in \mathbb{K}^{m \times n},  k \in \mathbb{N}, J \subseteq \Sigma$}
    \Output{$Q \subseteq \Sigma$, $\hat {X} \in \mathbb{K}^{n \times l}$}
        \State $\bar{X}^{J}  \gets ( \Phi_{J})^{\dagger} Y$
    \For{$l \in J $}   
    \State   $\zeta_l \gets  \left \| \bar{X}^{\{l\}}\right \|_2 $ 
    \EndFor
    \State $Q \gets \{ \textup{indices of the $k$-largest $\zeta_l$'s } \}$
    \State $\hat {X}^{Q}  \gets ( \Phi_{Q})^{\dagger} Y$
    \State \Return{${ Q,\hat X}$}
  \end{algorithmic}
\end{algorithm}

\begin{algorithm}\label{algo1}
  \caption{ESMS$_2$($J,\kappa$)}
  \begin{algorithmic}[1]
    \Input{$Y \in \mathbb{K}^{m \times l}, \Phi \in \mathbb{K}^{m \times n},  k \in \mathbb{N}, J \subseteq \Sigma$}
    \Output{$Q \subseteq \Sigma$, $\hat {X} \in \mathbb{K}^{n \times l}$}
         \State $\bar{X}^{J}  \gets ( \Phi_{J})^{\dagger} Y$
    \For{$l \in J $}  
    \State   $\zeta_l \gets  \left \| \bar{X}^{\{l\}}\right \|_2 $ 
    \EndFor
    \State $Q \gets \{ \textup{indices of $\zeta_l$'s which are larger than $\kappa$ } \}$
    \State $\hat {X}^{Q}  \gets ( \Phi_{Q})^{\dagger} Y$
    \State \Return{${ Q,\hat X}$}
    \end{algorithmic}
\end{algorithm}

\begin{algorithm}
  \caption{TSMP$_1$($k$) with QR decomposition}
    \begin{algorithmic}[1]
    \Input{$Y \in \mathbb{K}^{m \times l}, \Phi \in \mathbb{K}^{m \times n}$, $k\in \mathbb{N}.$}
    \Initialize{$T:=[t_{1},...,t_n]=0 \in \mathbb{K}^{m \times n}$, $\bar X, \hat X=0 \in \mathbb{K}^{n \times l}$, $A=I_m \in \mathbb{K}^{m \times m}$, $\Gamma = \o$.}
    \Output{$\Omega_c,\hat \Omega\subseteq \Sigma$, $\hat X \in  \mathbb{K}^{n \times l}$}
    \State Estimate an $r$-dimensional signal subspace as an orthnomal basis $U_y \in \mathbb{K}^{m \times r}$ from $\mathcal{R}(Y)$.
    \State Compute $\bar \Phi=[\bar \phi_{1},...,\bar \phi_{n}]$ s.t. $\bar \phi_{j} = \phi_{j} / \left \| \phi_{i} \right \|_2$ for $j \in \Sigma$.
	\For{i = 1 to $m-1$}
    \State Compute $QR$ decomposition of $U_y$ as $AU_y=Q_1R_1$.
    \State $ T_{\Sigma \setminus \Gamma} \gets A \bar \Phi_{\Sigma \setminus \Gamma}$
    \State $ a  \gets \underset{b \in \Sigma \setminus \Gamma}{\arg \max} \left \| Q^*_1 t_b \right \|_2 / \left \|  t_b \right \|_2$
    \State $\Gamma \gets \Gamma \cup \{a\}$   
    \State Compute $QR$ decomposition of $\bar \Phi_{\Gamma}$ as $\bar \Phi_{\Gamma}=Q_2R_2$.
    \State $A \gets I_m - Q_2Q_2^*$.
    \EndFor
            \State $\Omega_c \gets \Gamma$ 
            \State Compute $QR$ decomposition of $\Phi_{\Omega_c}$ as $\Phi_{\Omega_c}=Q_3R_3$.
      \State $\bar{X}^{\Omega_c}  \gets (R^*_3R_3)^{-1}\Phi^*_{\Omega_c} Y$
    \For{$l \in \Omega_c $}   
    \State   $\zeta_l \gets  \left \| \bar{X}^{\{l\}}\right \|_2 $ 
    \EndFor
    \State $\hat \Omega \gets \{ \textup{indices of the $k$-largest $\zeta_l$'s } \}$ 
    \State Compute $QR$ decomposition of $\Phi_{\hat \Omega}$ as $\Phi_{\hat \Omega}=Q_4R_4$.
    \State $\hat{X}^{\hat \Omega}  \gets (R^*_4R_4)^{-1}\Phi^*_{\hat \Omega} Y$
      \State \Return{$\Omega_c,\hat \Omega,\hat X$}
  \end{algorithmic}
\end{algorithm}

\section{Ideal condition for MMV}
\label{sec_ocmmv}
For the noiseless case in the ideal approach, sufficient and necessary conditions for the recovery of $X_0$ or its support are (\ref{unq_1}) or (\ref{unq_2}) due to the following results. 

\begin{thm}(\cite[Theorem 2]{davies2012rank})\label{thm_unq_1}
Either (\ref{unq_1}) or (\ref{unq_2}) is necessary and sufficient for the measurement matrix $Y=\Phi X_0$ to uniquely determine the true signal matrix $X_0$ from $\{ \forall X \in \mathbb{K}^{n \times l} | |\supp(X)|\leq k \}$ where $|\supp(X_0)|=k$.
\begin{align}\label{unq_1}
\krank(\Phi)>2k-\rank(Y)\\\label{unq_2}
\krank(\Phi)>2k-\rank(X_0)
\end{align}
\end{thm}
Theorem \ref{thm_unq_1} shows that the $l_0$ bound is equal to $2k+1-\rank(X_0)$ and no recovery algorithm can uniformly guarantee their success with smaller than the $l_0$ bound. This result also implies that the minimum value of the $l_0$ bound, $k+1$, can only be achieved when $\rank(X^{\Omega}_0)$ has full row rank (i.e., $\rank(X_0)=k$). Theorem \ref{thm_unq_3}, however, shows that if $\Phi$ or $X_0$ does not belong to a certain set with Lebesgue measure zero, the sufficient condition on required $m$ for the recovery of $X_0$ or its support reduces to $k+1$ irrespective to $\rank(X_0)$ or $\rank(Y)$ (i.e., $l$). Based on the result of Theorem \ref{thm_unq_3}, $m=k+1$ is defined as the ``$k+1$'' bound which provides a better lower bound for the minimum required $m$ for the successful recovery than the $l_0$ bound. This implies that a tractable algorithm whose minimum required $m$ is smaller than the $l_0$ bound (i.e., $2k+1-\rank(X_0)$) may exist irrespective of $l$.

\begin{thm}(\cite[Theorem 2.16]{foucart2013mathematical})\label{thm_unq_3} The measurement matrix $Y=\Phi X_0$ uniquely determines the true signal matrix $X_0$ from $\{ \forall X \in \mathbb{K}^{n \times l} | |\supp(X)|\leq k \}$ where $|\supp(X_0)|=k$  if and only if $\sigma_{k}(\Phi_{\Omega})>0$ and $X_0 \notin  \underset{J \in \Sigma, |J|=k,J \neq \Omega}{\bigcup} e(J)$ where
\begin{align}\label{unq_3}
e(J):=\{&X:=[x_1,...,x_l] \in \mathbb{K}^{n \times l} |\sigma_{k+1}([ \Phi_J, \Phi x_i]) = 0 \textup{ for all $i \in [l]$}\}.
\end{align}
\end{thm}
\begin{proof}[Proof of Theorem \ref{thm_unq_3}] Sufficiency follows from Theorem 2.16 in \cite{foucart2013mathematical}.
Next, neccessity will be proved by the following two steps. 
First, suppose that $\Phi (X_0 - X) \neq 0$ for $X (\neq X_0) \in \mathbb{K}^{n \times l}$ from $\{ \forall X \in \mathbb{K}^{n \times l} | |\supp(X):=\Delta|\leq k \}$. Then $\min\limits_{\underset{|\Delta|\leq k}{\Delta \subseteq \Sigma}}\sigma_{|\Omega \cup \Delta|}(\Phi_{\Omega \cup \Delta})>0$ is guaranteed so that $\sigma_{k}(\Phi_{\Omega})>0$ holds.  
Next, suppose that $X_0 \notin  \underset{J \in \Sigma, |J|=k,J \neq \Omega}{\bigcup} e(J)$. Then there exists a set of indices $Q \subseteq \Sigma$ such that $\mathcal{R}(Y) \subseteq \mathcal{R}(\Phi_Q)$, $|Q| = k$, and $Q\neq \Omega$. Therefore $X_0$ is not the unique solution of (\ref{l0minf1}).
\end{proof}
\begin{rem}
Suppose that $\krank(\Phi)\geq k+1$. Then $e(J)$ has Lebesgue measure zero and so does its finite union $\underset{J \in \Sigma, |J|=k,J \neq \Omega}{\bigcup} e(J)$.
\end{rem}

\section{Relationship between the ideal condition and OSMP/TSMP} 
\label{sec6}
In this section, tight sufficient conditions for the success of OSMP and TSMP are provided under certain constraints. These results are valid for both the noiseless and noisy cases and non-asymptotic for $(m,n,l,k,r)$.

\subsection{Measurement of noise magnitude}
It is assumed that there exists an estimator for extracting an $r$-dimensional subspace, $\hat S$, from $\mathcal{R}(Y)$ in Step 1 of the OSMP/TSMP algorithms. The following function is defined.
\begin{align}\label{precond1}
\rho(\hat S,\Phi_{\Omega}X_0^{\Omega}):=\min\limits_{\underset{\textup{s.t. }\dim(\bar S)=\dim(\hat S)}{\bar S \subseteq \mathcal{R}(\Phi_{\Omega}X_0^{\Omega})} } \left \| P_{\hat{S}}-P_{\bar{S}}\right \| 
\end{align}

$\rho(\hat S,\Phi_{\Omega}X_0^{\Omega})$ is simply denoted as $\rho(\hat S)$ in the rest of the paper.
$\rho(\hat S)$ increases as the noise power increases and $\rho(\hat S)$ is zero for any $r$-dimensional subspace $\hat{S}$ from $\mathcal{R}(Y)$ in the noiseless case. This means that in the noiseless case Step 1 of OSMP/TSMP is not needed, i.e., $(r,\hat S)$ is set to $(\rank(Y),Y)$. For these reasons, $\rho(\hat S)$ will be used as a measure for noise magnitude.
\subsection{Relationship between the optimality condition and OSMP/TSMP}
The following family of index subsets is defined:
\begin{align}\nonumber
t(a,b) &:= \{\forall J \subseteq \Sigma | \, |J \cap \Omega| \geq a, \,|J\cup \Omega| \leq b+|\Omega| \}
\end{align}

\begin{thm}\label{propmc_st3}
Let $\eta$ be a constant such that $\rho(\hat S)\leq \eta \leq 0.5$ where $\hat S$ is an $r$-dimensional space. Suppose that $X_0^{\Omega}$ is row-nondegenerate and $\krank(\Phi) \geq k+v_1$. Then, for any $\Gamma \in t(k-r,v_1)$, $\Omega \setminus \Gamma$ belongs to a set of indices selected by submp($\hat S,\Gamma,v_2$) such that $v_2\geq |\Omega \setminus \Gamma|$ if any of the following conditions hold:
\begin{align}\label{ex_msc1}
a_1(v_1)&<\frac{1- 4 \eta (1-\eta) }{1+4 \eta (1-\eta) } \\\label{ex_msc2}
a_2(v_1)&> 4 \eta (1-\eta)\\\label{ex_msc4}
a_3(v_1)&> 4 \eta (1-\eta),
\end{align}
where
\begin{align}\nonumber
a_1(x)&:= {\delta_{k}(\Phi_{\Omega};x+1)}\\\nonumber
a_2(x)&:= \underset{\Delta_{x} \subseteq \Sigma \setminus \Omega}{\min} \left[ \frac{\underset{i \in \Sigma \setminus \Omega \cup \Delta_{x}}{\min} \, \sigma^2_{k+x+1}(\Phi_{\Omega \cup \Delta_{x} \cup \{i\}  })}{\underset{j \in \Sigma \setminus \Omega \cup \Delta_{x}}{\max} \,\sigma^2_{1}(\Phi_{\Omega \cup \Delta_{x} \cup \{j\}  })}\right] \\ \nonumber
a_3(x)&:= \frac{\underset{\Delta_{(x+1)} \subseteq \Sigma \setminus \Omega}{\min}\, \sigma^2_{k+x+1}(\Phi_{\Omega \cup \Delta_{x+1}})}{\left \| \phi^{\max}_{\Sigma }\right \|^2_2}.
\end{align}
\end{thm}

\begin{proof}[Proof of Theorem \ref{propmc_st3}] 
See Appendix A.
\end{proof}

\begin{rem}
\label{rem1n2_4} 
Condition $(\ref{ex_msc4})$ when $(\Gamma,v_1,v_2,r) = (\o,0,k,k)$ and $(\Gamma,v_1,v_2) = (\o,0,r)$ covers the conditions in \cite[Theorem $7.1$]{lee2012subspace} and \cite[Theorem $7.7$]{lee2012subspace}. This fact implies that the theoretical guarantees of OSMP or TSMP are beyond that of SA-MUSIC$+$OSMP.
\end{rem}

The following are some definitions of some events. 
\begin{itemize}
\item $C_1(i)$: An event where OSMP succeeds to recover the first $i$ indices up to the $i$th step
\item $C_2(i,j)$: An event where TSMP succeeds to produce the first $i+j$ indices up to the $i+j$th step such that at least $i$ indices from the set of $i+j$ indices belong to the true support $\Omega$
\end{itemize}

\begin{cor}\label{corpropmc_st3}
Let $\eta$ be a constant such that $\rho(\hat S)\leq \eta \leq 0.5$ where $\hat S$ is an $r$-dimensional space. Suppose that $X_0^{\Omega}$ is row-nondegenerate. Then the following two statements hold.
\begin{itemize}
\item OSMP produces $\Omega$ as its output $\hat \Omega$ if $C_1(k-r)$ holds and any of the following conditions (\ref{corex_msc1})--(\ref{corex_msc4}) hold.
\begin{align}\label{corex_msc1}
a_1(0)&<\frac{1- 4 \eta (1-\eta) }{1+4 \eta (1-\eta) } \\\label{corex_msc2}
a_2(0)&> 4 \eta (1-\eta)\\\label{corex_msc4}
a_3(0)&> 4 \eta (1-\eta)
\end{align}
\item TSMP guarantees that $\Omega \subseteq \Omega_c$ where $\Omega_c$ is one of its outputs if $C_2(k-r, m-k-1)$ holds and any of the following conditions (\ref{cor2ex_msc1})--(\ref{cor2ex_msc4}) hold.
\begin{align}\label{cor2ex_msc1}
a_1(m-k-1)&<\frac{1- 4 \eta (1-\eta) }{1+4 \eta (1-\eta) } \\\label{cor2ex_msc2}
a_2(m-k-1)&> 4 \eta (1-\eta)\\\label{cor2ex_msc4}
a_3(m-k-1)&> 4 \eta (1-\eta)
\end{align}
\end{itemize} 
\end{cor}
\begin{proof}[Proof of Corollary \ref{corpropmc_st3}]
Theorem \ref{cortm2} where $(v_1,v_2)=(0,r)$ and $(v_1,v_2)=(m-k-1,r)$ guarantees (\ref{corex_msc1})--(\ref{corex_msc4}) and (\ref{cor2ex_msc1})--(\ref{cor2ex_msc4}), respectively.
\end{proof}

\begin{rem}
\label{rem1n2_3} 
The fact that the left-hand side of (\ref{corex_msc1}) is smaller than its uniform analog $\delta_{k+1}$ ($\delta_{k+1}<\frac{1- 4 \eta (1-\eta) }{1+4 \eta (1-\eta) }$) provides a uniform guarantee that OSMP recovers $\Omega$. In the noiseless case (i.e., $\eta=0$), the above condition reduces to $\delta_{k+1}<1$ so that $\krank(\Phi)>k$, which corresponds to (\ref{unq_3}) in Theorem \ref{thm1n2}. Therefore, for any $\Phi$ such that $\krank(\Phi)>k$, $k+1$ measurements are sufficient for OSMP to identify $\Omega$ in the noiseless case if $C_1(k-r)$ holds.

\end{rem}

\begin{rem}
If $C_2(k-r,m-k-1)$ holds and $\delta_{m}<\frac{1- 4 \eta (1-\eta) }{1+4 \eta (1-\eta) }$ for any $m >k$, TSMP ensures $\Omega \subseteq \Omega_c$ with $m$ measurements (Corollary \ref{corpropmc_st3}). Further, by Theorem \ref{sceosmp1}, TSMP identifies $\Omega$ if the following two conditions hold in the noiseless case (i.e., $\eta=0$); 1. $\Omega \subseteq \Omega_c$, 2. $\sigma_{|\Omega_c|}(\Omega_c)>0$.  Therefore, if $C_2(k-r,m-k-1)$ holds for any $m >k$, TSMP guarantees that $\hat \Omega = \Omega$ for any $\Phi$ such that $\delta_{m}<1$ or $\krank(\Phi)>m-1$ in the noiseless case.

\end{rem}

Since $\mathbb{P}(C_1(k-r))\leq \mathbb{P}(C_2(k-r,m-k-1))$ for $m$ $(>k)$ of OSMP and TSMP, the following relationship between the minimum value of $m$ for TSMP and OSMP to ensure the uniform recovery given a sparsity $k$ ($m_{TSMP}^*$ and $m_{OSMP}^*$) holds in the case of no noise and $\krank(\Phi)=m$:

\begin{align}\nonumber
&m_{OSMP}^*(=m_{SA-MUSIC}^*)\\\nonumber
&= \underset{\bar m}{\arg \min} \{ \bar m >k |\mathbb{P}(C_1(k-r))=1\} \\\nonumber
&\geq \underset{\hat m}{\arg \min} \{ \hat m >k |\mathbb{P}(C_2(k-r,\hat m-k-1))=1\} \\\nonumber
&= m_{TSMP}^*,
\end{align}
where $m_{SA-MUSIC}^*$ is the minimum $m$ to guarantee the uniform recovery of SA-MUSIC+OSMP given the same $k$.
This indicates that TSMP demands a smaller $m$ for the perfect recovery of $\Omega$ than OSMP and SA-MUSIC+OSMP at least in the high SNR region.

\section{Performance guarantee} 
\label{sec_pg}

In this section, we will analyze the non-asymptotical performances of OSMP and TSMP by considering both the noiseless and noisy cases. The following functions will be used. 

\begin{align}\nonumber
f(\eta,k,r)&:=\min\{f_1(\eta,k,r),f_2(\eta,k,r)\}\\\nonumber
f_1(\eta,k,r)&:=\left[(\frac{k}{k+r})(2\eta \sqrt{\frac{r}{k}}+\sqrt{\frac{k+r}{k}-4\eta^2})\right]^2  \\\nonumber
f_2(\eta,k,r)&:=\frac{1}{[\sqrt{{\frac{k}{r}}(\eta^2)+2}-\sqrt{\frac{k}{r}}\eta]^2}\\\nonumber
\lambda(x)&:=\frac{x + \sqrt{x^2+4x}}{1-2\sqrt{x}}
\end{align}

\subsection{Performance analysis for OSMP}
Two results of performance guarantees for OSMP will be shown through Theorems \ref{second_mainthm1} and \ref{cortm2}. 

\subsubsection{First approach}  Theorem \ref{second_mainthm1} provides a performance guarantee for OSMP by assuming that $\Phi$ follows a probability distribution $U(\Phi)$ which means that each $l_2$-normalized column vector of $\Phi \in \mathbb{R}^{m \times n}$ has a uniform distribution on an $m-1$ dimensional unit sphere. This assumption is valid for numerous probability distributions of $\Phi$ such as
\begin{itemize}
\item Gaussian model: each element of $\Phi$ is sampled independently from the standard normal distribution
\item Spherical model: each column of $\Phi$ is sampled independently and uniformly at random from the real sphere $\mathbb{S}^{m-1}$
\end{itemize}

\begin{thm}\label{second_mainthm1}
Suppose that $X_0^{\Omega}$ is row-nondegenerate. Let $\Phi = [\phi_1,...,\phi_n] \in \mathbb{R}^{m \times n}$ be a matrix where $\frac{\phi_i}{\left \|\phi_i \right \|_2}$ ($i \in [1,...,n]$) is independently and uniformly distributed on the $m-1$ dimensional unit sphere in $\mathbb{S}^{m-1}$. Let $\eta$ be a constant such that $\mathbb{P}(\rho(\hat S)\leq \eta \leq 0.5)=  1 $ for some $r$-dimensional space $\hat S$. Let $z$ be defined by
\begin{align}\nonumber
z:=\min\left\{\frac{1-f(\eta,k,r)}{k},\frac{1-4\eta(1-\eta)}{r}\right\}.  
\end{align}
Suppose that $m > k+\max \left\{4, \frac{1}{{\lambda^{-1}(z)}}\right\}\ln \, (4k^2 n/\epsilon)$. 
Then $\mathbb{P}_s$, the probability that submp($\hat S,\o,k$) recovers $\Omega$, exceeds $1-\epsilon$.
\end{thm}

\begin{proof}[Proof of Theorem \ref{second_mainthm1}]
See Appendix B.
\end{proof}

\begin{rem}\label{rem_gs1} For the noiseless case (i.e., $\eta=0$), $z$ reduces to $\frac{1}{2k}$. If $\Phi$ follows $U(\Phi)$ and $m$ is larger than the following quantity (\ref{rem_gs1ex1}), Theorem \ref{second_mainthm1} guarantees that $\Omega$ is fully recovered by OSMP with a probability higher than  $1- \epsilon$.
\begin{align}\label{rem_gs1ex1}
\max \left\{ k+\frac{\ln \, (4k^2 n/\epsilon)}{ \lambda^{-1}(\frac{1}{2k})},k+  4 \ln \, (4k^2 n/\epsilon)  \right \}
\end{align}
Note that $\frac{1}{\lambda^{-1}(\frac{1}{2k})} \leq 70k$ when $k \leq 100$ since $\lambda(x)<35x$ for $x>0.005$. 
\end{rem}

\subsubsection{Second approach} Theorem \ref{cortm2} provides another performance guarantee for OSMP in terms of the singular value of $\Phi$'s submatrix.

\begin{thm}\label{cortm2}
Let $\eta$ be a constant such that $\rho(\hat S)\leq \eta \leq 0.5$ where $\hat S$ is an $r$-dimensional space. Suppose that $X_0^{\Omega}$ is row-nondegenerate and $\sigma_{k}(\Phi_{\Omega})>0$. Then OSMP recovers $\Omega$ if both conditions $(\ref{thcond1})$ and $(\ref{thcond2})$ hold. 
\begin{align}\label{thcond1}
\sqrt{\frac{r}{k}}\alpha - \sqrt{1-\beta^2}-2 \eta >0 
\end{align}
\begin{align}\label{thcond2}
\min\limits_{\underset{\textup{s.t. }|\Gamma|=k-r}{\Gamma \subseteq \Omega}} \,\underset{i \in \Sigma \setminus (\Omega \cup \Gamma) }{\min}\sigma^2_{|\Omega \setminus \Gamma|+1} (\dot\Phi_{\Omega \cup \{i\} \setminus \Gamma } )> 4 \eta (1-\eta),
\end{align}
where
\begin{align}\nonumber
&\alpha:= \min\limits_{\underset{\textup{s.t. }|\Gamma|<k-r}{\Gamma \subseteq \Omega}} \, \sigma_{|\Omega \setminus \Gamma|}(\dot\Phi_{{\Omega} \setminus \Gamma})  \\\nonumber
&\beta:=  \min\limits_{\underset{\textup{s.t. }|\Gamma|<k-r}{\Gamma \subseteq \Omega}} \,  \underset{i \in \Sigma \setminus ({\Omega} \cup \Gamma)}{\min}\,\sigma_{|\Omega \setminus \Gamma|+1}(\dot\Phi_{(\{i\} \cup {\Omega} ) \setminus \Gamma}).
\end{align}
Both of the conditions $(\ref{thcond1})$ and $(\ref{thcond2})$ are also implied by any of the following conditions $(a)$--$(c)$.
\begin{enumerate}[(a)]
\item $a_1(0)<\min\left\{\frac{1-f_1(\eta,k,r)}{1+f_1(\eta,k,r)},\frac{1- 4 \eta (1-\eta) }{1+4 \eta (1-\eta) }\right\}$
\item $a_2(0)> \max\{f_1(\eta,k,r),4 \eta (1-\eta)\}$
\item $a_3(0)> \max\{f_1(\eta,k,r),4 \eta (1-\eta)\}$
\end{enumerate}
\end{thm}

\begin{proof}[Proof of Theorem \ref{cortm2}]
$(\ref{thcond1})$ holds when (\ref{thm_g3_given3_cor}) (i.e., $s_1(\bar \alpha ,\bar \beta,\eta,k,r)>0$) is satisfied in Corollary \ref{thm_g3_3_cor}. $(\ref{thcond2})$ holds when (\ref{mu_sup_c1}) with $(v_1,v_2)=(0,r)$ is satisfied in Corollary \ref{cor_music_c}. From Corollary \ref{thm_g3_3_cor}, submp($\hat S,\o,k-r$) produces a set of $k-r$ indices $\Gamma$ such that $\Gamma \subseteq \Omega$ if $(\ref{thcond1})$ holds. From Corollary \ref{cor_music_c}, submp($\hat S,\Gamma,r$) produces the remained $r$ indices $\Omega \setminus \Gamma$ as its output if $(\ref{thcond2})$ holds. Thus, OSMP identifies $\Omega$ if both of the conditions $(\ref{thcond1})$ and $(\ref{thcond2})$ hold.

Since the proofs of Theorem \ref{new3prp_3} and Theorem \ref{propmc_st3} with $(v_1,v_2)=(0,r)$ show that any of the conditions $(a)$--$(c)$ is a sufficient condition for both of the conditions $(\ref{thcond1})$ and $(\ref{thcond2})$, satisfying any of the conditions $(a)$--$(c)$ implies that OSMP recovers $\Omega$. 

\end{proof}

\begin{rem}\label{rem1_cortm2} Note that $f_1(\eta,k,r)$ reduces to $\frac{k}{k+r}$ for the noiseless case (i.e., $\eta=0$). By conditions $(a)$ and $(c)$, Theorem \ref{cortm2} guarantees that $\Omega$ is fully recovered by OSMP if any of the following conditions hold.
\begin{itemize}
\item ${\delta_{k}(\Phi_{\Omega};1)}<\frac{r}{2k+r} $
\item ${\delta_{k}(\Phi_{\Omega};1)}<\frac{r}{k+r} $ when each of the column vector in $\Phi$ is $l_2$-normalized
\end{itemize}
\end{rem}

Theorem \ref{cortm2} and Remark \ref{rem1_cortm2} show theoretically that OSMP guarantees its success as well as SA-MUSIC+OSMP under non-asymptotical analysis. Better conditions can be expressed by the weak-1 asymmetric RIP \cite{lee2012subspace} derived from condition $(b)$ in Theorem \ref{cortm2}.

As corollaries from the result of Theorem \ref{cortm2}, the minimum $m$ for the success of OSMP with the statistical assumption that $\Phi$ is either a random Gaussian matrix with arbitrary variance or a random partial discrete Fourier matrix (DFT) is evaluated. 

\begin{cor}
\label{3cortt_thm_g3_3}
Suppose that $X_0^{\Omega}$ is row-nondegenerate. Let $\Phi \in \mathbb{R}^{m \times n}$ be a matrix whose entries are i.i.d. Gaussian following $\mathcal{N}(0,\sigma^2)$. Let $\eta$ be a constant such that $\mathbb{P}(\rho(\hat S)\leq \eta \leq 0.5)=  1 $ for some $r$-dimensional space $\hat S$. Let $\theta(\tau):= \frac{1-\tau}{1+\tau}$, $\tau:=\max\{f_1(\eta,k,r),4 \eta (1-\eta)\}$. Suppose that $m \geq \frac{2}{(\sqrt{1+\theta(\tau)}-1)^2}\left[k+2\ln \,\left(\frac{2(n-k)}{\epsilon}\right)\right]$. Then, $\Omega$ is fully recovered by OSMP with a probability higher than $1 - \epsilon$. 
\end{cor}

\begin{proof}[Proof of Corollary \ref{3cortt_thm_g3_3}]
Combining the condition $(b)$ in Theorem \ref{cortm2} and Corollary \ref{3cor_thm_g3_3} with $A=\Phi$ completes the proof.
\end{proof}

\begin{cor}
\label{4cortt_thm_g3_3} 
Suppose that $X_0^{\Omega}$ is row-nondegenerate. Let $\eta$ be a constant such that $\mathbb{P}(\rho(\hat S)\leq \eta \leq 0.5)= 1 $ for some $r$-dimensional space $\hat S$. Let $\{c_1,...,c_m\} \subseteq \Sigma$ be a set of indices selected uniformly at ramdom. Let $\tau$ be $\min\left\{\frac{1-f_1(\eta,k,r)}{1+f_1(\eta,k,r)},\frac{1- 4 \eta (1-\eta) }{1+4 \eta (1-\eta) }\right\}$. For $j=1,...,m$, let the $j$th row of $\Phi$ be the $c_j$th row of the $n \times n$ DFT matrix divided by $\sqrt{m}$. Suppose that $m \geq \frac{2(3+\tau)(k+1)}{3\tau^2}\,\ln \, \left(\frac{2(k+1)(n-k)}{\epsilon} \right)$.  Then, $\Omega$ is fully recovered by OSMP with a probability higher than $1- \epsilon$.

\end{cor}

\begin{proof}[Proof of Corollary \ref{4cortt_thm_g3_3}]
Combining the condition $(a)$ in Theorem \ref{cortm2} and Proposition \ref{prop_dft2} with $A=\Phi$ completes the proof.
\end{proof}

\begin{rem}\label{rem1t1_cortm2}
Note that $\theta(\tau)$ in Corollary \ref{3cortt_thm_g3_3} or $\tau$ in Corollary \ref{4cortt_thm_g3_3} is equal to $\frac{r}{2k+r}$ in the noiseless case (i.e., $\eta=0$).
\end{rem}

\subsection{Performance analysis for TSMP} 
This section provides the performance guarantee of TSMP (i.e., TSMP$_1$ and TSMP$_2$) in two stages. For the output triplet ($\hat \Omega,\Omega_c,\hat X$) of TSMP, Theorem \ref{paeosmp} shows a sufficient condition for TSMP to guarantee $\Omega_c \supseteq \Omega$ in the first stage. For the next stage, Theorem \ref{sceosmp1} gives a sufficient condition for TSMP to produce $\Omega$ as its output $\hat \Omega$ if $\Omega_c \supseteq \Omega$ holds. Therefore, another main result of this paper is obtained, a guarantee for TSMP, by combining the conditions of Theorems \ref{paeosmp} and \ref{sceosmp1}.

\subsubsection{Sufficient condition that TSMP guarantees $\Omega \subseteq \Omega_c$}

\begin{thm}\label{paeosmp}
Suppose that $X_0^{\Omega}$ is row-nondegenerate. Let $\Phi \in \mathbb{R}^{m \times n}$ be a matrix whose elements are i.i.d. Gaussian following $\mathcal{N}(0,\sigma^2)$. Let $\eta$ be a constant such that $\mathbb{P}(\rho(\hat S)\leq \eta \leq 0.5)=  1$ for some $r$-dimensional space $\hat S$. Let $z$ be defined by $\min\left\{\frac{1-f(\eta,k,r)}{k},\frac{1-4\eta(1-\eta)}{r}\right\}$. Suppose that $m \geq k+t+\max \left\{4, \frac{1}{\lambda^{-1}(z)}\right\} \ln \, (4k n/\epsilon)$ for $t>k$. Then $\Omega_c$ (i.e., one of TSMP's outputs) includes $\Omega$ with a probability higher than $1-\sum_{i=t-k+1}^{t} {t\choose i} \epsilon^i$.
\end{thm}
\begin{proof}[Proof of Theorem \ref{paeosmp}]
See Appendix C.
\end{proof}

\begin{rem}\label{rem_paesomp1} For the noiseless case (i.e., $\eta=0$), $z$ reduces to $\frac{1}{2k}$. For the TSMP's output $\Omega_c$, Theorem \ref{paeosmp} guarantees that $\Omega_c \supseteq \Omega$ with a probability higher than $1-\sum_{i=t-k+1}^{t} {t\choose i} \epsilon^i$ if $m$ is larger than the following quantity:
\begin{align}\nonumber
k+t+\max \left \{4, \frac{1}{\lambda^{-1}(\frac{1}{2k})}\right\} \ln \, (4k n/\epsilon)
\end{align}
Note that $\frac{1}{\lambda^{-1}(\frac{1}{2k})} \leq 70k$ when $k \leq 100$ since $\lambda(x)<35x$ for $x>0.005$. 
\end{rem}

\subsubsection{Sufficient conditions that TSMP guarantees $\Omega=\hat \Omega$ given $\Omega \subseteq \Omega_c$}

\begin{thm}\label{sceosmp1} The following two statements are satisfied.

\begin{itemize}
\item TSMP$_1$($|\Omega|$) identifies $\Omega$ as its output $\hat \Omega$ if $\Omega \subseteq \Omega_c$ is satisfied for $\Omega_c$ (one of the TSMP$_1$'s outputs) and the following condition holds.
\label{sse_eval} 
\begin{align}\label{sse_eval_eq00}
\min\limits_{a \in \Omega}   \left \| {X_0}^{\{a\}} \right \|_2 >\frac{2\left \| W^* \right \|_{2,\infty}}{\sigma_{m}(\Phi_{\Omega_c})}
\end{align}
\item TSMP$_2$($\kappa$) identifies $\Omega$ as its output $\hat \Omega$ if $\Omega \subseteq \Omega_c$ is satisfied for $\Omega_c$ (one of the TSMP$_2$'s outputs) and the following condition holds. 
\label{sse_eval} 
\begin{align}\label{sse_eval_eq001}
\frac{\left \| W^* \right \|_{2,\infty}}{\sigma_{m}(\Phi_{\Omega_c})} < \kappa \leq \min\limits_{a \in \Omega}   \left \| {X_0}^{\{a\}}\right \|_2-\frac{\left \| W^* \right \|_{2,\infty}}{\sigma_{m}(\Phi_{\Omega_c})}
\end{align}
\end{itemize}
\end{thm}

\begin{proof}[Proof of Theorem \ref{sceosmp1}]
See Appendix D.
\end{proof}

\begin{cor}\label{sceosmp2} Suppose that each element of $W$ follows Gaussian distribution $N(0,\sigma^2)$.
Then TSMP$_1$($|\Omega|$) or TSMP$_2$($\kappa$) identifies $\Omega$ as its output $\hat \Omega$ with a probability higher than $ 1-m \cdot \exp(-\frac{c(\kappa,\sigma,X_0,\Omega_c)^2}{2})$, where 
\begin{align}\label{sse_eval_eq_cor2}
c(s,\bar \sigma,X,J):=\frac{s \cdot \sigma_{|J|}(\Phi_J)}{2\bar \sigma}-\sqrt{|J|}-1
\end{align}
if the followings three conditions hold: 1. $\Omega \subseteq \Omega_c$ for $\Omega_c$ (i.e., one of TSMP's outputs), 2. $c(\Phi,X_0,\Omega_c)>0$ and 3. (\ref{sse_eval_eq_cor1}).
\label{sse_eval} 
\begin{align}\label{sse_eval_eq_cor1}
0 < 2 \kappa &\leq \min\limits_{a \in \Omega}   \left \| {X_0}^{\{a\}}\right \|_2
\end{align}
\end{cor}

\begin{proof}[Proof of Corollary \ref{sceosmp2}]
See Appendix E.
\end{proof}

\begin{rem}\label{rem_sceosmp0_21}
Theorem \ref{sceosmp1} guarantees that if $\Omega_c \supseteq \Omega$ and $\sigma_{|\Omega_c|}(\Phi_{\Omega_c})>0$, TSMP$_1$($|\Omega|$) or TSMP$_2$($\kappa$) with any $\kappa$ satisfying (\ref{sse_eval_eq_cor1}) yields $\Omega$ as its output for the noiseless case $W=0$. Corollary \ref{sceosmp2} guarantees that if $\Omega_c \supseteq \Omega$ and $ \sigma_{|\Omega_c|}(\Phi_{\Omega_c})>0$, as $\sigma$ goes to zero, TSMP$_1$($|\Omega|$) or TSMP$_2$($\kappa$) such that $\kappa$ satisfies (\ref{sse_eval_eq_cor1}) produces $\Omega$ as its output with a probability increasing and converging to one.
\end{rem}

\begin{rem}\label{rem_sceosmp0_2} For the noiseless case, if $\Phi$ follows $U(\Phi)$ and $m$ is larger than the following quantity (\ref{fq1}), both Theorems \ref{paeosmp} and \ref{sceosmp1} (Remarks \ref{rem_paesomp1} and \ref{rem_sceosmp0_21}) guarantee that TSMP$_1$($|\Omega|$) or TSMP$_2$($\kappa$) where $\kappa$ satisfies (\ref{sse_eval_eq_cor1}) produces $\Omega$ as its output with a probability higher than  $1-\sum_{i=t-k+1}^{t} {t\choose i} \epsilon^i$.

\begin{align}\label{fq1}
k+t+\max \left \{4, \frac{1}{\lambda^{-1}(\frac{1}{2k})}\right\} \ln \, (4k n/\epsilon)
\end{align}
Note that $\frac{1}{\lambda^{-1}(\frac{1}{2k})} \leq 70k$ when $k \leq 100$ since $\lambda(x)<35x$ for $x>0.005$. 
\end{rem}

\begin{rem}\label{rem_sceosmp0_2signalm} For the noiseless case with $\sigma_{|\Omega|}(\Phi_{\Omega})>0$, TSMP$_1$($|\Omega|$) or TSMP$_2$($\kappa$) where $\kappa$ satisfies (\ref{sse_eval_eq_cor1}) recovers $X_0$ if $\hat \Omega=\Omega$.

\end{rem}

By comparing Remarks \ref{rem_gs1} and \ref{rem_sceosmp0_2}, it is shown that in the noiseless case, the probability of recovery failure for TSMP($\sum_{i=t-k+1}^{t} {t\choose i} \epsilon^i$) is considerably smaller than that of OSMP($\epsilon$). Since Theorems \ref{paeosmp} and \ref{sceosmp1} also cover the noisy case in terms of $\eta$, the minimum required number of measurements for the success of TSMP in the noisy case may be obtained for any set of finite values $(m,n,l,k,r)$.

\label{fig1}
\begin{figure}[t]
  \centering
  \footnotesize
  \subfigure[\scriptsize Noiseless case\label{fig:ctwo}]{\includegraphics[width=9cm, height=6.4cm]{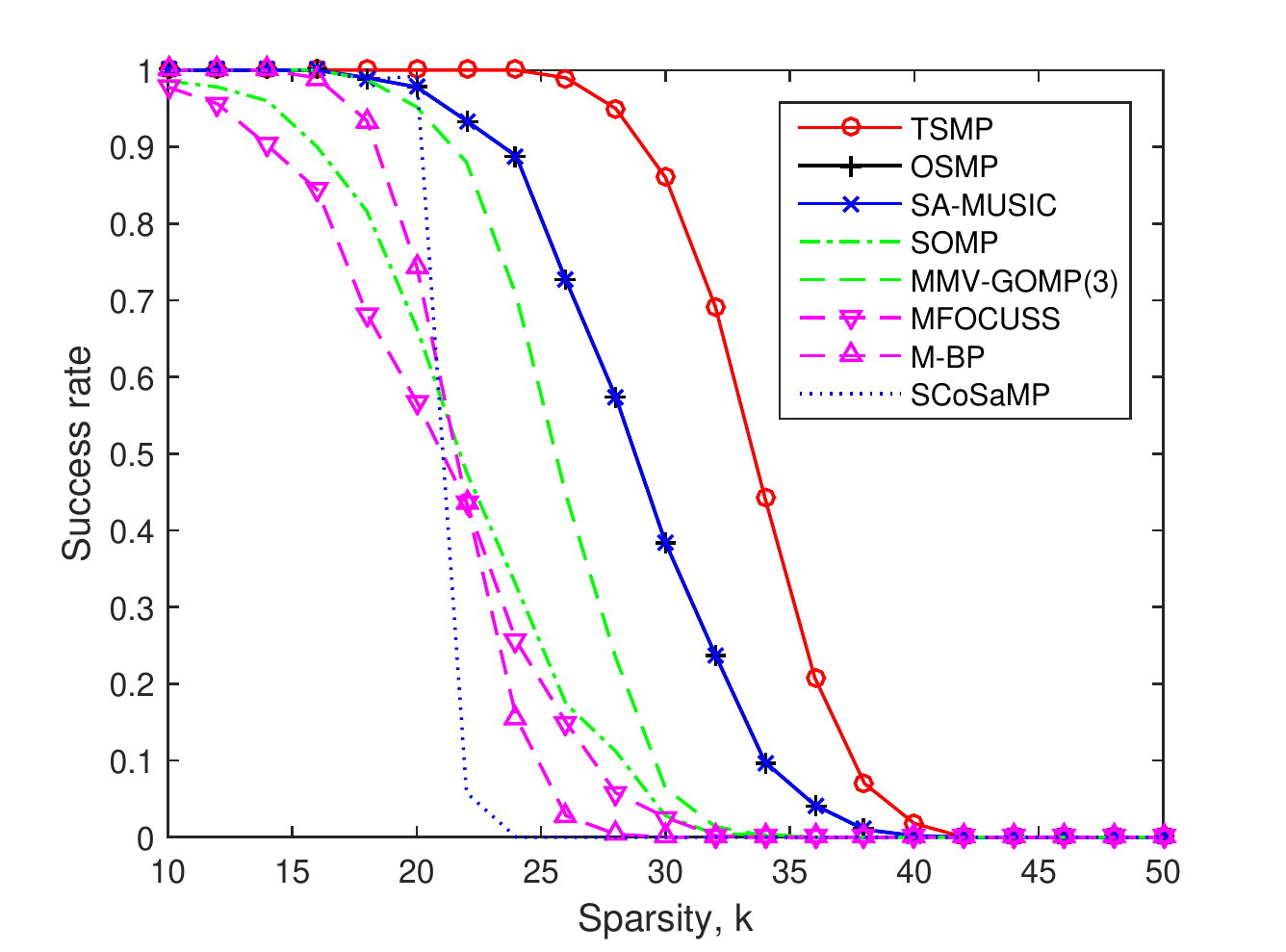}}
  \subfigure[\scriptsize SNR = $40$ dB\label{fig:ttwo}]{\includegraphics[width=9cm, height=6.4cm]{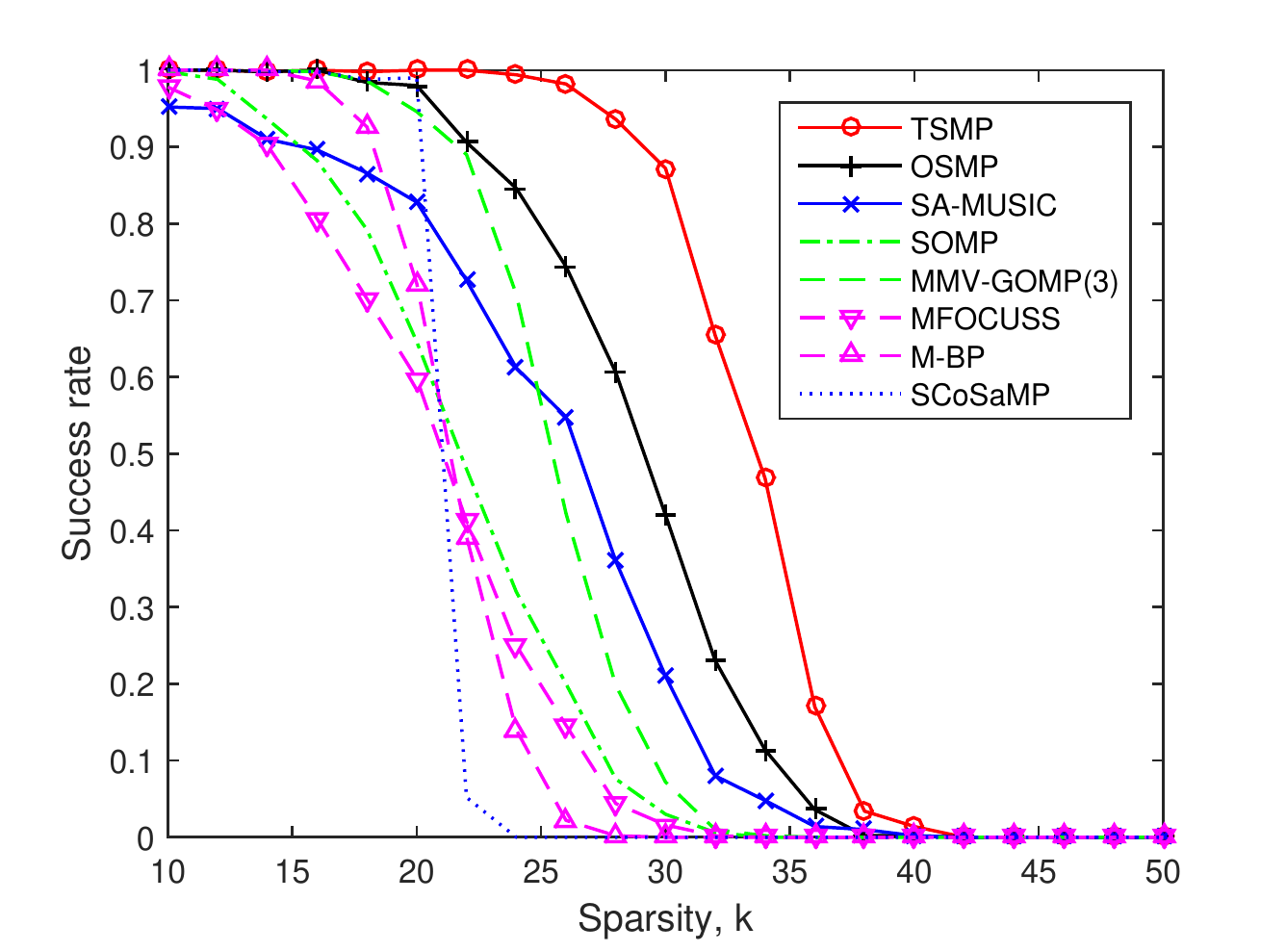}}
\caption{Success rate of true support recovery when $k$ is varied from 10 to 50 and $(m,n,l,r)=(64,512,3,3)$}
\end{figure}
\label{fig2}
\begin{figure}[t]
  \centering
  \footnotesize
  \subfigure[\scriptsize Noiseless case]{\includegraphics[width=9cm, height=6.4cm]{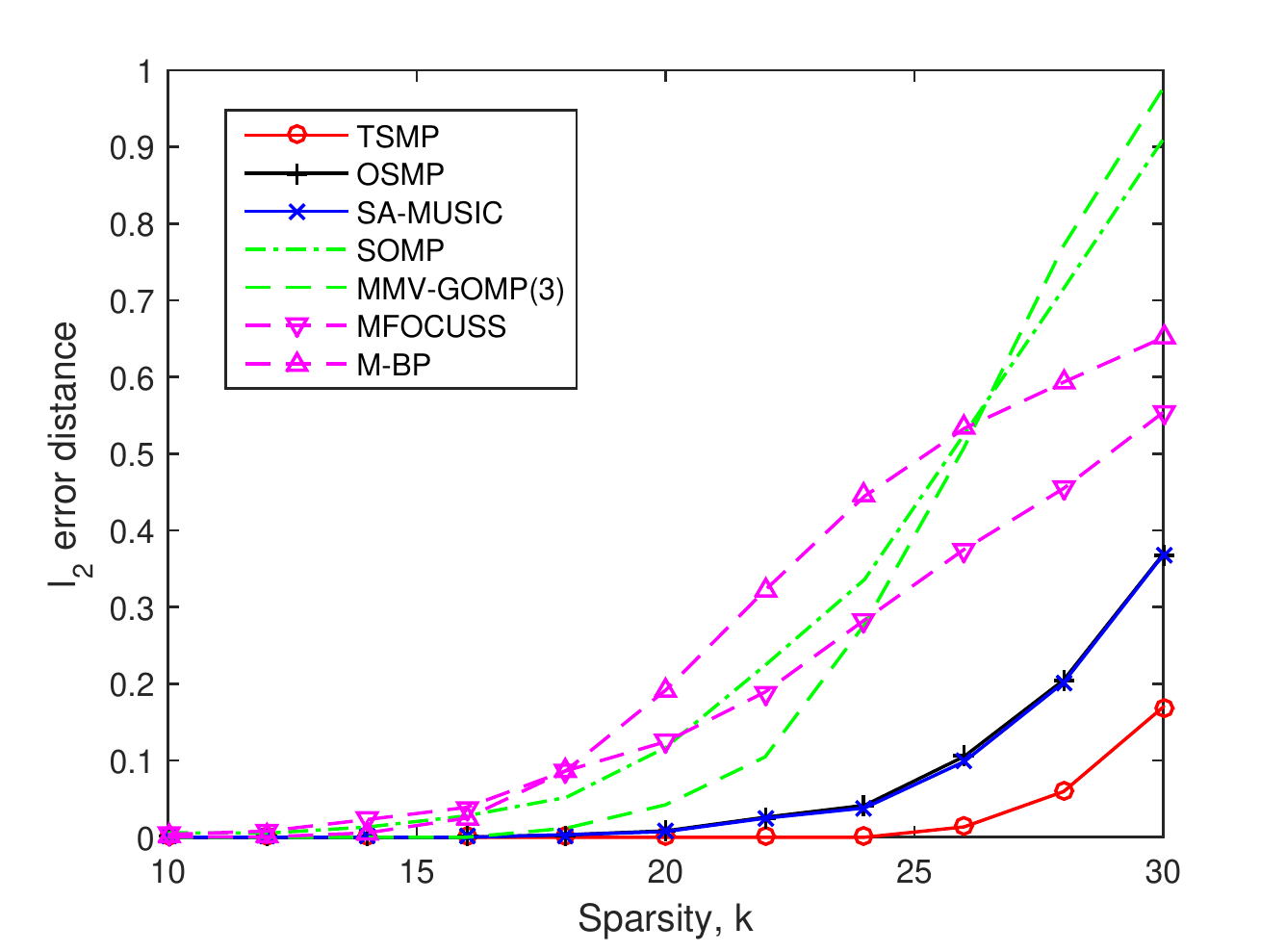}}
  \subfigure[\scriptsize SNR = $40$ dB]{\includegraphics[width=9cm, height=6.4cm]{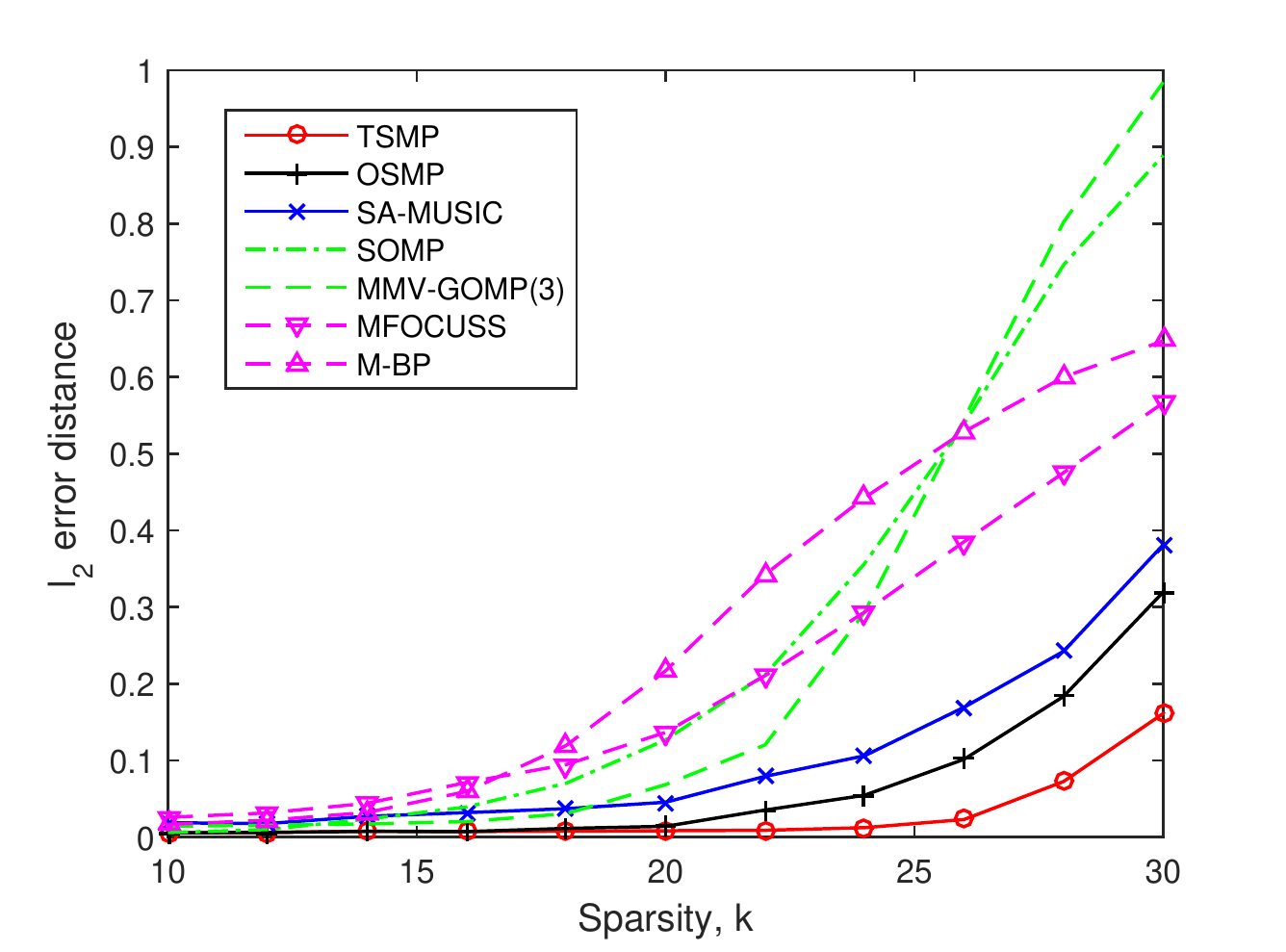}}
\caption{Empirical mean value of $l_2$ norm distance between $X_0$ and its estimated matrix when $k$ is varied from 10 to 30 and $(m,n,l,r)=(64,512,3,3)$}
\end{figure}
\section{Numerical experiments}
\label{sec_ne}
In this section, the performance of the proposed scheme TSMP$_1$ versus conventional MMV algorithms such as OSMP\cite{davies2012rank}\cite{lee2012subspace}, SA-MUSIC (i.e., SA-MUSIC+OSMP in this section) \cite{lee2012subspace},  MMV basis pursuit (M-BP, i.e., the $l_{2,1}$ norm minimization) \cite{chen2006theoretical,BergFriedlander:2008}\footnote{The standard deviation of the noise ($\sigma_w$) was used as a fixed input as an upper bound of the noise error in M-BP.}, MFOCUSS\cite{cotter2005sparse} and SOMP \cite{chen2006theoretical}, SCoSaMP\cite{blanchard2014greedy}\footnote{SCoSaMP was designed in consideration of the stopping criteria in \cite{blanchard2014greedy}.}, MMV-GOMP($t$)\footnote{MMV-GOMP($t$) is a direct extension of the GOMP algorithm \cite{wang2012generalized} for the MMV case such that t indices are selected per iteration.} are demonstrated. A probabilistic model, denoted by i.i.d. complex Gaussian model ($\textup{ICN}(a,b)$), was used for generating $\Phi$, $X^{\Omega}_0$, and $W$. If a matrix $A$ follows the i.i.d. complex Gaussian model $\textup{ICN}(a,b)$, the real and imaginary parts of each entry of $A$ are chosen independently according to Gaussian distribution with mean $a$ and variance $b$. The measurement matrix $\Phi \in \mathbb{K}^{m \times n}$ followed $\textup{ICN}(0,\sigma^2)$. All empirical results have similar behaviors irrespective of $\sigma$ despite its value is set to $1$ in this paper.\footnote{Each of the empirical results in this paper has the same trend as their corresponding results where $\Phi$ was generated from randomly selected $m$ rows from the $n \times n$ DFT matrix.} The support set $\Omega$ of the sparse coefficient matrix $X_0$ was uniformly generated at random such that $|\Omega|=k$ while the signal matrix $X_0^{\Omega}$ was established by the following model $X_0^{\Omega}= V_1 \Lambda V_2^*$: 
\begin{itemize}
\item $V_1 \in \mathbb{K}^{k \times r}$, $V_2 \in \mathbb{K}^{l \times r}$, and $\Lambda  \in \mathbb{K}^{r \times r}$ are independently set by $r$ uniformly random orthonormal columns of a same size matrix whose elements are independently generated by $\textup{ICN}(0,1)$, respectively. 
\end{itemize}
Since $\max\{r,l\}$ is equivalent to the rank of $X^{\Omega}_0$ in the above model, the effect of the rank defect (i.e., $\rank(\Phi X_0)<k$) can be observed by setting $\max\{r,l\}$ less than $k$. To compare the performance in the noisy case, the average per-sample signal-to-noise ratio (SNR) was defined as the ratio of the powers of the measured signal and noise.
\begin{align}\nonumber
\textup{SNR} := \frac{\mathbb{E}\left \| \Phi_{\Omega} X^{\Omega}_0 \right \|^2_{F}}{\mathbb{E}\left \|W \right \|^2_{F} }
\end{align}
In the noisy case, $W$ followed $\textup{ICN}(0,\sigma_w^2)$. In the noiseless case, $W=0$ so that $(\hat S,\eta,r) =(Y,0,\rank(Y))$. The performance was assessed by the following two metrics: first by the rate of successful support recovery and second by the $l_2$ norm distance between $X_0$ and its estimate $\hat X$ (i.e., $\left \| X_0 - \hat X\right \|_2 $). 
Under the above settings, we observed that in the majority of cases, TSMP exhibited the best empirical performance among the mentioned algorithms for the recovery of true signal matrix and its support irrespective of $(m,n,l,r)$ as long as the condition $m<n$ held and the SNR was larger than a certain level. Only the performance results in the $l \leq k$ case is shown in this paper since the condition $l \leq k$ is preferred than $l > k$ in many applications. This will be discussed in more detail in Section \ref{sec_d}. The case $\max\{r,l\} \leq k$ may be reduced to the following cases: $(a)$ $r=l\leq k$ and $(b)$ $r<l\leq k$. To compare the performances in case $(b)$, a common subspace estimator for the first process of TSMP, OSMP, and SA-MUSIC+OSMP was implemented to build up noise robustness. Cases when no subspace estimator is used at all and when the subspace estimator proposed in \cite{lee2012subspace} is used were both considered. Since our empirical results in case $(b)$ showed the same trends with results of case $(a)$, the description of case $(b)$ is omitted and only the results of case $(a)$ are shown. The results of various simulation runs were plotted in the following figures each produced with 500 iterations. In all the figures, $\Phi$, $X_0$, and $W$ were generated in the real field.\footnote{Though all of the figures covered the case where $\Phi$, $X_0$, and $W$ were generated in the real field, the empircal results had the same trend even when $\Phi$, $X_0$, and $W$ were sampled in the complex field.}
\label{fig3}
\begin{figure}[h]
  \centering
  \footnotesize
  \subfigure{\includegraphics[width=9cm, height=6.4cm]{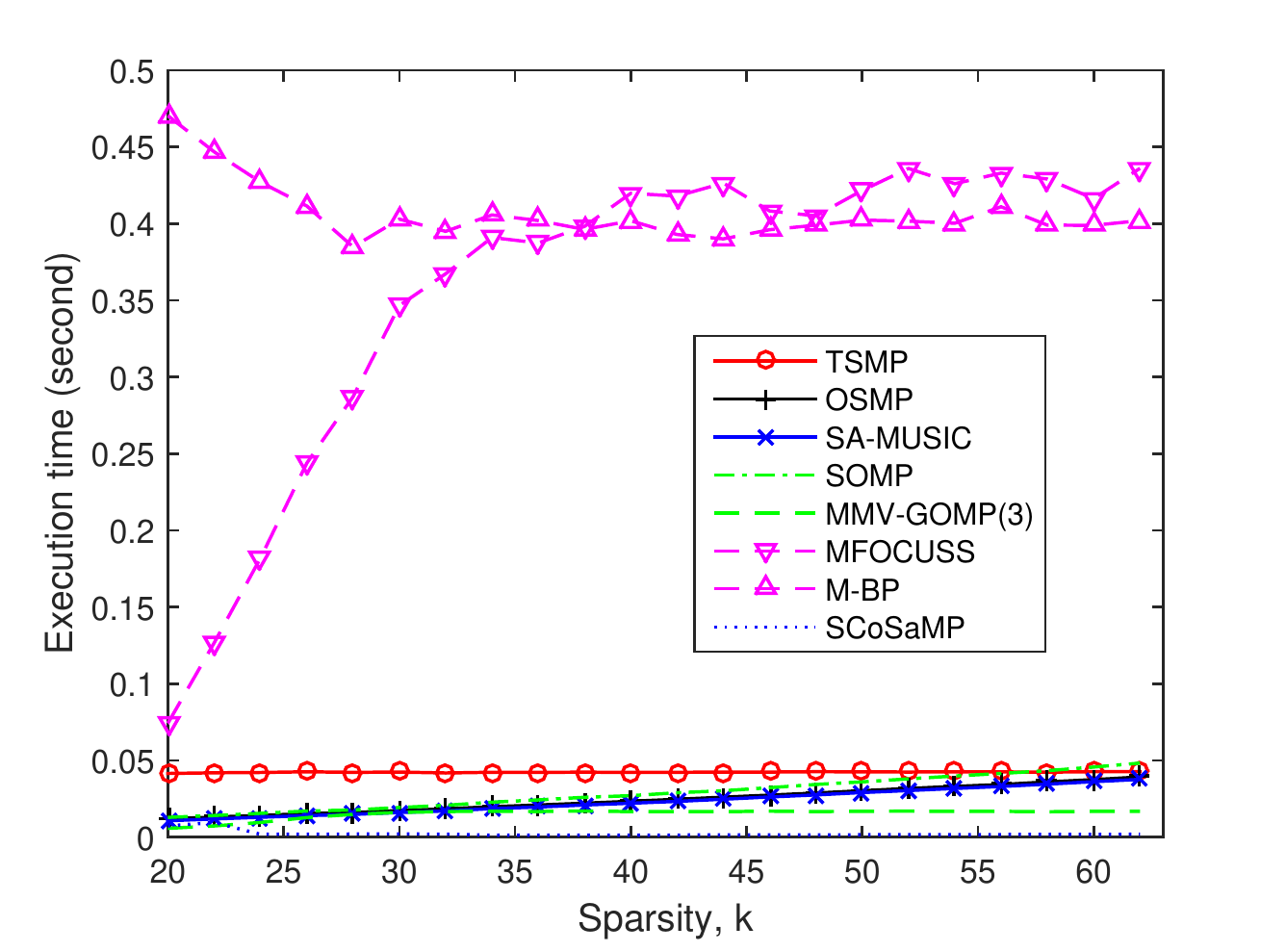}}
\caption{Execution time when $k$ varied from 20 to 63 and $(m,n,l,r)=(64,512,3,3)$ with noiseless data}
\end{figure}
Fig. 1(a) and 1(b) compare the performance of algorithms in terms of successful support recovery rates with varying number of sparsity $k$ and fixed triplet $(m,n,l)=(64,512,3)$ for the noiseless and noisy cases (i.e., SNR = $40$ dB), respectively. In order to estimate the support with algorithms that only provide the estimated signal matrix, a subroutine identifying the index set of the rows with the largest row $l_2$ norms of the estimated signal matrix is additionally implemented. Our simulation results show that TSMP exhibits the best recovery performance. Another fact worth noticing is that OSMP outperforms SA-MUSIC+OSMP and this tendency continues in most cases when $r=l<k$. This looks contradictory since according to the empirical result in \cite{lee2012subspace}, SA-MUSIC+OSMP exhibits better performance than OSMP. Discussion on this relationship will be provided in detail in Section \ref{sec_d}. 

\label{fig4}
\begin{figure}[t]
  \centering
  \footnotesize
  \subfigure[\scriptsize]{\includegraphics[width=9cm, height=6.4cm]{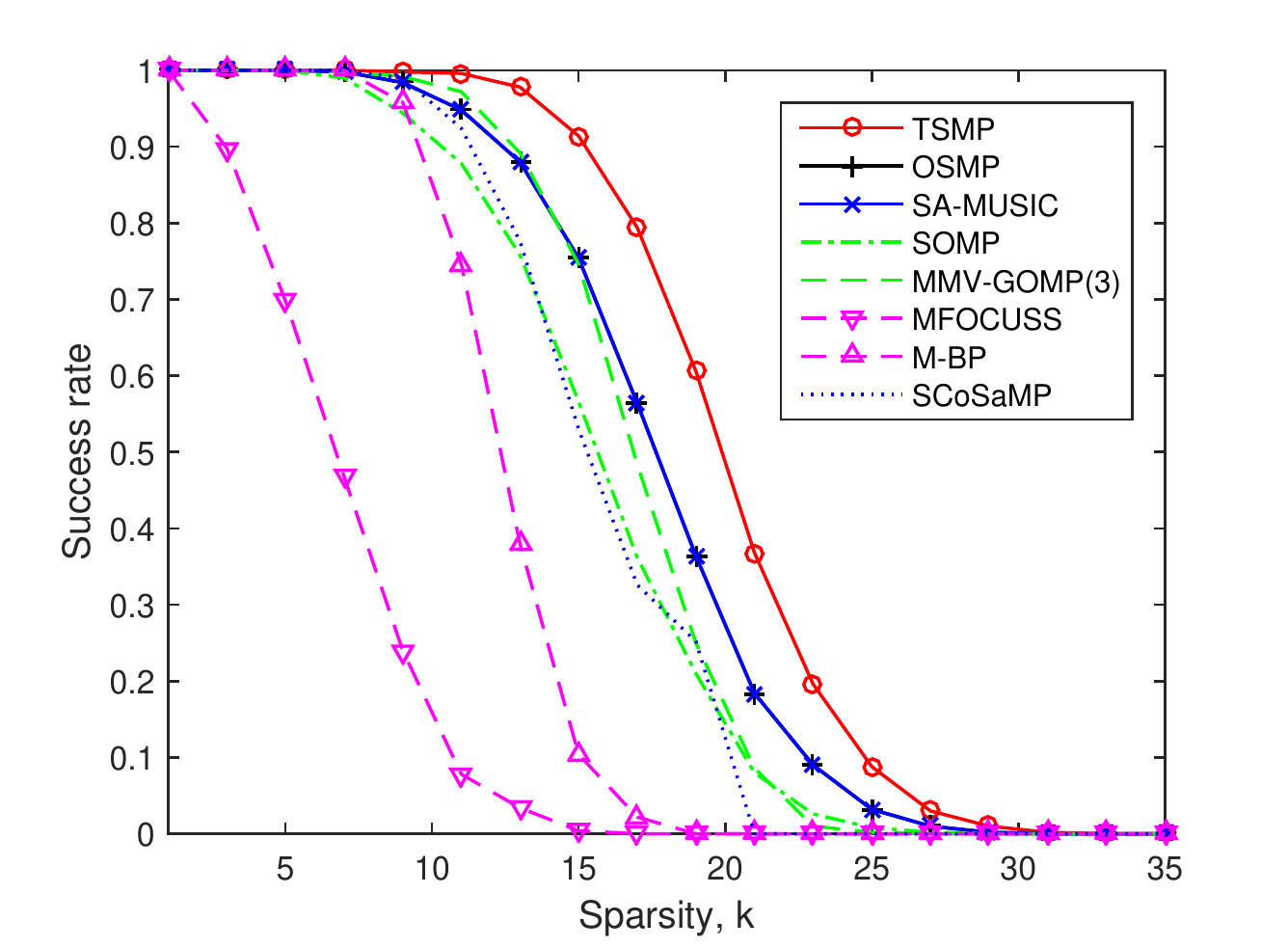}}
  \subfigure[\scriptsize]{\includegraphics[width=9cm, height=6.4cm]{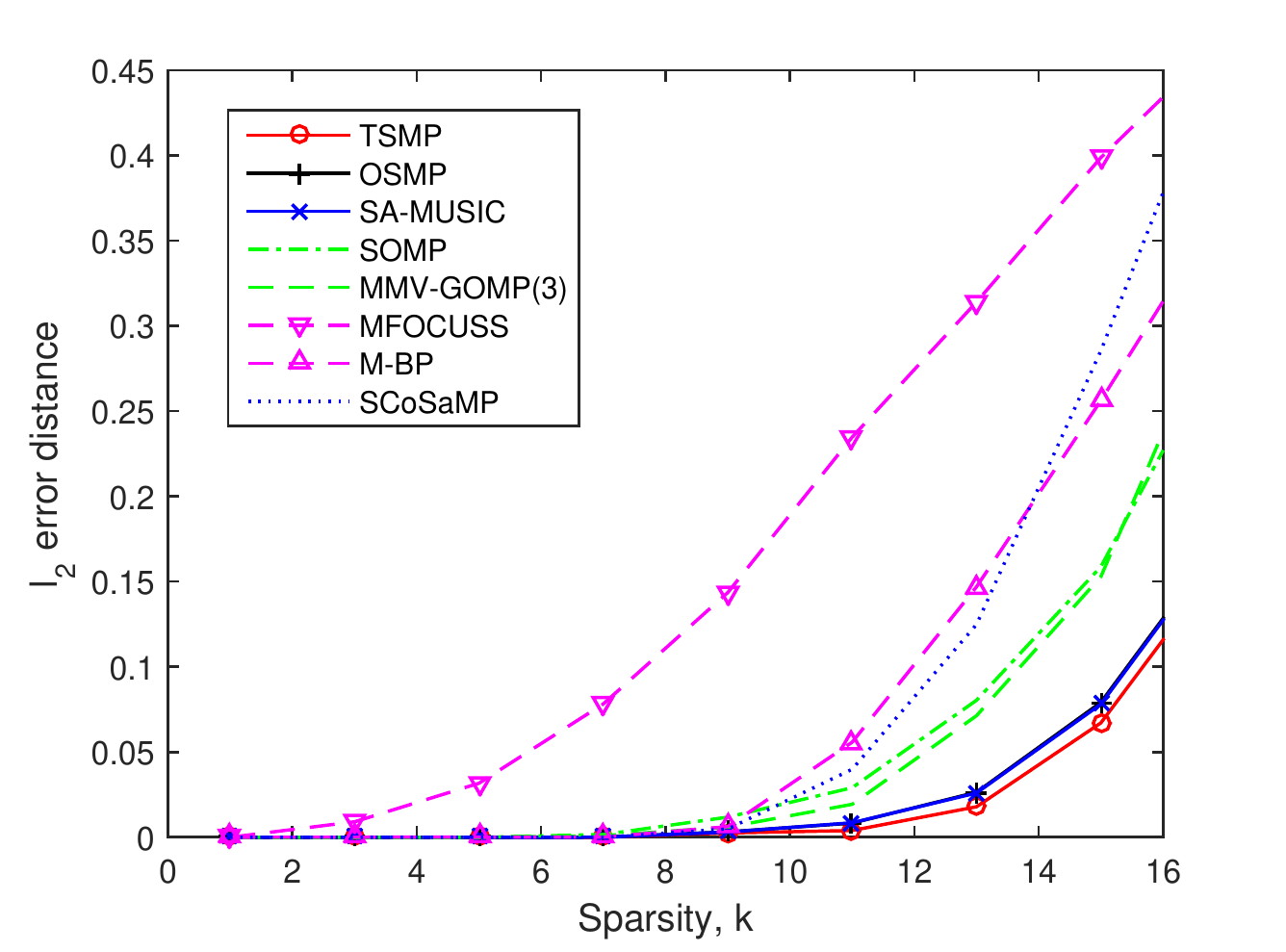}}
    \caption{Performance of various algorithms with noiseless data when $(m,n,l,r)=(64,512,1,1)$ $(a)$ Success rate of true support recovery $(b)$ Empirical mean value of $l_2$ norm distance between $X_0$ and its estimated matrix}
\end{figure}

Fig. 2 corresponds to the same scenario as in Fig. 1, but uses a different metric, i.e., the $l_2$ distance between $X_0$ and the estimated $X_0$ of each algorithm.\footnote{In order to estimate the signal matrix with algorithms that only provide the estimated support ($\hat \Omega$), an additional subroutine is implemented to provide an estimated signal matrix by calculating the inverse of $\Phi_{\hat \Omega}$ (i.e., $\hat X = \Phi^{-1}_{\hat \Omega}Y$).} TSMP still outperformed the other algorithms.  

Fig. 3 illustrates the execution time of each algorithm using the same parameters as in Fig. 1(a). While the execution time of TSMP remained almost unchanged, that of other greedy algorithms such as OSMP or SA-MUSIC+OSMP increased as the sparsity level approaches $m$. This is favorable to TSMP since the main focus of compressive sensing is in the case when the sparsity level is relatively large and close to $m$. Most greedy algorithms including TSMP have relatively fast running times than optimization-based schemes such as M-BP and its faster version, MFOCUSS.  

Though only the performance of the above algorithms are compared by using fixed values of $(m,n,l,r)=(64,512,3,3)$, TSMP most likely exhibited better performance for the recovery of the true support and signal matrix than the existing algorithms in most of the parameter space if SNR exceeded a certain level. A comparison in performance in the conventional SMV case is shown as an example in Fig. 4 by fixing $(m,n,l,r,$SNR$)=(64,512,1,1,\infty)$. TSMP still exhibited the best performance for the recovery of $\Omega$ and $X_0$.

\label{fig5}
\begin{figure}[h]
  \centering
  \footnotesize
  \subfigure{\includegraphics[width=9cm, height=6.4cm]{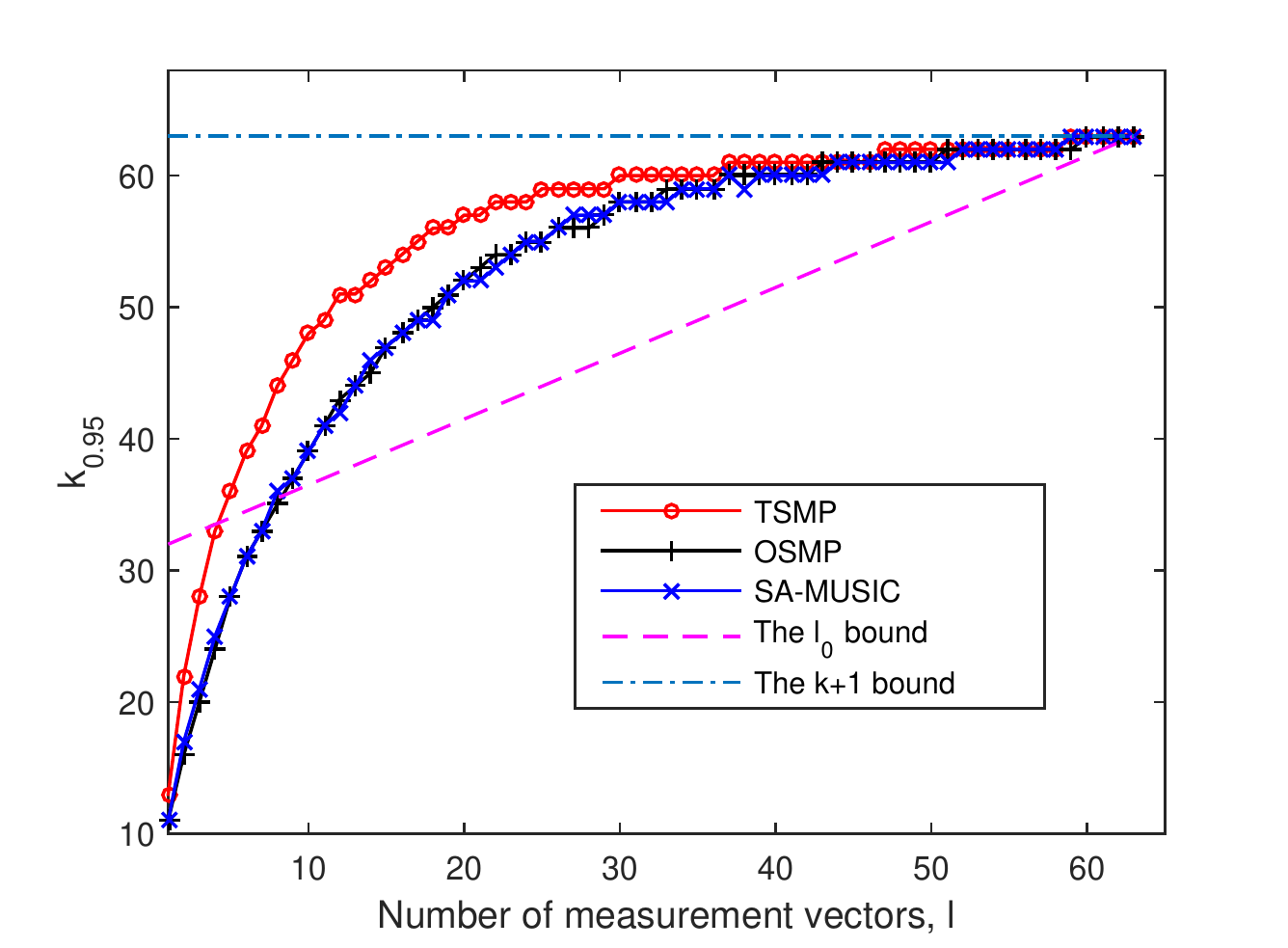}}
\caption{$k_{0.95}$ of TSMP/OSMP/SA-MUSIC versus $l$ (from $l=1$ to $l=63$) in the noiseless and $(m,n)=(64,128)$ case}
\end{figure}

Fig. 5 shows a plot of $k_{0.95}$ (i.e., the maximum sparsity $k$ such that the probability TSMP identifies $\Omega$ is larger than $0.95$) versus $l$ $(=r)$ given fixed parameters $(m,n,$SNR$)=(64,128,\infty)$. Fig. 5 shows that the maximal sparsity level for TSMP to recover $\Omega$ in every case surpasses the $l_0$ bound (i.e., $(m+l-1)/2$) and converges to the ``$k+1$'' bound as the rank of $X_0$ (i.e., the number of columns of $X_0$) increases. This implies that even a few number of measurement vectors ($l$) is significantly helpful to improve the performance of TSMP. This verifies one of the main advantages of the MMV setting over the SMV case.\footnote{Similar arguments were suggested by Tang, Eldar, et al. \cite{eldar2010average}, \cite{tang2010performance}. They theoretically proved that the recovery rate increases exponentially with the number of measurement vectors under certain mild conditions. Though their focus was on different algorithms, these results also support the advantage of the MMV problem.}

\section{Discussion}
\label{sec_d}
\subsection{Comparison between two cases; $l \leq k$ and $l > k$.}
As Zhang, et al. suggested \cite{zhang2011sparse}, the sparsity assumption is valid only for a small $l$ such that $l \leq k$ in most applications regarding the MMV problem (i.e., EEG/MEG source localization, DOA estimation, etc.) since the support profile of practical signal vectors (i.e., columns of $X_0$) is time-varying. While the case $l < k$ always belongs to the rank defective case, the full row rank case $\rank(X_0^{\Omega}) = k$ could easily occur if $l >k$. For instance if each column of $X^{\Omega}_0$ does not belong to a certain set with Lebesgue measure zero (i.e., range of the other columns), $\rank(X_0^{\Omega}) = k$ is satisfied when $l> k$. This full row rank case is not the focus of this study since the performances of MUSIC-like algorithms such as OSMP, SA-MUSIC+OSMP, TSMP, etc. are the same as MUSIC which has the lowest complexity. Our empirical results thus focus on the case when $l \leq k$.

\subsection{Comparison between OSMP and SA-MUSIC+OSMP}
The only difference between OSMP and SA-MUSIC+OSMP \cite{davies2012rank,lee2012subspace} is the selection rule for the last $r$ indices to estimate the true support. Lee, et al. showed in \cite{lee2012subspace} that in the case of $l>k\geq r$, there exists a region in the parameter space $(m,n,l,r,k,$SNR$)$ such that SA-MUSIC+OSMP outperforms OSMP by setting a common and specific subspace estimator to extract an $r$-dimensional subspace $\hat S$ from $\mathcal{R}(Y)$. According to our empirical results presented in Section \ref{sec_ne}, there exists another region such that OSMP outperforms SA-MUSIC+OSMP. This performance advantage of OSMP was observed in most cases when $l = r \leq k$ and the performance gap increased as $l$ decreased. Small $l$, on the other hand, is preferred in many applications as discussed earlier. Remark \ref{rem1n2_4} shows that the theoretical performance guarantee of OSMP is no worse than that of SA-MUSIC+OSMP.
It is expected that OSMP will provide more practical solutions than SA-MUSIC+OSMP in recovering the sparse signal in MMV problems.

\subsection{Comparison to M-SBL and T-SBL}
M-SBL \cite{wipf2007empirical} (i.e., T-SBL \cite{zhang2011sparse} when $X$ is uncorrelated) is known to be a scheme with theoretical guarantee that achieves the ``$k+1$'' bound just as MUSIC (or SA-MUSIC or OSMP) when $\rank(X_0^{\Omega})=k$ and columns of $X$ are orthogonal. Since the orthogonality condition is more restrictive in M-SBL or T-SBL than MUSIC, fundamental analysis on M-SBL or T-SBL has been limited despite its good empirical performance. Furthermore, M-SBL or T-SBL is likely to be more computationally expensive than other subspace greedy algorithms. 

\subsection{Selecting a method for subspace estimation when $\rank(\Phi_{\Omega}X_0^{\Omega}) < l$}
In the noisy case with $\rank(\Phi_{\Omega}X_0^{\Omega}) < l$, it is common to face the situation where $\rank(Y)>\rank(\Phi_{\Omega}X_0^{\Omega})$ due to random noise. When $\Phi_{\Omega}X_0^{\Omega}$ is ill-conditioned (i.e., the last few singular values of $\Phi_{\Omega}X_0^{\Omega}$ are relatively small), estimating $\hat S$ as the space for $r$ $(<l)$ largest singular vectors of $Y$ can be used to improve the robustness against noise. Based on this principle, Lee, et al. \cite{lee2012subspace} proposed an eigenvalue decomposition-based scheme, SSE($\tau$), to estimate an $r$-dimensional signal subspace $\hat S$ and showed that $\eta$ is arbitrarily bounded with finite $l$ through SSE($\tau$) for a mixed multi-channel model. The subspace estimation scheme, however, is not restricted to the specific method. The selection of a good estimator to reduce noise depends on the conditions of each individual case. For example, the robust principal component analysis \cite{candes2011robust} will provide a better estimate than the usual singular value decomposition in case of sparse noise \cite{lee2012subspace}.

\section{Conclusion}
Sparsity ($k$) plus one measurements (the ``$k+1$'' bound) are sufficient ideally to recover almost all sparse signals irrespective of the number of measurement vectors $l$. To better approach the ``$k+1$'' bound with low computational complexity, an improved scheme of the OSMP called TSMP was proposed as a greedy subspace method for joint sparse recovery which provides both better empirical performance and less restrictive theoretical guarantees approaching the bound than most existing algorithms. The empirical results showed that the minimum $m$ required for the uniform recovery of TSMP decreases below the $l_0$ bound as $l$ increases and more rapidly approaches the ``$k+1$'' bound than most existing algorithms with a small $l$. Furthermore, performance guarantees for OSMP and TSMP were derived with regard to the sensing matrix properties such as the weaker version of RIP or the new variant of mutual coherence to improve results.
The theoretical results are non-asymptotic for $(m,n,l,k,\rank(X_0))$, valid for the noisy case, and applicable to a widely used class of sensing matrices for real applications.
 
Though the proposed greedy algorithm with low computational complexity outperformed most of existing greedy methods, there might be a new algorithm beyond the scheme since $k+1$ measurements are ideally sufficient even for the SMV case (i.e., $l=1$). Therefore, the new algorithm could guarantee the success of joint sparse recovery even when $(m,l)$ are jointly much closer to $(k+1,1)$. The case where $(m,l)=(k+1,1)$ is the fundamental limit beyond the conventional bottleneck of SMV, $(2k,1)$. This direction of research might provide new understandings on joint sparse recovery when $l$ is small. 
\appendices

\section{Proof of Theorem \ref{propmc_st3}}

Define  
\begin{align}\nonumber
g(y)&:= 4y(1-y)\\\nonumber
s_1&:=\min\limits_{\Gamma \in t(k-v_2,v_1)} \, \underset{i \in \Sigma \setminus (\Omega \cup \Gamma) }{\min}\sigma^2_{|\Omega \setminus \Gamma|+1} (\dot\Phi_{\Omega \cup \{i\} \setminus \Gamma } )\\\nonumber
s_2&:=\frac{\overset{c_1}{\overbrace{\min\limits_{\Delta_{(k+v_1+1)} \subseteq \Sigma}\,\sigma^2_{k+v_1+1}(\Phi_{\Delta_{(k+v_1+1)} })}}}{\underset{c_2}{\underbrace{\max\limits_{\Delta_{(k+v_1+1)} \subseteq \Sigma}\,\sigma^2_{1}(\Phi_{\Delta_{(k+v_1+1)}  })}}} \\\nonumber
s_3&:=\frac{\min\limits_{\Gamma \in t(k-v_2,v_1)} \,\,\underset{i \in \Sigma \setminus {\Omega} \cup \Gamma}{\min} \,\sigma^2_{|{\Omega} \cup \Gamma \cup \{i\}|}(\Phi_{{\Omega} \cup \Gamma \cup \{i\}  })}{ \underset{l \in {\Omega} \cup \{i\} \setminus \Gamma }{\max} \left \| P^{\perp}_{\mathcal{R}(\Phi_{\Gamma})}\phi_l \right \|^2_2}\\\nonumber
s_4&:=\frac{\min\limits_{\Gamma \in t(k-v_2,v_1)} \,\, \underset{i \in \Sigma \setminus {\Omega} \cup \Gamma}{\min} \,\sigma^2_{|({\Omega} \cup \{i\}) \setminus \Gamma|}(P^{\perp}_{\mathcal{R}(\Phi_{\Gamma})}\Phi_{({\Omega} \cup \{i\}) \setminus \Gamma})}{ \underset{l \in ({\Omega} \cup \{i\}) \setminus \Gamma }{\max} \left \| P^{\perp}_{\mathcal{R}(\Phi_{\Gamma})}\phi_l \right \|^2_2}.
\end{align}
The proof is similar to the proof of Theorem \ref{new3prp_3}. $\{h_1,...,h_7\}$, $f_1(\eta,k,r)$, $s_1(x_1,...,x_5)$, and (\ref{thm_g3_given3_cor}) of Corollary \ref{thm_g3_3_cor} in the proof of Theorem \ref{new3prp_3} are substituted by $\{s_1,s_1,a_1(v_1),s_2,a_3(v_1),s_3,s_4\}$, $g(\eta)$, $q(x_1,...,x_5):=x_1 \cdot x_2 - g(x_3)$, and (\ref{mu_sup_c1}) of Corollary \ref{cor_music_c}, respectively. The fact that $s_2 \leq a_2(v_1) \leq a_3(v_1)$ is used. 

\section{Proof of Theorem \ref{second_mainthm1}}
For $i \in \{0,...,k-1\}$, let $A_i(\Gamma_i)$ denot an event where submp($\hat S,\Gamma_i,1$) produces an arbitrary index $a_i$ in $\Omega \setminus \Gamma_i$ where $| \Gamma_0|=0$ and $\Gamma_i=\{a_i\} \cup \Gamma_{i-1}$.   
Suppose that 
\begin{align}\label{meacond1}
m > k+\max \left\{4, \frac{1}{{\lambda^{-1}(z)}}\right\}\ln \, (n^c).
\end{align}
Then it follows that
\begin{align}\nonumber
\mathbb{P}_s &= \mathbb{P}(\cap_{i=0}^{i=k-1}A_{i}(\Gamma_{i}))\\\nonumber
&\geq 1-\sum_{i=0}^{k-1}\mathbb{P}(A^c_i(\Gamma_{i}))\\\nonumber
&\overset{(a)}\geq 1-\sum_{i=0}^{k-1}\mathbb{P}(\max\limits_{j \in \Sigma \setminus ( \Omega \cup \Gamma_i)}\mu(\Omega \cup \{j\},\Gamma_i) \geq z| \Phi_{\Gamma_i})\\\nonumber
&\overset{(b)} \geq 1- 4k^2 \cdot n^{1-c},
\end{align}
where $(a)$ follows from Theorem \ref{new3prp_3_sp}, Theorem \ref{propmc_st3_sp}, and condition $\rho(\hat S)\leq \eta \leq 0.5$, $(b)$ follows from (\ref{meacond1}) and (\ref{cor_eachiter_low1_2}) in Lemma \ref{cor_eachiter_low1}. 
Setting $c$ as $\frac{- \ln (4k^2 n /\epsilon)}{\ln \, n}$ completes the proof.

\section{Proof of Theorem \ref{paeosmp}}
Define $\mathbb{P}_{s}(x)$ as the probability that submp($\hat S,\o,x$) produces a set $\Gamma$ of $x$ indices such that $\Gamma \supseteq \Omega$. 
For $i \in \{1,...,g\}$, $\mathbb{P}_{e}(g;i|\Phi_{\Omega})$ denotes the conditional probability for a given $\Phi_{\Omega}$ such that submp($\hat S,\o,g$) produces a set $\Gamma$ of $g$ indices where $|\Gamma \setminus \Omega| \geq i$ (i.e., $\Gamma$ includes at least $i$ false indices), and ${E}_i$ denotes an event where submp($\hat S,\Gamma,1$) such that $|\Gamma|=i-1$ produces an index in $\Sigma \setminus (\Omega \cup \Gamma)$. 
The following events for $\Gamma \subseteq \Sigma$, $v \in \Sigma$, and a constant $x$ are defined:
\begin{align}\nonumber
&A(\Omega,\Gamma;x):=\{\max\limits_{i \in \Sigma \setminus ( \Omega \cup \Gamma)}\mu(\Omega \cup \{i\},\Gamma) \geq x \}\\\nonumber
&B_{\{v\}}(\Omega,\Gamma;x):=\{\mu(\Omega \cup {\{v\}},\Gamma) \geq x \}
\end{align}
If $\rho(\hat S)\leq \eta \leq 0.5$ and $m \geq k+g+\max \left\{4, \frac{1}{{\lambda^{-1}(z)}}\right\}\ln \, (n^c)$ hold for $g \in \mathbb{N}$ $(>2)$ and a constant $c \in \mathbb{R}$, it follows that
\begin{align}\nonumber
& \mathbb{P}((E_j)_{j \in \Delta},  \Delta \subseteq [g], |\Delta|=2| \Phi_{\Omega})\\\nonumber
&\overset{(a)} \leq \sum\limits_{\underset{\textup{ s.t. $i < j$}}{\{i,j\} \subseteq \{1,...,g\}}} \sum_{v \in \Sigma \setminus ( \Omega \cup {\bar\Gamma})}[\mathbb{P}(B_{\{v\}}(\Omega,{\bar\Gamma};z)|\Phi_{\Omega},|{\bar\Gamma}|=i) \cdot\mathbb{P}(A(\Omega,{\hat\Gamma};z)|\Phi_{\Omega}, |{\hat\Gamma}|=j,{\hat\Gamma}  \supseteq {\bar\Gamma} \cup \{v\}, B_{\{v\}}(\Omega,{\bar\Gamma};z))] \\\nonumber
&\overset{(b)} \leq \sum\limits_{\underset{\textup{ s.t. $i < j$}}{\{i,j\} \subseteq \{1,...,g\}}} \sum_{v \in \Sigma \setminus ( \Omega \cup {\bar\Gamma})}[\mathbb{P}(B_{\{v\}}(\Omega,{\bar\Gamma};z)|\Phi_{\Omega},|{\bar\Gamma}|=i)] \cdot(4k \cdot n^{1-c})]\\\nonumber
&\overset{(c)} \leq \sum\limits_{\underset{\textup{ s.t. $i < j$}}{\{i,j\} \subseteq \{1,...,g\}}} \sum_{\,v \in \Sigma \setminus ( \Omega \cup {\bar\Gamma})}(4k \cdot n^{1-c})\cdot(4k \cdot n^{-c})\\\label{examplenum2}
&\leq {g\choose 2} (4k \cdot n^{1-c})^2,
\end{align}
where $(a)$ follows from Theorem \ref{new3prp_3_sp}, Theorem \ref{propmc_st3_sp}, and $\rho(\hat S)\leq \eta \leq 0.5$, $(b)$ follows from Lemma \ref{cor_eachiter_low1_2new1} that shows $\mathbb{P}(A(\Omega,{\hat\Gamma};z)| \Phi_{\Omega \cup {\hat\Gamma}}) \leq 4k \cdot n^{1-c}$ for a given $\Phi_{\Omega \cup {\hat\Gamma}}$, and $(c)$ follows from Lemma \ref{cor_eachiter_low1_2new1} that shows $\mathbb{P}(B_{\{v\}}(\Omega,{\bar\Gamma};z)| \Phi_{\Omega \cup {\bar\Gamma}}) \leq 4k \cdot n^{-c}$ for a given $\Phi_{\Omega \cup {\bar\Gamma}}$. 
Similarly with (\ref{examplenum2}), if $\rho(\hat S)\leq \eta \leq 0.5$ and the following condition holds for a constant $c \in \mathbb{R}$
\begin{align}\label{measurecond2}
m \geq k+t+\max \left\{4, \frac{1}{{\lambda^{-1}(z)}}\right\}\ln \, (n^c),
\end{align}
then it follows that for $t \in \mathbb{N}$ $(>k)$ and $i \in \{t-k+1,...,t\}$,
\begin{align}\label{examplenum2_t1}
& \mathbb{P}((E_j)_{j \in \Delta}, \Delta \subseteq [t], |\Delta|=i| \Phi_{\Omega})\leq {t\choose i} (4k \cdot n^{1-c})^i.
\end{align}
For $t\in \mathbb{N}$ such that $k<t\leq m-1$, the following holds
\begin{align}\nonumber
1-\mathbb{P}_{s}&=1-\mathbb{P}_{s}(m-1) \\\nonumber
&\leq 1-\mathbb{P}_{s}(t) \\\label{defpe}
&\leq  \int_{\Phi_{\Omega},\Omega}\mathbb{P}_{e}(t;t-k+1|\Phi_{\Omega})\mathbb{P}({\Phi_{\Omega},\Omega})\,d(\Phi_{\Omega},\Omega).
\end{align}
If $\rho(\hat S)\leq \eta \leq 0.5$ and (\ref{measurecond2}) are satisfied, then 
\begin{align}\nonumber
&\mathbb{P}_{e}(t;t-k+1|\Phi_{\Omega})\\\nonumber
& =\sum_{i=t-k+1}^{t}\mathbb{P}((E_j)_{j \in \Delta},(E^c_u)_{u \in [t] \setminus \Delta}, \Delta \subseteq [t],|\Delta|=i|\Phi_{\Omega})\\\nonumber
& \leq \sum_{i=t-k+1}^{t}\mathbb{P}((E_j)_{j \in \Delta}, \Delta \subseteq [t], |\Delta|=i|\Phi_{\Omega})\\\label{examplenum2tt1}
&\overset{(a)} \leq\sum_{i=t-k+1}^{t} {t\choose i} (4k \cdot n^{1-c})^i,
\end{align}
where $(a)$ follows from (\ref{examplenum2_t1}) and $\rho(\hat S)\leq \eta \leq 0.5$. 
Therefore, the following equality is obtained by applying (\ref{examplenum2tt1}) to (\ref{defpe}).
\begin{align}\label{examplenum2tt222}
\mathbb{P}_{s} \geq 1-\sum_{i=t-k+1}^{t} {t\choose i} (4k \cdot n^{1-c})^i
\end{align}
Setting $c$ as $\frac{- \ln (4k n /\epsilon)}{\ln \, n}$ in (\ref{examplenum2tt222}) yields 
\begin{align}\nonumber
\mathbb{P}_{s} \geq 1- \sum_{i=t-k+1}^{t} {t\choose i} \epsilon^i.
\end{align}

\section{Proof of Theorem \ref{sceosmp1}}
It suffices to show that the following two statements hold.
\begin{itemize}
\item Suppose that $\Omega \subseteq {\Omega_c}$. Then ESMS$_1$(${\Omega_c},k$) yields $\Omega$ as its output $Q$ if (\ref{sse_eval_eq00}) holds.
\item Suppose that $\Omega \subseteq {\Omega_c}$. Then ESMS$_2$(${\Omega_c},\kappa$) yields $\Omega$ as its output $Q$ if (\ref{sse_eval_eq001}) holds.
\end{itemize}

First, it will be proven that (\ref{sse_eval_eq00}) is a sufficient condtion for ESMS$_1$(${\Omega_c},k$) to recover $\Omega$. Without loss of generality, we assume that $\sigma_{m}(\Phi_{\Omega_c})>0$.
Since 
\begin{align}\nonumber
\frac{\left \| W^* \right \|_{2,\infty}}{\sigma_{m}(\Phi_{\Omega_c})}&=\left \| \Phi^{\dagger}_{\Omega_c}\right \|_2 \cdot \max\limits_{c \in [m]}\left \| W^{\{c\}}\right \|_2 \\\label{sse_eval_a1}
&\geq \max\limits_{c \in [m]}\left \| (\Phi^{\dagger}_{\Omega_c}W)^{\{c\}}\right \|_2,
\end{align}
(\ref{sse_eval_eq00}) implies that
\begin{align}\label{sse_eval_p6}
\min\limits_{a \in \Omega}   \left \| {X_0}^{\{a\}} \right \|_2 > 2 \max\limits_{c \in [m]}\left \| (\Phi^{\dagger}_{\Omega_c}W)^{\{c\}}\right \|_2.
\end{align}
From the triangle inequality, we get
\begin{align}\label{sse_eval_p6_11}
&\min\limits_{a \in \Omega}   \left \| {X_0}^{\{a\}} \right \|_2-\max\limits_{b \in [m]}\left \| (\Phi^{\dagger}_{\Omega_c}W)^{\{b\}}\right \|_2\leq \min\limits_{a \in \Omega}   \left \| ({ X_0}+ \Phi^{\dagger}_{\Omega_c}W)^{\{a\}}\right \|_2. 
\end{align}
Then (\ref{sse_eval_p6}) guarantees 
\begin{align}\label{sse_eval_p62}
\min\limits_{a \in \Omega}   \left \| ({ X_0}+ \Phi^{\dagger}_{\Omega_c}W)^{\{a\}}\right \|_2 > \max\limits_{b \in [m]}\left \| (\Phi^{\dagger}_{\Omega_c}W)^{\{b\}}\right \|_2.
\end{align}
Finally, from Lemma \ref{sceosmp0} with $\bar Y=\Phi_{\Omega}X^{\Omega}_0$ and $\sigma_{|{\Omega_c}|}(\Phi_{\Omega_c})>0$, (\ref{sse_eval_p62}) implies  
\begin{align}\label{sse_eval_p7temp}
\min\limits_{a \in \Omega}   \left \| {(\Phi^{\dagger}_{\Omega_c} Y)}^{\{a\}} \right \|_2 > \max\limits_{b \notin \Omega}   \left \| {(\Phi^{\dagger}_{\Omega_c} Y)}^{\{b\}} \right \|_2.
\end{align}
That is, (\ref{sse_eval_eq00}) is a sufficient condition for (\ref{sse_eval_p7temp}) which guarantees that ESMS$_{1}$(${\Omega_c},k$) recovers $\Omega$.

Next, it will be proven that (\ref{sse_eval_eq001}) is a sufficient condtion for ESMS$_2$(${\Omega_c},\kappa$) to recover $\Omega$. 
From (\ref{sse_eval_a1}), (\ref{sse_eval_p6_11}), and (\ref{sse_eval_p62}), (\ref{sse_eval_eq001}) implies
\begin{align}\label{sse_eval_p9}
\min\limits_{a \in \Omega}   \left \| ({ X_0}+ \Phi^{\dagger}_{\Omega_c}W)^{\{a\}}\right \|_2 \geq \kappa > \max\limits_{b \in [m]}\left \| (\Phi^{\dagger}_{\Omega_c}W)^{\{b\}}\right \|_2.
\end{align}
From Lemma \ref{sceosmp0} with $\bar Y=\Phi_{\Omega}X_0$, (\ref{sse_eval_p9}) implies 
\begin{align}\label{sse_eval_p7}
\min\limits_{a \in \Omega}   \left \| {(\Phi^{\dagger}_{\Omega_c} Y)}^{\{a\}} \right \|_2 \geq \kappa > \max\limits_{b \notin \Omega}   \left \| {(\Phi^{\dagger}_{\Omega_c} Y)}^{\{b\}} \right \|_2.  
\end{align}
That is, (\ref{sse_eval_eq001}) is a sufficient condition for (\ref{sse_eval_p7}), which guarantees that ESMS$_{2}$(${\Omega_c},\kappa$) recovers $\Omega$. 

\section{Proof of Corollary \ref{sceosmp2}}
From Lemma \ref{thm_lot}, it follows that for any $a \in [m]$ and $t>0$, 
\begin{align}\label{gauseval1}
&\mathbb{P} \left(\left \| W^{\{a\}}\right \|_2 \geq \sigma(\sqrt{m}+1+t)\right) \leq \exp \left(-\frac{t^2}{2}\right).
\end{align}
By applying the union bound to (\ref{gauseval1}), it follows that
\begin{align}\label{gauseval2}
&\mathbb{P} \left(\max\limits_{a \in [m]} \left \| W^{\{a\}}\right \|_2< \sigma(\sqrt{m}+1+t)\right) > 1-m \cdot \exp \left(-\frac{t^2}{2}\right).
\end{align}
Set $t$ as $c(\kappa,\sigma,X_0,\Omega_c)$ defined as
\begin{align}\label{gauseval3}
c(s,\sigma,X,J):=\frac{s \cdot \sigma_{|J|}(\Phi_J)}{2\sigma}-\sqrt{|J|}-1.
\end{align}
Then 
\begin{align}\label{gauseval4}
&\mathbb{P} \left(\frac{\left \| W^* \right \|_{2,\infty}}{\sigma_{m}(\Phi_{\Omega_c})} < \kappa \right) > 1-m \cdot \exp \left(-\frac{c(\kappa,\sigma,X_0,\Omega_c)^2}{2} \right).
\end{align}
Applying (\ref{sse_eval_eq_cor1}) and (\ref{gauseval4}) to  (\ref{sse_eval_eq00}) or (\ref{sse_eval_eq001}) shows that a probability satisfying (\ref{sse_eval_eq00}) or (\ref{sse_eval_eq001}) is more than $ 1-m \cdot \exp(-\frac{c(\kappa,\sigma,X_0,\Omega_c)^2}{2})$ so that the proof is completed by Theorem \ref{sceosmp1}.

\section{Performance guarantee} 

\subsection{Performance analysis with LCP}

\begin{thm}\label{new3prp_3_sp}
Let $\eta$ be a constant such that $\rho(\hat S)\leq \eta \leq 0.5$ with an $r$-dimensional space $\hat S$. Suppose that $X_0^{\Omega}$ is row-nondegenerate. 
Then, given $\Gamma$ such that $|\Omega \cap \Gamma| \leq k-r$, submp($\hat S,\Gamma,1$) produces an index in $\Omega \setminus \Gamma$ if $\sigma_{|\Omega \cup \Gamma|}(\Phi_{\Omega \cup \Gamma})>0$ and 
\begin{align}\label{ex1_new3prp_3_1_sp}
\max\limits_{i \in \Sigma \setminus ( \Omega \cup \Gamma)}\mu(\Omega \cup \{i\},\Gamma)   <\frac{1-f(\eta,k,r)}{k},
\end{align}
where
\begin{align}\nonumber
f(\eta,k,r)&\,:=\min\{f_1(\eta,k,r),f_2(\eta,k,r)\}\\\nonumber
f_1(\eta,k,r)&\,:=\left[(\frac{k}{k+r})(2\eta \sqrt{\frac{r}{k}}+\sqrt{\frac{k+r}{k}-4\eta^2})\right]^2  \\\nonumber
f_2(\eta,k,r)&\,:=\frac{1}{[\sqrt{{\frac{k}{r}}(\eta^2)+2}-\sqrt{\frac{k}{r}}\eta]^2}.
\end{align}
\end{thm}

\begin{proof}[Proof of Theorem \ref{new3prp_3_sp}]
It suffices to show that any of the following two conditions with $\sigma_{|\Omega \cup \Gamma|}(\Phi_{\Omega \cup \Gamma})>0$ is sufficient to ensure that submp($\hat S,\Gamma,1$) produces an index in $\Omega \setminus \Gamma$.
\begin{align}\label{ex1_new3prp_3_1_sp}
\max\limits_{i \in \Sigma \setminus ( \Omega \cup \Gamma)}\mu(\Omega \cup \{i\},\Gamma)   &<\frac{1-f_1(\eta,k,r)}{k}\\\label{thm_g2_giveno}
\max\limits_{i \in \Sigma \setminus ( \Omega \cup \Gamma)}\mu(\Omega \cup \{i\},\Gamma)&<\frac{1-f_2(\eta,k,r)}{k}
\end{align}

First, it will be shown that (\ref{ex1_new3prp_3_1_sp}) is a sufficient condtion that submp($\hat S,\Gamma,1$) produces an index in $\Omega \setminus \Gamma$. We assume that 
\begin{align}\label{g2_asp2_new1_1_1}
\sigma_{|\Omega \cup \Gamma|}(\Phi_{\Omega \cup \Gamma})>0.
\end{align}
Let $f_1(\eta,k,r)$ denote the unique solution $x$ of $s_{1}(\sqrt{x},\sqrt{x},\eta,k,r)=0$, where
\begin{align}\nonumber
&s_1(x_1,...,x_5):=\sqrt{\frac{x_5}{x_4}}x_1 - \sqrt{1-x_2^2}-2x_3.
\end{align}
The following parameters are defined:
\begin{align}\nonumber
&h_1:=\max\limits_{i \in \Sigma \setminus ( \Omega \cup \Gamma)}\mu(\Omega \cup \{i\},\Gamma) \\\nonumber
&h_2:=\underset{i \in \Sigma \setminus ({\Omega} \cup \Gamma)}{\max}\, \underset{a \in S:=(\{i\} \cup {\Omega} ) \setminus \Gamma}{\max}\,  \sum\limits_{b \in S, a \neq b}|\left \langle \dot\phi_{a},\dot\phi_{b}\right \rangle|\\\nonumber
&h_3:= \underset{i \in \Sigma \setminus ({\Omega} \cup \Gamma)}{\min}\,\sigma^2_{|\Omega \setminus \Gamma|+1}(\dot\Phi_{(\{i\} \cup {\Omega} ) \setminus \Gamma}) \\\nonumber
&h_4:=\sigma^2_{|\Omega \setminus \Gamma|}(\dot\Phi_{{\Omega} \setminus \Gamma})
\end{align}
Since $h_2 \leq k \cdot h_1$,
\begin{align}\nonumber
h_1<\frac{1-f_1(\eta,k,r)}{k}
\end{align}
implies 
\begin{align}\label{g3pn_2}
h_2<1-f_1(\eta,k,r).
\end{align}
From Lemma \ref{gdt}, (\ref{g3pn_2}) implies
\begin{align}\label{g3pn_2_1}
h_3>f_1(\eta,k,r).
\end{align}
Since $s_{1}(x_1,...,x_5)$ is monotonically non-decreasing for $(x_1,x_2)$ and $ \min\{h_3,h_4\} = h_3$, (\ref{g3pn_2_1}) implies 
\begin{align}\label{g3pn_4}
s_{1}(\sqrt{h_3},\sqrt{h_4},\eta,k,r)>0.
\end{align}
From (\ref{g2_asp2_new1_1_1}), (\ref{g3pn_4}) and Proposition \ref{thm_g3_3}, it is guaranteed that (\ref{ex1_new3prp_3_1_sp}) is a sufficient condition for submp($\hat S,\Gamma,1$) to recover a true index.

Next, it will be shown that (\ref{thm_g2_giveno}) is a sufficient condtion that submp($\hat S,\Gamma,1$) produces an index in $\Omega \setminus \Gamma$. The condition $\rho(\hat S)\leq \eta$ implies that there exists an $r$-dimensional subspace of $\mathcal{R}(\Phi_{\Omega}X_0^{\Omega})$, denoted by $\bar{S}$, satisfying ${\left \| P_{ \bar S}- P_{ \hat S} \right \|_2}\leq \eta$. Set a constant $d:=\mathcal{R}(P^{\perp}_{\mathcal{R}(\Phi_{\Gamma})} \bar S)$. 
Let $f_2(\eta,k,r)$ denote $1-x^*$ where $x^*$ is the unique solution $x$ of $s_2( x, x,\eta,k,r)=0$, where
\begin{align}\nonumber
s_2(x_1,...,x_5)&:=1-\frac{ x_1}{1-x_2}- \sqrt{\frac{x_4}{x_5}}\frac{2x_3}{\sqrt{1-x_2}}.
\end{align}
Then (\ref{thm_g2_giveno}) implies 
\begin{align}\label{n_g1new_gi4_11}
s_2( \alpha_1, \alpha_1,\eta,k,r)>0,
\end{align}
where 
\begin{align}\nonumber
\alpha_1&:= k \cdot \max\limits_{i \in \Sigma \setminus ( \Omega \cup \Gamma)}\mu(\Omega \cup \{i\},\Gamma),
\end{align}
since $\alpha_1= k \cdot h_1$ and $s_2(x_1,...,x_5)$ in (\ref{n_g1new_gi4_11}) is monotonically non-increasing for $(x_1,x_2,x_4)$. 
Note that $d$ is equal to $r$ from Lemma \ref{rankbound_rownond} with (\ref{g2_asp2_new1_1_1}).
From the above results and conditions $\alpha_1 \geq \alpha_{2}$ and $\alpha_1 \geq \beta_{2}$, (\ref{n_g1new_gi4_11}) implies
\begin{align}\label{g2_gw1}
s_2( \alpha_{2}, \beta_{2},\eta,|\Omega \setminus \Gamma|,d)>0,
\end{align}
where
\begin{align}\nonumber
\alpha_{2}&:=  { \underset{a \in \Sigma-\Omega-\Gamma}{\max} \left(\underset{b \in \Omega \setminus \Gamma}{\sum} \left| \left \langle \dot\phi_{a},\dot\phi_{b} \right \rangle \right| \right)} \\\nonumber
\beta_{2}&:= {\max\limits_{a \in S=\Omega \setminus \Gamma} \left( \sum\limits_{b \in S,  \textup{ s.t. } b \neq a} \left |\left \langle \dot\phi_{a},\dot\phi_{b} \right \rangle \right| \right)}.
\end{align}
(\ref{g2_gw1}) implies 
\begin{align}\label{g2_gw4}
\frac{\alpha_{2}}{1-\beta_{2}}<1- \sqrt{\frac{|\Omega \setminus \Gamma|}{d}}\frac{2\eta}{\sigma_{|\Omega \setminus \Gamma|}(\hat{\Phi}_{\Omega \setminus \Gamma})},
\end{align}
since the following inequality holds by Lemma \ref{gdt} and $\beta_2 \leq 1$
\begin{align}\nonumber
&\sqrt{1-\beta_{2}} \leq \sigma_{|\Omega \setminus \Gamma|}(\dot{\Phi}_{\Omega \setminus \Gamma}).
\end{align}
By Lemma \ref{piub} with $\beta_2 \leq 1$, it follows that  
\begin{align}\label{g2_gw5}
&\underset{i \in \Sigma - \Omega -\Gamma}{\max}\left \| \dot\Phi^{\dagger}_{\Omega \setminus \Gamma}\dot\phi_{i} \right \|_1 \leq \frac{\alpha_{2}}{1-\beta_{2}}.
\end{align}
By applying (\ref{g2_gw5}) to (\ref{g2_gw4}), (\ref{g2_gw4}) implies 
\begin{align}\label{g2_gw6}
&\underset{i \in \Sigma - \Omega -\Gamma}{\max}\left \| \dot\Phi^{\dagger}_{\Omega \setminus \Gamma}\dot\phi_{i} \right \|_1 <1- \sqrt{\frac{|\Omega \setminus \Gamma|}{d}}\frac{2\eta}{\sigma_{|\Omega \setminus \Gamma|}(\dot{\Phi}_{\Omega \setminus \Gamma})}.
\end{align}
Therefore, by Proposition \ref{erc_won1} with (\ref{g2_asp2_new1_1_1}) and (\ref{g2_gw6}), (\ref{thm_g2_giveno}) is guaranteed as another sufficient condition for submp($\hat S,\Gamma,1$) to recover the true index.
\end{proof}

\begin{thm}\label{propmc_st3_sp}
Let $\eta$ be a constant such that $\rho(\hat S)\leq \eta \leq 0.5$ with an $r$-dimensional space $\hat S$. Suppose that $X_0^{\Omega}$ is row-nondegenerate, and $\sigma_{|\Omega\cup \Gamma|}(\Phi_{\Omega\cup \Gamma})>0$. Then, given $\Gamma$ such that $|\Omega \cap \Gamma| \geq k-r$, submp($\hat S,\Gamma,1$) produces an index in $\Omega \setminus \Gamma$ if 
\begin{align}\label{ex_msc1_sp}
\max\limits_{i \in \Sigma \setminus ( \Omega \cup \Gamma)}\mu(\Omega \cup \{i\},\Gamma) &<\frac{1-4\eta(1-\eta)}{|\Omega \setminus \Gamma|}.
\end{align}
\end{thm}

\begin{proof}[Proof of Theorem \ref{propmc_st3_sp}]
Define 
\begin{align}\nonumber
&h_1:= \max\limits_{i \in \Sigma \setminus ( \Omega \cup \Gamma)}\mu(\Omega \cup \{i\},\Gamma)  \\\nonumber
&h_2:=\underset{i \in \Sigma \setminus ({\Omega} \cup \Gamma) }{\max} \, \underset{a \in S:={\Omega} \cup \{i\} \setminus \Gamma }{\max}\, \left (\sum_{b \in S, b\neq a}\left|\left \langle \dot\phi_{a}, \dot\phi_{b} \right \rangle \right| \right).
\end{align}
Since $h_2 \leq |\Omega \setminus \Gamma| \cdot h_1$,
\begin{align}\nonumber
h_1<\frac{1-4\eta(1-\eta)}{|\Omega \setminus \Gamma|}
\end{align}
implies 
\begin{align}\label{propmc_u1_sp}
h_2<1-4\eta(1-\eta).
\end{align}
By Lemma \ref{gdt}, (\ref{propmc_u1_sp}) implies 
\begin{align}\label{propmc_u3_sp}
\underset{i \in \Sigma \setminus ({\Omega} \cup \Gamma) }{\min}\sigma^2_{|\Omega \setminus \Gamma|+1} (\dot\Phi_{{\Omega} \cup \{i\} \setminus \Gamma } )>4\eta(1-\eta).
\end{align}
Then, by Proposition \ref{thm1n2} with (\ref{propmc_u3_sp}), it is guaranteed that (\ref{thm_g2_giveno}) is a sufficient condition that submp($\hat S,\Gamma,1$) produces an index in $\Omega \setminus \Gamma$. 
\end{proof}

\subsection{Performance analysis with singular value}

\begin{thm}\label{new3prp_3}
Let $\eta$ be a constant such that $\rho(\hat S)\leq \eta \leq 0.5$ is satisfied. Suppose that $X_0^{\Omega}$ is row-nondegenerate. 
Then submp($\hat S,\o,k-r$) produces a set of $k-r$ indices in $\Omega$ if $\sigma_{k}(\Phi_{\Omega})>0$ and any of the following conditions are satisfied:
\begin{align}
\label{ex1_new3prp_3_3}
&a_1(0)<\frac{1-f_1(\eta,k,r)}{1+f_1(\eta,k,r)} \\\label{ex1_new3prp_3_2}
&a_2(0)> f_1(\eta,k,r) \\\label{ex1_new3prp_3_2temp1}
&a_3(0)> f_1(\eta,k,r),
\end{align}
where
\begin{align}\nonumber
f_1(\eta,k,r)&\,:= \left[(\frac{k}{k+r})(2\eta \sqrt{\frac{r}{k}}+\sqrt{\frac{k+r}{k}-4\eta^2})\right]^2  \\\nonumber
a_1(x)&:= {\delta_{k}(\Phi_{\Omega};x+1)}\\\nonumber
a_2(x)&:= \underset{\Delta_{x} \subseteq \Sigma \setminus \Omega}{\min} \left[ \frac{\underset{i \in \Sigma \setminus \Omega \cup \Delta_{x}}{\min} \, \sigma^2_{k+x+1}(\Phi_{\Omega \cup \Delta_{x} \cup \{i\}  })}{\underset{j \in \Sigma \setminus \Omega \cup \Delta_{x}}{\max} \,\sigma^2_{1}(\Phi_{\Omega \cup \Delta_{x} \cup \{j\}  })} \right] \\ \nonumber
a_3(x)&:= \frac{\underset{\Delta_{(x+1)} \subseteq \Sigma \setminus \Omega}{\min}\, \sigma^2_{k+x+1}(\Phi_{\Omega \cup \Delta_{x+1}})}{\left \| \phi^{\max}_{\Sigma }\right \|^2_2}.
\end{align}
\end{thm}

\begin{proof}[Proof of Theorem \ref{new3prp_3}]
Let $f_1(\eta,k,r)$ denote the unique solution $x$ of $s_1(\sqrt{x},\sqrt{x},\eta,k,r)=0$, where
\begin{align}\nonumber
&s_1(x_1,...,x_5):=\sqrt{\frac{x_5}{x_4}}x_1 - \sqrt{1-x_2^2}-2x_3.
\end{align}
Define the following parameters.
\begin{align}\nonumber
h_1&:= \min\limits_{\underset{\textup{s.t. }|\Gamma|<k-r}{\Gamma \subseteq \Omega}} \,  \sigma^2_{|\Omega \setminus \Gamma|}(\dot\Phi_{{\Omega} \setminus \Gamma})\\\nonumber
h_2&:=  \min\limits_{\underset{\textup{s.t. }|\Gamma|<k-r}{\Gamma \subseteq \Omega}} \,  \underset{i \in \Sigma \setminus ({\Omega} \cup \Gamma)}{\min}\,\sigma^2_{|\Omega \setminus \Gamma|+1}(\dot\Phi_{(\{i\} \cup {\Omega} ) \setminus \Gamma})\\\nonumber
h_3&:=  a_1(0) \\\nonumber 
h_4&:= a_2(0)=\frac{\overset{c_1}{\overbrace{\underset{i \in \Sigma \setminus \Omega}{\min} \,\sigma^2_{k+1}(\Phi_{\Omega \cup \{i\}  })}}}{\underset{c_2}{\underbrace{\underset{j \in \Sigma \setminus \Omega}{\max} \,\sigma^2_{1}(\Phi_{\Omega \cup \{j\}  })}}}  \\\nonumber 
h_{5}&:= a_3(0) \\\nonumber  
h_{6}&:=\frac{\min\limits_{\underset{\textup{s.t. }|\Gamma|<k-r}{\Gamma \subseteq \Omega}} \,\underset{i \in \Sigma \setminus {\Omega} \cup \Gamma}{\min} \,\sigma^2_{|{\Omega} \cup \Gamma \cup \{i\}|}(\Phi_{{\Omega} \cup \Gamma \cup \{i\}  })}{ \underset{l \in {\Omega} \cup \{i\} \setminus \Gamma }{\max} \left \| P^{\perp}_{\mathcal{R}(\Phi_{\Gamma})}\phi_l \right \|^2_2}\\\nonumber
h_{7}&:=\frac{\min\limits_{\underset{\textup{s.t. }|\Gamma|<k-r}{\Gamma \subseteq \Omega}} \, \underset{i \in \Sigma \setminus {\Omega} \cup \Gamma}{\min} \,\sigma^2_{|({\Omega} \cup \{i\}) \setminus \Gamma|}(P^{\perp}_{\mathcal{R}(\Phi_{\Gamma})}\Phi_{({\Omega} \cup \{i\}) \setminus \Gamma})}{ \underset{l \in ({\Omega} \cup \{i\}) \setminus \Gamma }{\max} \left \| P^{\perp}_{\mathcal{R}(\Phi_{\Gamma})}\phi_l \right \|^2_2} 
\end{align}
Then
\begin{align}\label{g3pn_5}
h_3<\frac{1- f_1(\eta,k,r) }{1+f_1(\eta,k,r)  } 
\end{align}
is equivalent to 
\begin{align}\label{g3pn_6}
\frac{1-h_3}{1+h_3}>f_1(\eta,k,r).
\end{align}
Based on the facts that $c_1 \geq c(1- h_3)$ and $c_2 \leq c(1+ h_3)$ from the definition of WRIP, (\ref{g3pn_6}) implies
\begin{align}\label{g3pn_7}
h_4> f_1(\eta,k,r).
\end{align}
From the definition of the induced norm and the variational characterization of singular values, it follows that (\ref{g3pn_7}) implies 
\begin{align}\label{g3pn_8n}
&h_{5}> f_1(\eta,k,r).
\end{align}
The following inequalities also hold.
\begin{align}\label{g3pn_temp2}
&h_{5} \leq h_{6}\overset{(a)}\leq h_{7} \overset{(b)}\leq h_{2}\overset{(c)}\leq h_{1},
\end{align}
where
(a), (b), and (c) follow from Lemma \ref{cpslb}, Lemma \ref{dnu2}, and the variational characterization of singular values, respectively. 
From (\ref{g3pn_temp2}), (\ref{g3pn_8n}) implies any of the followings, since $s_1(x_1,...,x_5)$ is monotonically non-decreasing for $(x_1,x_2)$:
\begin{align}\label{g3pn_temp4}
&s_1(\sqrt{v_2},\sqrt{v_1},\eta,k,r) >0 \\\label{g3pn_temp6}
&v_1>f_1(\eta,k,r),
\end{align}
where 
\begin{align}\nonumber
&v_2:= h_i \textup{ for any i $\in \{1,2,5,6,7\}$}\\\nonumber
&v_1:= h_i \textup{ for any i $\in \{2,5,6,7\}$}.
\end{align}
Then, any of the conditions in (\ref{g3pn_temp4}) and (\ref{g3pn_temp6}) implies
\begin{align}\label{g3pn_temp7}
&s_1(\sqrt{h_{1}},\sqrt{h_{2}},\eta,k,r) >0.
\end{align}

Since (\ref{g3pn_temp7}) is equal to (\ref{thm_g3_given3_cor}), the following statement $I$ implies (\ref{thm_g3_given3_cor}) of Corollary \ref{thm_g3_3_cor}. 
\begin{itemize}
\item $I$: Any of the condtions in (\ref{g3pn_5})--(\ref{g3pn_8n}), (\ref{g3pn_temp4}), and (\ref{g3pn_temp6}) is satisfied.
\end{itemize}
\end{proof}

\section{Proof of lemmas}

\begin{lem}(\cite[Lemma A.4]{lee2012subspace})\label{A.4}
Suppose that $A \in \mathbb{K}^{s \times r}$ where $r \leq s$ satisfies 
\begin{align}\nonumber
\krank(A^*)=r
\end{align}
and $B \in  \mathbb{K}^{s \times t}$ for $t < r$ spans a $t$-dimensional subspace of $\mathcal{R}(A)$. Then 
\begin{align}\nonumber
\rank(B^*)=t.
\end{align}
\end{lem}

\begin{lem}\label{rankbound_rownond}
Let $X_0 \in \mathbb{K}^{n \times l}$ be row $k$-sparse with $\Omega \subseteq \Sigma$. Let $\bar{S}$ be an $r$-dimensional subspace of ${R}(\Phi_{\Omega}X_0^{\Omega})$. Let $\Gamma$ be a proper subset of $\Sigma$. Suppose that $\sigma_{|\Omega \cup \Gamma|}(\Phi_{\Omega \cup \Gamma})>0$ and the row-nongenerate condition on $X_0^{\Omega}$ holds. Then $\rank(P^{\perp}_{\mathcal{R}(\Phi_{\Gamma})}\bar{S})= \min \{ |\Omega \setminus \Gamma|,r \}$.
\end{lem}
\begin{proof}[Proof of Lemma \ref{rankbound_rownond}]
The proof is based on \cite[Appendix $I$]{lee2012subspace}. For a matrix $A \in \mathbb{K}^{p \times q}$, $\nulli(A)$ denotes the nullity of $A$. Since $P^{\perp}_{\mathcal{R}(\Phi_{\Gamma})}\Phi_{\Omega \setminus \Gamma}$ has full column rank, $\nulli(P^{\perp}_{\mathcal{R}(\Phi_{\Gamma})}\Phi_{\Omega \setminus \Gamma}U^{\Omega \setminus  \Gamma})=\nulli(U^{\Omega \setminus  \Gamma})$. There exists a row $k$-sparse matrix $U \in \mathbb{K}^{n \times r}$ with support $\Omega$ such that $\bar S= \mathcal{R}(\Phi_{\Omega} U^{\Omega})$. Then it follows that $\nulli(U^{\Omega \setminus  \Gamma})=\max \{|\Omega \cap \Gamma|-k+r,0 \}$ due to the row-nongeneracy condition on $X_0^{\Omega}$, Remark 5.3 in \cite{lee2012subspace}, and Lemma \ref{A.4}. Therefore, it is guaranteed that $\rank(P^{\perp}_{\mathcal{R}(\Phi_{\Gamma})}\bar{S})=\rank(P^{\perp}_{\mathcal{R}(\Phi_{\Gamma})}\Phi_{\Omega \setminus \Gamma}U^{\Omega \setminus  \Gamma})=r-\max\{|\Omega \cap \Gamma|-k+r,0\}=\min \{ |\Omega \setminus \Gamma|,r \}$.
\end{proof}

\begin{lem}(\cite[Proposition 5.4]{lee2012subspace})\label{subaug}
Let $X_0 \in \mathbb{K}^{n \times l}$ be row $k$-sparse with $\Omega \subseteq \Sigma$. Let $\Gamma$ be a proper subset of $\Omega$ such that $|\Gamma|=k-r$. Let $\bar{S}$ be an $r$-dimensional subspace of ${R}(\Phi_{\Omega}X_0^{\Omega})$. Suppose that $\sigma_{k}(\Phi_{\Omega})>0$ and the row-nongeneracy condition on $X_0^{\Omega}$ holds.  Then 
\begin{align}\nonumber
\bar S + \mathcal{R}(\Phi_{\Gamma})=\mathcal{R}(\Phi_\Omega).
\end{align}
\end{lem}

\begin{lem}(\cite[Lemma A.2]{lee2012subspace})\label{cpslb}
Let $A \in \mathbb{K}^{m \times n}$ and let $J,Q \subseteq \Sigma$. Then, it follows that, for $i=1,...,|Q \setminus J|$
\begin{align}\nonumber
\sigma_{i}(A_{Q\cup J}) \geq \sigma_{i}(P^{\perp}_{\mathcal{R}(A_{J})}A_{Q \setminus J}) \geq \sigma_{i+|J|}(A_{Q \cup J}).
\end{align}
\end{lem}

\begin{lem}\label{proj_ineq}
Let $A \in \mathbb{K}^{m \times a},B \in \mathbb{K}^{m \times b}$ and $C \in \mathbb{K}^{m \times c}$. 
Then any of the following inequalities hold.
\begin{align}\label{proj_ineq1}
\left \| P_{\mathcal{R}([A,C])}-P_{\mathcal{R}([B,C])}\right \|_2 &\leq \left \| P_{\mathcal{R}(A)}-P_{\mathcal{R}(B)}\right \|_2 
\end{align}
\begin{align}\label{proj_ineq2}
\left \| P_{\mathcal{R}(P^{\perp}_{\mathcal{R}(C)}A)}-P_{\mathcal{R}(P^{\perp}_{\mathcal{R}(C)}B)}\right \|_2 &\leq \left \| P_{\mathcal{R}(A)}-P_{\mathcal{R}(B)}\right \|_2 
\end{align}
\end{lem}

\begin{proof}[Proof of Lemma \ref{proj_ineq}]
The following equality is well known given two subspaces $\mathcal{R}([A,C])$ and $\mathcal{R}([B,C])$.
\begin{align}\label{p_proj_ineq0}
&\left \| P_{\mathcal{R}([A,C])}-P_{\mathcal{R}([B,C])}\right \|_2 = \max\left\{\overset{c_1}{\overbrace{\left \| P^{\perp}_{\mathcal{R}([A,C])}P_{\mathcal{R}([B,C])}\right \|_2}},\overset{c_2}{\overbrace{\left \| P^{\perp}_{\mathcal{R}([B,C])}P_{\mathcal{R}([A,C])}\right \|_2}}\right\} 
\end{align}
First, an upper bound on $c_1$ is defined by
\begin{align}\label{hplz1}
c_1&= \underset{\underset{\left \| x \right \|_2 =1 }{x \in \mathcal{R}([B,C])} }{\sup} \,\, \underset{ y \in \mathcal{R}([A,C])}{\inf}\, \left \| x-y\right \|_2 \\\nonumber
&\overset{(a)}= \underset{\underset{\left \| x \right \|_2 =1 }{x \in \mathcal{R}(B)} }{\sup} \,\, \underset{ y \in \mathcal{R}([A,C])}{\inf}\, \left \| x-y\right \|_2 \\\nonumber
&\overset{(b)} \leq  \underset{\underset{\left \| x \right \|_2 =1 }{x \in \mathcal{R}(B)} }{\sup} \,\, \underset{ y \in \mathcal{R}(A)}{\inf}\, \left \| x-y\right \|_2 \\\label{p_proj_ineq1}
&= \left \| P^{\perp}_{\mathcal{R}(A)}P_{\mathcal{R}(B)}\right \|_2,
\end{align}
where (a) follows from the fact that $x^*\in \mathcal{R}([A,C])^{\perp}$ ($x^*$ is the solution $x$ in (\ref{hplz1})) and $\mathcal{R}(C) \subseteq \mathcal{R}([A,C])$, and (b) follows from the fact that $\mathcal{R}(A) \subseteq \mathcal{R}([A,C])$, respectively. 
Similarly, an upper bound on $c_2$ is defined by
\begin{align}\label{p_proj_ineq2}
&c_2 \leq \left \| P^{\perp}_{\mathcal{R}(B)}P_{\mathcal{R}(A)}\right \|_2.
\end{align}
Since 
\begin{align}\label{p_proj_ineq3}
&\left \| P_{\mathcal{R}(A)}-P_{\mathcal{R}(B)}\right \|_2= \max\{\left \| P^{\perp}_{\mathcal{R}(A)}P_{\mathcal{R}(B)}\right \|_2,\left \| P^{\perp}_{\mathcal{R}(B)}P_{\mathcal{R}(A)}\right \|_2\}, %\\\label{p_proj_ineq3}&
\end{align}
(\ref{proj_ineq1}) is derived by applying (\ref{p_proj_ineq0}), (\ref{p_proj_ineq1}), and (\ref{p_proj_ineq2}) to (\ref{p_proj_ineq3}). 
By the projection update rule, we obtain
\begin{align}\label{p_proj_ineq4}
 P_{\mathcal{R}([V,C])}= P_{\mathcal{R}(P^{\perp}_{\mathcal{R}(C)}V)}+P_{\mathcal{R}(C)},
\end{align}
where $V$ is $A$ or $B$. 
Then we obtain the following inequalities from (\ref{p_proj_ineq4}).
\begin{align}\label{p_proj_ineq5}
\left \| P_{\mathcal{R}([A,C])}-P_{\mathcal{R}([B,C])}\right \|_2&=\left \| P_{\mathcal{R}(P^{\perp}_{\mathcal{R}(C)}A)}-P_{\mathcal{R}(P^{\perp}_{\mathcal{R}(C)}B)}\right \|_2.
\end{align}
Therefore, applying (\ref{p_proj_ineq5}) to (\ref{proj_ineq1}) yields (\ref{proj_ineq2}).
\end{proof}

\begin{lem}\label{dnu2}(Generalization of \cite[Theorem 9]{merikoski2004inequalities})
Let $A \in \mathbb{K}^{m \times n}$ such that $m \geq n$ and $B \in \mathbb{K}^{n \times g}$ such that $g \geq n$. If $1 \leq s \leq i \leq n$ and $1 \leq l \leq n-i+1$, then
\begin{align}
\sigma_{i+l-1}(A)\sigma_{n-l+1}(B) \leq \sigma_{i}(AB) \leq \sigma_{i-s+1}( A)\sigma_{k}(B).
\end{align}
\end{lem}

\begin{proof}[Proof of Lemma \ref{dnu2}]
Let $A = U_1 \Lambda_1 V^*_1$ and $ B = U_2 \Lambda_2 V^*_2$ denote the extended SVD of $A$ and $B$, respectively, where $\Lambda_1, V_1,U_2, \Lambda_2 \in \mathbb{K}^{n \times n}$. Let $\hat A, \hat B \in \mathbb{K}^{n \times n}$ be defined as $V_1 \Lambda_1 V^*_1$ and $U_2 \Lambda_2 U^*_2$, respectively. Then it follows that if $1 \leq i \leq n$,
\begin{align}\nonumber
\sigma_{i}(AB)&=\sigma_{i}(U_1 \Lambda_1 V^*_1 U_2 \Lambda_1 V^*_2)\\\nonumber
&=\sigma_{i}(\Lambda_1 V^*_1 U_2 \Lambda_1)\\\nonumber
&=\sigma_{i}(V_1\Lambda_1 V^*_1 U_2 \Lambda_1 U^*_2)\\\label{dnu2_1}
&=\sigma_{i}(\hat A \hat B).
\end{align}
Similarly, we obtain $\sigma_{i}(A)=\sigma_{i}(\hat A)$ and $\sigma_{i}(B)=\sigma_{i}(\hat B)$. 
Therefore, it follows that if $1 \leq s \leq i \leq n$ and $1 \leq l \leq n-i+1$,
\begin{align}\nonumber
\sigma_{i+l-1}(A)\sigma_{n-l+1}(B)&=\sigma_{i+l-1}(\hat A)\sigma_{n-l+1}(\hat B) \\\nonumber
&\leq \sigma_{i}(\hat A \hat B) \,\,(=\sigma_{i}(AB)) \\\nonumber
&\leq \sigma_{i-s+1}(\hat A)\sigma_{s}(\hat B)\\\nonumber
&=\sigma_{i-s+1}( A)\sigma_{s}(B),
\end{align}
where the equalities follow from (\ref{dnu2_1}) and the inequalities follow from \cite[Theorem 9]{merikoski2004inequalities}.
\end{proof}

\begin{lem}(\cite[Theorem 3.5]{tropp2004greed})
\label{piub}
Let $A:=[a_1,...,a_n] \in \mathbb{K}^{m \times n}$ be a matrix with $l_2$-normalized independent columns and ${J},{Q} \subseteq [n]$ be disjoint subsets such that $\sigma_{|{Q}|}(A_{Q})>0$. Then the following inequality holds if $1>\max\limits_{i \in {Q}}\sum\limits_{j \neq i,j \in {Q}}|\left \langle a_{i},a_{j} \right \rangle|$.
\begin{align}\label{msi2}
\underset{i \in {J}}{\max} \left \| A^{\dagger}_{{Q}} a _i\right \|_1 \leq \frac{ \underset{i \in {J}}{\max} (\underset{j \in {Q}}{\sum}|\left \langle a_{i},a_{j} \right \rangle|)}{1-\max\limits_{i \in {Q}}\sum\limits_{j \neq i,j \in {Q}}|\left \langle a_{i},a_{j} \right \rangle|}
\end{align}
\end{lem}

\begin{proof}[Proof of Lemma \ref{piub}] The proof is based on \cite[Theorem 3.5]{tropp2004greed}. 
The following is derived.
\begin{align}\nonumber
\underset{i \in {J}}{\max} \left \| A^{\dagger}_{{Q}} a_i\right \|_1 & \overset{(a)}=\underset{i \in {J}}{\max} \left \| (A^*_{{Q}}A_{{Q}})^{-1}A^*_{{Q}} a _i\right \|_1 \\\nonumber
&\leq \left \|(A^*_{{Q}}A_{{Q}})^{-1}\right \|_{1,1} \underset{i \in {J}}{\max} \left \| A^*_{{Q}}a_i \right \|_1 \\\label{msi1}
&=\left \|(A^*_{{Q}}A_{{Q}})^{-1}\right \|_{1,1} \underset{i \in {J}}{\max} (\underset{j \in {Q}}{\sum}|\left \langle a_{i},a_{j} \right \rangle|),
\end{align}
where $(a)$ follows from the definition of the Moore-Penrose pseudoinverse with $\sigma_{|{Q}|}(A_{Q})>0$. 
Note that $A^*_{{Q}}A_{{Q}}$ has a unit diagonal because all indices are $l_2$-normalized. The off-diagonal part $B$ thus satifies
\begin{align}
A^*_{{Q}}A_{{Q}}=I+B,
\end{align}
where $B_{ii}=0$ and $B_{ij}=\left \langle a_{i},a_{j} \right \rangle$ for $i,j(i\neq j)\in Q$.
Then
\begin{align}\nonumber
\left \| (A^*_{{Q}}A_{{Q}})^{-1} \right \|_{1,1}&=\left \| (I+B)^{-1} \right \|_{1,1}\\\nonumber
&=\left \| \sum\limits_{k=0}^{\infty}(-B)^k \right \|_{1,1} \\\nonumber
&\leq \sum\limits_{k=0}^{\infty} \left \| B \right \|^k_{1,1}\\\nonumber
&=\frac{1}{1- \left \| B \right \|_{1,1}}\\\label{sse_eval_temp}
& \leq \frac{1}{1-\max\limits_{i \in {Q}}\sum\limits_{j \neq i,j \in {Q}}|\left \langle a_{i},a_{j} \right \rangle|},
\end{align}
where $(a)$ follows from the fact that the Neumann series $\sum\limits_{k=0}^{\infty}(-B)^k$  converges to the inverse $(I+B)^{-1}$ if $\left \| B \right \|_{1,1}<1$ \cite{kreyszig1989introductory,tropp2004greed}. 
Therefore, (\ref{msi2}) is derived by applying (\ref{sse_eval_temp}) to (\ref{msi1}).
\end{proof}

\begin{lem}(\cite[Theorem 5.3]{foucart2013mathematical})
\label{gdt}
Let $A_{S}:=[a_1,...,a_q] \in \mathbb{K}^{m \times q}$ be a matrix with $l_2$-normalized columns. For any $x \in \mathbb{K}^{q}$,
\begin{align}\nonumber
&\left(1-\max\limits_{i \in S}\sum\limits_{j \in S, i \neq j}|\left \langle{a}_{i},{a}_{j}\right \rangle| \right) \left \| x \right \|^2_2 \leq \left \| A_{S}x \right \|^2_2 \leq \left(1+\max\limits_{i \in S}\sum\limits_{j \in S, i \neq j}|\left \langle{a}_{i},{a}_{j}\right \rangle| \right) \left \| x \right \|^2_2
\end{align}
or equivelently, the squared sigular values of $A_{S}$ or the eigenvalues of $A_{S}^*A_{S}$ lie in the interval \\
$\left[1-\max\limits_{i \in S}\sum\limits_{j \in S, i \neq j}|\left \langle{a}_{i},{a}_{j}\right \rangle|,1+\max\limits_{i \in S}\sum\limits_{j \in S, i \neq j}|\left \langle{a}_{i},{a}_{j}\right \rangle|\right]$.
\end{lem}

\begin{lem}(Generalization of \cite[Lemma 1]{laurent2000adaptive})\label{gl1}
Let $y:=[y_1,...,y_m]^{\top} \in \mathbb{K}^{m}$ be a vector whose elements are i.i.d. Gaussian variables, with mean $0$ and variance $\sigma^2$. Let $z=\sum\limits_{i=1}^{m}(y^2_i-\sigma^2)$. Then, the following inequality holds for any $t$ and negative $u$ such that $-\frac{1}{2\sigma^2}<u<0$.
\begin{align}\nonumber
\mathbb{P}(z \leq t)  \leq e^{ m\sigma^4 u^2 -u t}
\end{align}
If we set $u$ as $-\frac{\sqrt{x}}{\sigma^2 \sqrt{m}}$ such that $x < \frac{m}{4}$ and set $t$ as $-2\sigma^2 \sqrt{m x}$, the following inequality holds for $x < \frac{m}{4}$.
\begin{align}\nonumber
\mathbb{P}(z \leq -2\sigma^2 \sqrt{m x})  \leq e^{ -x}
\end{align}
\end{lem}

\begin{proof}[Proof of Lemma \ref{gl1}]
Let $\bar y$ be a random variable with $\mathcal{N}(0, \sigma^2)$. Let $\psi$ denote the logarithm of the Laplace transform of $\bar y^2-\sigma^2$. Then, for $-\frac{1}{2\sigma^2}<u<0$, 
\begin{align}\nonumber
\psi(u)&= \log [\mathbb{E}[\exp (u(\bar y^2-\sigma^2))]]\\\nonumber
&=-u\sigma^2-\frac{1}{2} \log (1-2 u\sigma^2) \\\nonumber
&\leq \sigma^4 u^2.
\end{align}
So we have
\begin{align}\nonumber
\log [\mathbb{E}[\exp (uz)]]&= \sum\limits_{i=1}^{m} \log [\mathbb{E}[\exp ( u(y_i^2-\sigma^2))]]\\\nonumber
&\leq m \sigma^4 u^2.
\end{align}
With the Laplace transform method,
the following is obtained.
\begin{align}\nonumber
\mathbb{P}(z \leq t) \leq e^{-u t} \,\,\mathbb{E}[\exp (uz)] \leq e^{m \sigma^4 u^2-u t}
\end{align}
\end{proof}

\begin{lem}[Lemma 6 in \cite{haupt2010toeplitz}]
\label{tc}
Let $x_i,y_i$ $(i=1,...,m)$ be sequences of i.i.d. Gaussian variables with mean $0$ and variance $\sigma^2$. Then
\begin{align}\nonumber
\mathbb{P}\left( \left|\sum\limits_{i=1}^{m}x_iy_i \right| \geq t \right)  \leq 2 \exp \left(-\frac{t^2}{4 \sigma^2(m \sigma^2 + t/2)} \right).
\end{align}
\end{lem}

\begin{lem}(Generalization of \cite[Theorem $\it{II}$.13]{davidson2001local})
\label{thm_lot}
Given $s,m \in \mathbb{N}$ with $s \leq m$, consider a matrix $A \in \mathbb{R}^{m \times s}$ whose entries are i.i.d. Gaussian variables with mean $0$ and variance $\bar \sigma^2$. Then, for any $t>0$,
\begin{align}\nonumber
&\mathbb{P}(\sigma_{1}(A) \geq \bar \sigma(\sqrt{m}+\sqrt{s}+t)) \leq \exp \left(-\frac{t^2}{2} \right)\\\nonumber
&\mathbb{P}(\sigma_{s}(A) \leq \bar \sigma(\sqrt{m}-\sqrt{s}-t)) \leq \exp \left(-\frac{t^2}{2} \right).
\end{align}
\end{lem}
\begin{proof}[Proof of Lemma \ref{thm_lot}]
Let $\tilde A \in \mathbb{R}^{m \times s}$ be a matrix whose entries are i.i.d. Gaussian variables with mean $0$ and variance $1/m$. Then by  \cite[Theorem $\textrm{II}$.13]{davidson2001local},
\begin{align}\nonumber
\bar \sigma(\sqrt{m}-\sqrt{s}) &< \bar \sigma \sqrt{m} \cdot \mathbb{E}(\sigma_s(\tilde A))\\\nonumber
&=\mathbb{E}(\sigma_s(A))\\\nonumber
&\leq \mathbb{E}(\sigma_1(A))\\\nonumber
&=\bar \sigma \sqrt{m} \cdot \mathbb{E}(\sigma_1(\tilde A))\\\nonumber
&<\bar \sigma(\sqrt{m}+\sqrt{s})
\end{align}
and for $t>0$,
\begin{align}\nonumber
&\max [P(\sigma_1(A) \geq \bar \sigma(\sqrt{m}+\sqrt{s})+\bar \sigma t), P(\sigma_s(A) \leq \bar \sigma(\sqrt{m}-\sqrt{s})-\bar \sigma t)]\\\nonumber 
&=\max [P(\sigma_1(\tilde A) \geq (1+\frac{\sqrt{s}}{\sqrt{m}})+\frac{t}{ \sqrt{m}}),P(\sigma_s(\tilde A) \leq  (1-\frac{\sqrt{s}}{\sqrt{m}})-\frac{t}{ \sqrt{m}})]\\\nonumber 
&<\exp (-\frac{t^2}{2}).
\end{align}
\end{proof}

\begin{lem}\label{lem_tnlb}
Let $X_0 \in \mathbb{K}^{n \times l}$ be row $k$-sparse with $\Omega \subseteq \Sigma$. Let $\bar{S}$ be an $r$-dimensional subspace of $\mathcal{R}(\Phi_{\Omega}X_0^{\Omega})$. Let $d$ $(\geq 1)$ be $\rank(P^{\perp}_{\mathcal{R}(\Phi_{\Gamma})}\bar{S})$. 
Then the following inequality holds.
\begin{align}
\label{tnlb}
&\underset{l \in \Omega \setminus  \Gamma}{\max}\left \| P_{P^{\perp}_{\mathcal{R}(\Phi_\Gamma)}\bar S}\dot{\phi}_l \right \|_2 \geq \sqrt{\frac{d}{|\Omega \setminus \Gamma|}}\sigma_{|\Omega \setminus \Gamma|}(\dot{\Phi}_{\Omega \setminus \Gamma}),
\end{align}
where $d=\min(r,|\Omega \setminus \Gamma|)$ if $X_0^{\Omega}$ is row-nondegenerate and $\sigma_{|\Omega \cup \Gamma|}(\Phi_{\Omega\cup \Gamma})>0$.

\end{lem}

\begin{proof}[Proof of Lemma \ref{lem_tnlb}] The proof is based on that of \cite[Theorem 7.10]{lee2012subspace}. 
Let $d$ be $\rank (P_{P^{\perp}_{R(\Phi_\Gamma)}\bar S})$. 
Then, for $i \in \{1,...,d\}$,
\begin{align}\nonumber
&\sigma_{i}(P_{P^{\perp}_{R(\Phi_\Gamma)}\bar S} \dot{\Phi}_{\Omega \setminus \Gamma})\\\nonumber
&\overset{(a)}\geq \sigma_{d}(P_{P^{\perp}_{R(\Phi_\Gamma)}\bar S} \dot{\Phi}_{\Omega \setminus \Gamma}) \\\nonumber
&\overset{(b)}= \max\limits_{\underset{\dim \mathcal{M} =d}{\mathcal{M}\subseteq \mathbb{K}^{|\Omega \setminus \Gamma|}}} \min\limits_{\underset{\left\| x \right\|_2=1}{x \in \mathcal{M}}} \left \| P_{P^{\perp}_{R(\Phi_\Gamma)}\bar S} \dot{\Phi}_{\Omega \setminus \Gamma} x \right \|_2 \\\nonumber
&\overset{(c)}\geq \min\limits_{\underset{\left\| x \right\|_2=1}{x \in R(P_{P^{\perp}_{R(\Phi_\Gamma)}\bar S} \dot{\Phi}_{\Omega \setminus \Gamma})}} \left \| P_{P^{\perp}_{R(\Phi_{\Gamma })}\bar S} \dot{\Phi}_{\Omega \setminus \Gamma} x \right \|_2 \\\nonumber
&= \min\limits_{\underset{\left\| x \right\|_2=1}{x \in R(P_{P^{\perp}_{R(\Phi_\Gamma)}\bar S} \dot{\Phi}_{\Omega \setminus \Gamma})}} \left \| \dot{\Phi}_{\Omega \setminus \Gamma} x \right \|_2 \\\nonumber
&\geq \min\limits_{\left\| x \right\|_2=1 } \left \| \dot{\Phi}_{\Omega \setminus \Gamma} x \right \|_2 \\\label{ranklow}
&= \sigma_{|\Omega \setminus \Gamma|}(\dot{\Phi}_{\Omega \setminus \Gamma}),
\end{align}
where $(a)$ and $(b)$ follow from Corollaries $\textup{III}.1.5$ and $\textup{III}.1.2$ in  \cite{bhatia2013matrix}, respectively, and $(c)$ follows from the fact that $\rank(P_{P^{\perp}_{R(\Phi_{\Gamma })}\bar S} \dot{\Phi}_{\Omega \setminus \Gamma})=d$ (i.e., $\mathcal{R}(P^{\perp}_{\mathcal{R}(\Phi_\Gamma)}\bar S) \subseteq \mathcal{R}(\dot{\Phi}_{\Omega \setminus \Gamma})$). 
The left-hand side of (\ref{tnlb}) is bounded by
\begin{align}\nonumber
\underset{l \in \Omega \setminus  \Gamma}{\max}\left \| P_{P^{\perp}_{\mathcal{R}(\Phi_\Gamma)}\bar S}\dot{\phi}_l \right \|_2&=\underset{l \in \Omega \setminus  \Gamma}{\max}\left \| \dot{\phi}_l^{*}P_{P^{\perp}_{\mathcal{R}(\Phi_\Gamma)}\bar S} \right \|_2\\\nonumber
&=\left \| \dot{\Phi}_{\Omega \setminus \Gamma}^{*}P_{P^{\perp}_{\mathcal{R}(\Phi_\Gamma)} \bar S} \right \|_{2,\infty}\\\label{tneq}
&\geq \frac{\left \| \dot{\Phi}_{\Omega \setminus \Gamma}^{*}P_{P^{\perp}_{\mathcal{R}(\Phi_\Gamma)}\bar S} \right \|_{F}}{\sqrt{|\Omega \setminus \Gamma|}}.
\end{align}
The last expression in (\ref{tneq}) has a lower bound of
\begin{align}\nonumber
\frac{\left \| \dot{\Phi}_{\Omega \setminus \Gamma}^{*}P_{P^{\perp}_{\mathcal{R}(\Phi_\Gamma)}\bar S} \right \|_{F}}{\sqrt{|\Omega \setminus \Gamma|}} &= \frac{\sqrt{\sum\limits_{i=1}^{d}\sigma^2_{i}(P_{P^{\perp}_{\mathcal{R}(\Phi_\Gamma)}\bar S} \dot{\Phi}_{\Omega \setminus \Gamma})} }{\sqrt{|\Omega \setminus \Gamma|}}\\\nonumber
&\overset{(a)}\geq (\frac{\sum_{g=1}^{d}\sigma^2_{|\Omega \setminus \Gamma|}(\dot{\Phi}_{\Omega \setminus \Gamma})}{{|\Omega \setminus \Gamma|}})^{1/2}\\\label{rank_lb2}
&= \sqrt{\frac{d}{|\Omega \setminus \Gamma|}}\sigma_{|\Omega \setminus \Gamma|}(\dot{\Phi}_{\Omega \setminus \Gamma}),
\end{align}
where $(a)$ follows from $(\ref{ranklow})$. 
Therefore, (\ref{tnlb}) is obtained from (\ref{tneq}) and (\ref{rank_lb2}).
If $X_0^{\Omega}$ is row-nondegenerate and $\sigma_{|\Omega \cup \Gamma|}(\Phi_{\Omega\cup \Gamma})>0$, then it follows from Lemma \ref{rankbound_rownond} that 
\begin{align}\nonumber
d=\min (r,|\Omega \setminus \Gamma|).
\end{align}
\end{proof}

\begin{lem}\label{lem_tnlb2}
Suppose that $X_0 \in \mathbb{K}^{n \times l}$ is row $k$-sparse with $\Omega \subseteq \Sigma$ and $X_0^{\Omega}$ is row-nondegenerate. Let $\Gamma$ be a proper subset of $\Sigma$. Let $\bar{S}$ be an $r$-dimensional subspace of $\mathcal{R}(\Phi_{\Omega}X_0^{\Omega})$. Suppose that $|\Omega \cap \Gamma| \geq k-r$ and $\sigma_{|\Omega \cup \Gamma|}(\Phi_{\Omega\cup \Gamma})>0$. Then the following equality holds for $l \in \Omega \setminus  \Gamma$.
\begin{align}
\label{tnlb2}
&\left \| P_{P^{\perp}_{\mathcal{R}(\Phi_\Gamma)}\bar S}\dot{\phi}_l \right \|_2 =1
\end{align}
\end{lem}

\begin{proof}[Proof of Lemma \ref{lem_tnlb2}]
From the assumptions that $X_0^{\Omega}$ is row-nondegenerate, $|\Omega \cap \Gamma| \geq k-r$, and $\sigma_{|\Omega \cup \Gamma|}(\Phi_{\Omega\cup \Gamma})>0$, the following holds by Lemma \ref{subaug}. 
\begin{align}\label{prop54}
\bar S + \mathcal{R}(\Phi_{\Omega \cap \Gamma})=\mathcal{R}(\Phi_{\Omega})
\end{align}
With (\ref{prop54}), one obtains
\begin{align}\nonumber
\bar S + \mathcal{R}(\Phi_{\Gamma})&=\bar S + \mathcal{R}(\Phi_{\Omega \cap \Gamma} \cup \Phi_{\Gamma\setminus \Omega})\\\nonumber
&=[\bar S + \mathcal{R}(\Phi_{\Omega \cap \Gamma})]\cup \mathcal{R}(\Phi_{\Gamma\setminus \Omega})\\\nonumber
&=\mathcal{R}(\Phi_{\Omega} \cup \Phi_{\Gamma \setminus \Omega})\\\label{prop54_t1}
&=\mathcal{R}(\Phi_{\Omega \cup \Gamma}).
\end{align}
Hence,
\begin{align}\nonumber
P_{P^{\perp}_{\mathcal{R}(\Phi_{\Gamma})} \bar S}&=P_{P^{\perp}_{\mathcal{R}(\Phi_{\Gamma})} [\bar S + \mathcal{R}(\Phi_{\Gamma})]}\\\nonumber
&\overset{(a)}=P_{P^{\perp}_{\mathcal{R}(\Phi_{\Gamma})} \mathcal{R}(\Phi_{\Omega \cup \Gamma})}\\\nonumber
&=P_{\mathcal{R}(P^{\perp}_{\mathcal{R}(\Phi_{\Gamma})}\Phi_{\Omega})},
\end{align}
where $(a)$ follows from (\ref{prop54_t1}). 
Therefore, for $i \in \Omega \setminus \Gamma$,
\begin{align}\nonumber
\left \| P_{P^{\perp}_{\mathcal{R}(\Phi_{\Gamma})} \bar S} \dot\phi_i \right \|_2
=\left \| P_{\mathcal{R}(P^{\perp}_{\mathcal{R}(\Phi_{\Gamma})}\Phi_{\Omega})}\dot\phi_i \right \|_2
=\left \|\dot\phi_i \right \|_2
= 1,
\end{align}
since $ \mathcal{R}(\dot\phi_i) \subseteq \mathcal{R}(P^{\perp}_{\mathcal{R}(\Phi_{\Gamma})}\Phi_{\Omega})$.
\end{proof}

\begin{lem}\label{lem_ntnub}
Let $X_0 \in \mathbb{K}^{n \times l}$ be row $k$-sparse with $\Omega \subseteq \Sigma$. Let $\Gamma$ be a proper subset of $\Sigma$. Let $\bar{S}$ be an $r$-dimensional subspace of $\mathcal{R}(\Phi_{\Omega}X_0^{\Omega})$.  Suppose that $\sigma_{|\Omega \cup \Gamma|}(\Phi_{\Omega \cup \Gamma }) >0$. Then the following inequality holds.
\begin{align}\label{ntnub}
&\underset{i \in \Sigma \setminus (\Omega \cup \Gamma)}{\max}{\left \| P_{\mathcal{R}(P^{\perp}_{\mathcal{R}(\Phi_{\Gamma})}\bar{S})}\dot\phi_i \right \|_2} \leq  \left(\underset{i \in \setminus (\Omega \cup \Gamma)}{\max}\left \| \dot\Phi^{\dagger}_{\Omega \setminus \Gamma }\dot\phi_{i} \right \|_1 \right) \left(\underset{l \in \Omega \setminus \Gamma}{\max}{\left \| P_{\mathcal{R}(P^{\perp}_{\mathcal{R}(\Phi_{\Gamma})}\bar{S})}\dot\phi_l \right \|_2} \right) 
\end{align}
\end{lem}

\begin{proof}[Proof of Lemma \ref{lem_ntnub}] The proof is based on that of \cite[Proposition 1]{davies2012rank}. 
Let $U=[u_1,...,u_d] \in \mathbb{K}^{m \times d}$ be an orthonormal basis of $\mathcal{R}(P^{\perp}_{\mathcal{R}(\Phi_{\Gamma})} \bar S)$. 
Derive a lower bound of the right-hand side of $(\ref{ntnub})$ by
\begin{align}\nonumber
&(\underset{i \in \setminus (\Omega \cup \Gamma)}{\max}\left \| \dot\Phi^{\dagger}_{\Omega \setminus \Gamma }\dot\phi_{i} \right \|_1)(\underset{l \in \Omega \setminus \Gamma}{\max}{\left \| P_{\mathcal{R}(P^{\perp}_{\mathcal{R}(\Phi_{\Gamma})}\bar{S})}\dot\phi_l \right \|_2})\\\nonumber
&=(\underset{i \in \Sigma - \Omega -\Gamma}{\max}\left \| \dot\Phi^{\dagger}_{\Omega \setminus \Gamma}\dot\phi_{i} \right \|_1)(\underset{l \in \Omega-\Gamma}{\max}\left \| \dot\phi^{*}_l U  \right \|_2)\\\nonumber
&= (\underset{i \in \Sigma - \Omega -\Gamma}{\max}\left \| \dot\Phi^{\dagger}_{\Omega \setminus \Gamma}\dot\phi_{i} \right \|_1)(\underset{x \neq 0}{\max}\frac{\left \| \dot\Phi_{\Omega \setminus \Gamma}^{*}Ux \right \|_{\infty}}{\left \| x \right \|_2}) \\\nonumber
&\overset{(a)}\geq \underset{x \neq 0}{\max}\,\,\underset{i \in \Sigma - \Omega -\Gamma}{\max} \, [\frac{|(\dot\Phi^{\dagger}_{\Omega \setminus \Gamma}\dot\phi_{j})^{*}\dot\Phi_{\Omega \setminus \Gamma}^{*}Ux|}{\left \| \dot\Phi_{\Omega \setminus \Gamma}^{*}Ux \right \|_{\infty}}\frac{\left \| \dot\Phi_{\Omega \setminus \Gamma}^{*}Ux \right \|_{\infty}}{\left \| x \right \|_2}]\\\nonumber
&=\underset{i \in \Sigma - \Omega -\Gamma}{\max}\,\,\underset{x \neq 0}{\max}\frac{|(\dot\Phi^{\dagger}_{\Omega \setminus \Gamma}\dot\phi_{j})^{*}\dot\Phi_{\Omega \setminus \Gamma}^{*}Ux|}{\left \| x \right \|_2}\\\nonumber
&\overset{(b)}= \underset{i \in \Sigma - \Omega -\Gamma}{\max}\left \| \dot\phi^{*}_i P_{\mathcal{R}(\dot\Phi_{\Omega \setminus \Gamma})}U \right \|_2\\\nonumber
&\overset{(c)}=\underset{i \in \Sigma - \Omega -\Gamma}{\max}\left \| \dot\phi^{*}_i U \right \|_2\\\label{ntnub1}
&=\underset{i \in \Sigma \setminus (\Omega \cup \Gamma)}{\max}{\left \| P_{\mathcal{R}(P^{\perp}_{\mathcal{R}(\Phi_{\Gamma})}\bar{S})}\dot\phi_i \right \|_2},
\end{align}
where (a)--(c) follow from the Holder's inequality, conditions $P_{\mathcal{R}(\dot\Phi_{\Omega \setminus \Gamma})}=(\dot\Phi^{\dagger}_{{\Omega \setminus \Gamma}})^{*}\dot\Phi_{{\Omega \setminus \Gamma}}^{*}$, and $\mathcal{R}(U) \subseteq \mathcal{R}(\dot\Phi_{\Omega-\Gamma})$, respectively.
\end{proof}

\begin{lem}\label{lem_2ntnub}
Let $X_0 \in \mathbb{K}^{n \times l}$ be row $k$-sparse with $\Omega$.  
Let $\bar{S}$ be an $r$-dimensional subspace of $\mathcal{R}(\Phi_{\Omega}X_0^{\Omega})$.  
Then the following inequality holds for $\Gamma \subseteq \Sigma$ and $i \in \Sigma - \Omega -\Gamma$.
\begin{align}\label{2ntnub}
& {\left \| P_{\mathcal{R}(P^{\perp}_{\mathcal{R}(\Phi_{\Gamma})}\bar{S})}\dot\phi_i \right \|_2}
\leq \sqrt{1-  \sigma^2_{|\{i\} \cup \Omega \setminus \Gamma|}(\dot\Phi_{\{i\} \cup \Omega \setminus \Gamma})}
\end{align}
\end{lem}
\begin{proof}[Proof of Lemma \ref{lem_2ntnub}]
The proof is based on that of \cite[Theorem 7.10]{lee2012subspace}. 
The following is obtained.
\begin{align}\nonumber
&\left \|  P_{\mathcal{R}(P^{\perp}_{\mathcal{R}(\Phi_{\Gamma})}\bar{S})}\dot\phi_i \right \|_2 \\\nonumber
&\overset{(a)}\leq \left \|  P_{\mathcal{R}(P^{\perp}_{\mathcal{R}(\Phi_{\Gamma})}\Phi_{\Omega \setminus \Gamma})}\dot\phi_i \right \|_2 \\\nonumber
&\overset{(b)}= \left \|  P_{\mathcal{R}(\dot\Phi_{\Omega \setminus \Gamma})}\dot\phi_i \right \|_2 \\\label{2ntnub2}
&=  \sqrt{1- \left \|  P^{\perp}_{\mathcal{R}(\dot\Phi_{\Omega \setminus \Gamma})}\dot\phi_i \right \|^2_2},
\end{align}
where (a), (b), and (c) follow from that $\mathcal{R}(P^{\perp}_{\mathcal{R}(\Phi_{\Gamma})}\bar{S}) \subseteq \mathcal{R}(P^{\perp}_{\mathcal{R}(\Phi_{\Gamma})}\Phi_{\Omega \setminus \Gamma})$ and the definition of $\dot\Phi_{\Omega \setminus \Gamma}$, respectively. 
By Lemma \ref{cpslb}, the last expression in (\ref{2ntnub2}) has an upper bound of
\begin{align}\label{2ntnub3}
\sqrt{1-  \sigma^2_{|\{i\} \cup \Omega \setminus \Gamma|}(\dot\Phi_{\{i\} \cup \Omega \setminus \Gamma})}
\end{align}
so that the proof is complete from (\ref{2ntnub2}) and (\ref{2ntnub3}).
\end{proof}

\begin{lem}
\label{cor_eachiter_low1}
Let $\Phi = [\phi_1,...,\phi_n] \in \mathbb{R}^{m \times n}$ be a matrix where $\frac{\phi_i}{\left \|\phi_i \right \|_2}$ for $i \in [1,...,n]$ is independently and uniformly distributed on the unit sphere $\mathbb{S}^{m-1}$. Suppose that for $\Gamma \subseteq \Sigma$,  $m \geq m^*:=|\Gamma|+\frac{c \ln \, n}{\lambda^{-1}(z)}$ and $m > |\Gamma|+4c \ln \, n$ for any constant $c$ and $z$.
Then, for ${i \in \Sigma \setminus ( \Omega \cup \Gamma)}$,
\begin{align}\label{cor_eachiter_low1_222}
&\mathbb{P}(\mu(\Omega \cup \{i\},\Gamma) \geq z | \Phi_{\Gamma}) \leq 4k \cdot n^{-c} 
\end{align}
and
\begin{align}\label{cor_eachiter_low1_2}
&\mathbb{P}(\max\limits_{j \in \Sigma \setminus ( \Omega \cup \Gamma)}\mu(\Omega \cup \{j\},\Gamma) \geq z| \Phi_{\Gamma})\leq 4k \cdot n^{1-c}.
\end{align}
\end{lem}

\begin{proof}[Proof of Lemma \ref{cor_eachiter_low1}]
Given $\Gamma \subseteq \Sigma$, since $\lambda(x)$ is a non-decreasing function for $x$, it follows that 
\begin{align}\label{cor_eachiter_low1_1}
&\lambda(\frac{ c\ln\,n}{m-|\Gamma|})\leq \lambda(\frac{ c\ln\,n}{m^*-|\Gamma|})=z.
\end{align}
The following inequalities hold for $a,b \in \Sigma \setminus \Gamma$
\begin{align}\nonumber
&\mathbb{P}(|\left \langle \dot\phi_a,\dot\phi_b \right \rangle| \geq z | \Phi_{\Gamma}) \\\nonumber
&\overset{(a)}\leq \mathbb{P}(|\left \langle \dot\phi_a,\dot\phi_b \right \rangle| \geq \lambda(\frac{ c\ln\,n}{m-|\Gamma|})| \Phi_{\Gamma}) \\\label{cor_eachiter_low1_3}
&\overset{(b)}\leq 4n^{-c},
\end{align}
where $(a)$ and $(b)$ follow from (\ref{cor_eachiter_low1_1}) and Corollary \ref{propcol_pcu} where $(p,q)=(m-|\Gamma|,n)$, respectively.
Finally, applying the union bound to (\ref{cor_eachiter_low1_3}) yields (\ref{cor_eachiter_low1_222}) and (\ref{cor_eachiter_low1_2}).
\end{proof}

\begin{lem}
\label{cor_eachiter_low1_2new1}
Let $\Phi = [\phi_1,...,\phi_n] \in \mathbb{R}^{m \times n}$ be a matrix where $\frac{\phi_i}{\left \|\phi_i \right \|_2}$ for $i \in [1,...,n]$ is independently and uniformly distributed in the unit sphere in $\mathbb{R}^{m}$. Suppose that for $\Gamma \subseteq \Sigma$,  $m \geq m^*:=|\Omega \cup \Gamma|-1+\frac{c \ln \, n}{\lambda^{-1}(z)}$ and $m > |\Omega \cup \Gamma|-1+4c \ln \, n$ for any constants $c$ and $z$.
Then, for ${i \in \Sigma \setminus ( \Omega \cup \Gamma)}$,
\begin{align}\label{cor_eachiter_low1_2_2}
\mathbb{P}(\mu(\Omega \cup \{i\},\Gamma)& \geq z| \Phi_{\Omega \cup \Gamma}) \leq 4k \cdot n^{-c} 
\end{align}
and
\begin{align}\label{cor_eachiter_low1_2_2_1}
\mathbb{P}(\max\limits_{i \in \Sigma \setminus ( \Omega \cup \Gamma)}\mu(\Omega \cup \{i\},\Gamma) &\geq z| \Phi_{\Omega \cup \Gamma}) \leq 4k \cdot n^{1-c}.
\end{align}
\end{lem}

\begin{proof}[Proof of Lemma \ref{cor_eachiter_low1_2new1}]
Given $\Gamma \subseteq \Sigma$, since $\lambda(x)$ is a non-decreasing function of $x$, it follows that 
\begin{align}\label{cor_eachiter_low1_1_2}
&\lambda(\frac{ c\ln\,n}{m-|\Omega \cup \Gamma|})\leq \lambda(\frac{ c\ln\,n}{m^*-|\Omega \cup \Gamma|})=z.
\end{align}
The following inequalities hold for $t \in \Sigma \setminus (\Omega \cup \Gamma)$.
\begin{align}\nonumber
&\mathbb{P}(\mu(\Omega \cup \{t\},\Gamma) \geq z| \Phi_{\Omega \cup \Gamma}) \\\nonumber
&= \mathbb{P}(\max\limits_{\{a,b\} \subseteq \Omega \cup \{t\} \setminus \Gamma}\mu(\{a,b\},\Gamma ) \geq z | \Phi_{\Omega \cup \Gamma}) \\\nonumber
&\overset{(a)}\leq \mathbb{P}(\max\limits_{\{a,b\} \subseteq \Omega \cup \{t\}}\mu(\{a,b\},\Gamma \cup \Omega \setminus \{a,b\}) \geq z | \Phi_{\Omega \cup \Gamma}) \\\nonumber
& \overset{(b)} = \mathbb{P}(\max\limits_{a \in \Omega}\mu(\{a,t\},\Gamma \cup \Omega \setminus \{a\}) \geq z | \Phi_{\Omega \cup \Gamma}) \\\nonumber
& \overset{(c)}\leq \mathbb{P}(\max\limits_{a \in \Omega}\mu(\{a,t\},\Gamma \cup \Omega \setminus \{a\}) \geq \lambda(\frac{ c\ln\,n}{m-|\Omega \cup \Gamma|+1}) \\\nonumber
& \,\,\,\,\,\,\,\, \,\,\,\, \,\,\,\,  | \Phi_{\Omega \cup \Gamma}) \\\label{cor_eachiter_low1_3_2}
&\overset{(d)}\leq 4k\cdot n^{-c},
\end{align}
where (a) follows from Lemma \ref{proj_ineq} and the fact that $|\Gamma \cap \{a,b\}|=0$, and (b), (c), and (d) follow from Lemma \ref{ineq_unew}, (\ref{cor_eachiter_low1_1_2}), and Corollary \ref{propcol_pcu} where $(p,q)=(m-|\Omega \cup \Gamma|+1,n)$, respectively.
It follows that
\begin{align}\nonumber
\mathbb{P}(\max\limits_{i \in \Sigma \setminus ( \Omega \cup \Gamma)}\mu(\Omega \cup \{i\},\Gamma) \geq z| \Phi_{\Omega \cup \Gamma})
&\overset{(a)}\leq 4k \cdot n^{1-c},
\end{align}
where $(a)$ follows from the application of the union bound to (\ref{cor_eachiter_low1_3_2}). 
\end{proof}

\begin{lem}\label{ineq_unew}
For $\Omega,\Gamma \subseteq \Sigma$ and $\Phi \in \mathbb{R}^{m \times n}$, it follows that for $t \in \Sigma \setminus ( \Omega \cup \Gamma)$,
\begin{align}\label{ineq_unew_eq1}
&\max\limits_{\{a, b\} \subseteq (\Omega \cup \{t\})}\mu(\{a,b\},\Gamma \cup \Omega \setminus \{a,b\})=  \underset{a \in \Omega}{\max}\,\mu(\{a,t\},\Gamma \cup \Omega \setminus \{a\}). 
\end{align}
\end{lem}
\begin{proof}
Since the left-hand side of $(\ref{ineq_unew_eq1})$ is upper bounded by
\begin{align}\label{ineq_unew_eq11}
\max\{&\underset{(*)}{\underbrace{\max\limits_{\{a, b\} \subseteq \Omega}\mu(\{a,b\},\Gamma \cup \Omega \setminus \{a,b\})}}, \underset{(**)}{\underbrace{ \underset{a \in \Omega}{\max}\,\mu(\{a,t\},\Gamma \cup \Omega \setminus \{a\})}}\}, 
\end{align}
it suffices to show that $(*)\leq (**)$.  
Define a function such that $w(x):=\sqrt{1-x^2}$. For $A \in \mathbb{K}^{m \times a}$, $B \in \mathbb{K}^{m \times b}$ and  $C \in \mathbb{K}^{m \times c}$, it follows that
\begin{align}\nonumber
&\left \| P^{\perp}_{\mathcal{R}([A,C])}P_{\mathcal{R}([B,C])}\right \|_2\\\nonumber
&= \underset{\underset{\left \| x \right \|_2 =1 }{x \in \mathcal{R}([B,C])} }{\sup} \,\, \underset{ y \in \mathcal{R}([A,C])}{\inf}\, \left \| x-y\right \|_2 \\\nonumber
&= \underset{\underset{\left \| x_1\right \|^2_2 +\left \|x_2 \right \|^2_2 =1 }{\underset{x_2 \in \mathcal{R}(C),}{x_1 \in \mathcal{R}(P^{\perp}_{\mathcal{R}(C)}B),}}}{\sup} \,\, \underset{ \underset{y_2 \in \mathcal{R}(C)}{y_1 \in \mathcal{R}(P^{\perp}_{\mathcal{R}(C)}A),}}{\inf}\, \left \| x_1-y_1+x_2-y_2  \right \|_2 \\\nonumber
&= \underset{\underset{\left \| x_1 \right \|_2 =1 }{x_1 \in \mathcal{R}(P^{\perp}_{\mathcal{R}(C)}B)} }{\sup} \,\, \underset{ y_1 \in \mathcal{R}(P^{\perp}_{\mathcal{R}(C)}A)}{\inf}\, \left \| x_1-y_1\right \|_2 \\\label{projectionuprule}
&= \left \| P^{\perp}_{\mathcal{R}(P^{\perp}_{\mathcal{R}(C)}A)}P_{\mathcal{R}(P^{\perp}_{\mathcal{R}(C)}B)}\right \|_2.
\end{align}
Then we obtain
\begin{align}\nonumber
&w(\max\limits_{\{a, b\} \subseteq \Omega}\mu(\{a,b\},\Gamma \cup \Omega \setminus \{a,b\})) \\\nonumber
&=\min\limits_{\{a, b\} \subseteq \Omega} w(\mu(\{a,b\},\Gamma \cup \Omega \setminus \{a,b\})) \\\nonumber
& =\underset{\{a, b\} \subseteq \Omega}{\min}  \left \|P_{\mathcal{R}(P^{\perp}_{\mathcal{R}(\Phi_{\Gamma \cup \Omega \setminus \{a,b\}})}  \Phi_{\{a\}})}-P_{\mathcal{R}(P^{\perp}_{\mathcal{R}(\Phi_{\Gamma \cup \Omega \setminus \{a,b\}})} \Phi_{\{b\}})} \right \|_2 \\\nonumber 
& \overset{(a)}=   \min\limits_{\{a, b\} \subseteq \Omega} \left \|P_{\mathcal{R}(\Phi_{\Gamma \cup \Omega \setminus \{a\}})}-P_{\mathcal{R}(\Phi_{\Gamma \cup \Omega \setminus \{b\}})}  \right \|_2 \\\nonumber
& \overset{(b)}\geq    \min\limits_{a \in \Omega} \, \left \|P_{\mathcal{R}(\Phi_{\Gamma \cup \Omega \setminus \{a\}})}-P_{\mathcal{R}(\Phi_{\Gamma \cup \Omega})}  \right \|_2 \\\nonumber
& \geq    \min\limits_{a \in \Omega} \,\left \|P^{\perp}_{\mathcal{R}(\Phi_{\Gamma \cup \Omega \setminus \{a\}})}P_{\mathcal{R}(\Phi_{\Gamma \cup \Omega })} \right \|_2 \\\nonumber
& \geq    \underset{a \in \Omega}{\min} \,\left \|P^{\perp}_{\mathcal{R}(\Phi_{\Gamma \cup \Omega \cup \{t\} \setminus \{a\}})}P_{\mathcal{R}(\Phi_{\Gamma \cup \Omega })} \right \|_2 \\\nonumber
& \overset{(c)}=    \underset{a \in \Omega}{\min} \left \|P^{\perp}_{\mathcal{R}(P^{\perp}_{\mathcal{R}(\Phi_{\Gamma \cup \Omega \setminus \{a\}})}  \Phi_{\{t\}})}P_{\mathcal{R}(P^{\perp}_{\mathcal{R}(\Phi_{\Gamma \cup \Omega \setminus \{a\}})} \Phi_{\{a\}})} \right \|_2 \\\nonumber
&=  w(\underset{a \in \Omega}{\max}\,\mu(\{a,t\},\Gamma \cup \Omega \setminus \{a\}),
\end{align}
where (a), (b), and (c) follow from the projection update rule, Lemma \ref{proj_ineq}, and (\ref{projectionuprule}), respectively. 
Therefore, $(*)\leq (**)$ since $w(x) \geq w(y)$ implies $x \leq y$ for $0 \leq x,y \leq 1$.
\end{proof}

\begin{lem}\label{sceosmp0} Suppose that $\Omega \subseteq J$ and $\bar Y=\Phi_{\Omega}X^\Omega_0$. Then $\Phi_{J}^{\dagger} \bar Y=X_0$ holds if $\sigma_{|J|}(\Phi_{J})>0$.
\end{lem}
\begin{proof}[Proof of Lemma \ref{sceosmp0}] Consider a least square problem such as
\begin{align}\label{uniqdag}
\tilde X :=\underset{X}{\arg\min} \left \| \bar Y-\Phi_{J}X \right \|_2.
\end{align}
Since $\left \| \Phi_{J}X_0-\Phi_{J}X \right \|_2=0$ implies $\Phi_{J}(X_0-X)=0$,
$X_0$ is the unique solution of (\ref{uniqdag}) if $\sigma_{|J|}(\Phi_{J})>0$.
Therefore if $\sigma_{|J|}(\Phi_{J})>0$, $\Phi_{J}^{\dagger} \bar Y=X_0$ holds since $\left \| \bar Y-\Phi_{J} (\Phi_{J}^{\dagger} \bar Y) \right \|_2=\left \| \bar Y-P_{\mathcal{R}(\Phi_{J})} \bar Y \right \|_2=0$.
\end{proof} 

\section{Propositions and their Corollaries}

\begin{prop} 
\label{prop_pcu}
Define $A=[a_1,...,a_n] \in \mathbb{R}^{m \times n}$ as a matrix whose entries are i.i.d. Gaussian variables with mean $0$ and variance $\sigma^2$. Next, construct an $l_2$-normalized vector matrix $B:=[b_1,...,b_n]$ by taking $b_i:=\frac{a_i}{\left \| a_i \right \|_2}$ for $i \in [n]$. Provided that $m > (4 \cdot c) \ln \,n$, the following inequality holds for $i,j \in [n]$.
\begin{align}\label{prop_pcu_st}
&\mathbb{P} \left(|\left \langle b_i,b_j \right \rangle| <  \frac{c\sigma^2 \ln\,n + \sqrt{(c\sigma^2 \ln\,n)^2+4c\sigma^4m \ln \,n}}{m(\sigma^2-\sqrt{\frac{4c \sigma^4}{m} \ln\, n})} \right)> 1-4n^{-c} 
\end{align}
\end{prop}

\begin{proof}[Proof of Proposition \ref{prop_pcu}]
The proof is based on \cite[Theorem $8$]{bajwa2012two}. 
For $i\in [n]$, by Lemma \ref{gl1} where $y=a_i \in \mathbb{K}^{m}$, it follows that if $x < \frac{m}{4}$,
\begin{align}\label{gl1_applying_temp}
&\mathbb{P}(\left \| a_i \right \|^2_2 \leq m \sigma^2 -2 \sigma^2\sqrt{m x})  \leq e^{-x}.
\end{align}
Then (\ref{gl1_applying_temp}) guarantees that if $\delta_1 < \sigma^2$,
\begin{align}\label{gl1_applying}
&\mathbb{P}(\left \| a_i \right \|^2_2 \leq m (\sigma^2 - \delta_1)  \leq e^{- \frac{m \delta_1^2}{4 \sigma^4}}.
\end{align}
In addition, by (\ref{gl1_applying}) where $\delta_1=\sqrt{\frac{4c \sigma^4}{m} \ln\, n}$, the following probability upper bound is obtained, which holds when $ m > 4c \ln \,n$.
\begin{align}\label{ngf2}
\mathbb{P}(\left \| a_i \right \|^2_2 \leq m( \sigma^2 -\delta_1))  \leq e^{- \frac{m \delta_1^2}{4 \sigma^4}} = n^{-c} 
\end{align}
for $i \in [n]$.
Next, the inner product of two Gaussian random vectors is bounded as follows by using Lemma \ref{tc}.
\begin{align}
\mathbb{P}(|\left \langle a_i,a_j \right \rangle| \geq \delta_2)  \leq 2 \exp (-\frac{\delta_2^2}{4 \sigma^2(m \sigma^2 + \delta_2/2)})
\end{align}
By setting $\delta_2=c\sigma^2 \ln\,n + \sqrt{(c\sigma^2 \ln\,n)^2+4c\sigma^4m \ln \,n}$, it follows that for $i,j \in [n]$, 
\begin{align}\label{ngf1}
\mathbb{P}(|\left \langle a_i,a_j \right \rangle| \geq \delta_2)  \leq 2 \exp (-\frac{\delta_2^2}{4 \sigma^2(m \sigma^2 + \delta_2/2)}) = 2n^{-c}.
\end{align}
Therefore, by (\ref{ngf2}) and  (\ref{ngf1}), the following inequalities are obtained from the union bound:
\begin{align}\nonumber
&\mathbb{P}(|\left \langle b_i,b_j \right \rangle| \geq  \frac{c\sigma^2 \ln\,n + \sqrt{(c\sigma^2 \ln\,n)^2+4c\sigma^4m \ln \,n}}{m(\sigma^2-\sqrt{\frac{4c \sigma^4}{m} \ln\, n})}) \\\nonumber
&=\mathbb{P}(\frac{|\left \langle a_i,a_j \right \rangle|}{\left \| a_i \right \|_2\left \| a_j \right \|_2}\geq \frac{\delta_2}{m(\sigma^2-\delta_1)}) \\\nonumber
&\leq\mathbb{P}({|\left \langle a_i,a_j \right \rangle|} \geq {\delta_2})+ \\\nonumber
&\,\,\,\,\,\,\,\,\mathbb{P}({\left \| a_i \right \|^2_2}\leq {m(\sigma^2-\delta_1)})+\mathbb{P}({\left \| a_j \right \|^2_2}\leq {m(\sigma^2-\delta_1)}) \\\label{ngf21}
&\leq 4n^{-c},
\end{align}
where $i,j \in [n]$, $\delta_1=\sqrt{\frac{4c \sigma^4}{m} \ln\, n}$, and $\delta_2=c\sigma^2 \ln\,n + \sqrt{(c\sigma^2 \ln\,n)^2+4c\sigma^4m \ln \,n}$.
\end{proof}

\begin{cor}\label{propcol_pcu} 
Let $A := [a_1,...,a_q] \in \mathbb{R}^{p \times q}$ be a matrix where each column of $A$ is $l_2$-normalized and independently and uniformly distributed on the unit sphere $\mathbb{S}^{p-1}$. 
Provided that $p > (4 \cdot c) \ln \,q$, then for $i,j \in [q]$,
\begin{align}\nonumber
\mathbb{P}(|\left \langle a_i, a_j \right \rangle| \geq  \lambda(\frac{ c\ln\,q}{p})) \leq 4 \cdot q^{-c}
\end{align}
and consequently, 
\begin{align}\nonumber
\mathbb{P}(|\left \langle a_i, a_j \right \rangle| \geq  \lambda(\frac{ c\ln\,q}{p})|a_i) \leq 4 \cdot q^{-c},
\end{align}
where
\begin{align}\nonumber
&\lambda(x):=\frac{x + \sqrt{x^2+4x}}{1-2\sqrt{x}}.
\end{align} 
\end{cor}
\begin{proof}[Proof of Corollary \ref{propcol_pcu}]
Replacing $\sigma^2$ by $1/m$ in (\ref{prop_pcu_st}) completes the proof.
\end{proof}

\begin{prop}
\label{prop_guass}
Given $s,m \in \mathbb{N}$ with $s \leq m$, consider a matrix $A \in \mathbb{R}^{m \times s}$ whose elements are i.i.d. Gaussian variables with mean $0$ and variance $\sigma^2$. Then, for $g$ such that $1\geq g> \frac{s}{m}+2\sqrt{\frac{s}{m}}$,
\begin{align}\nonumber
&\mathbb{P} \left(\max \left[\frac{\sigma^2_1(A)}{m\sigma^2}-1,1-\frac{\sigma^2_k(A)}{m\sigma^2} \right]\geq g \right)\leq 2\exp \left[-\frac{m}{2} \left(\sqrt{1+g}-1-\sqrt{\frac{s}{m}} \right)^2 \right]. 
\end{align}
\end{prop}
\begin{proof}[Proof of Proposition \ref{prop_guass}]
The proof is based on \cite[Proposition $6.1$]{lee2012subspace}.
By Lemma \ref{thm_lot} where $\bar \sigma =\sigma$ and $t=\sqrt{m}(\sqrt{1+g}-1-\sqrt{\frac{s}{m}})$, it follows that for a value $g$ such that $1-\sqrt{1-g}-\sqrt{\frac{s}{m}}>0$,
\begin{align}\nonumber
&\mathbb{P}(\sigma_{1}(A) \geq \sigma\sqrt{m}\sqrt{1+g}) \leq \exp \left[-\frac{m}{2}\left(\sqrt{1+g}-1-\sqrt{\frac{s}{m}} \right)^2 \right]. 
\end{align}
By Lemma \ref{thm_lot} where $\bar \sigma =\sigma$ and $t=\sqrt{m}(1-\sqrt{1-g}-\sqrt{\frac{s}{m}})$, it holds that for a value $g$ such that $\sqrt{1+g}-1-\sqrt{\frac{s}{m}}>0$,
\begin{align}\nonumber
&\mathbb{P}(\sigma_{k}(A) \leq \sigma\sqrt{m}\sqrt{1-g}) \leq \exp \left[-\frac{m}{2} \left(1-\sqrt{1-g}-\sqrt{\frac{s}{m}}\right)^2 \right]. 
\end{align}
It follows that for $g$ such that $1\geq g> \frac{s}{m}+2\sqrt{\frac{s}{m}}$,
\begin{align}\nonumber
&\mathbb{P} \left(\max \left[\frac{\sigma^2_1(A)}{m\sigma^2}-1,1-\frac{\sigma^2_k(A)}{m\sigma^2} \right]\geq g \right)\leq 2\exp \left[-\frac{m}{2} \left(\sqrt{1+g}-1-\sqrt{\frac{s}{m}}\right)^2 \right], 
\end{align}
since $1-\sqrt{1-g}-\sqrt{\frac{s}{m}} \geq \sqrt{1+g}-1-\sqrt{\frac{s}{m}}>0$.
\end{proof}

\begin{cor}(Generalization of \cite[Proposition $6.1$]{lee2012subspace})
\label{3cor_thm_g3_3}
Let $A \in \mathbb{R}^{m \times n}$ be a matrix whose entries are i.i.d. Gaussian following $\mathcal{N}(0,\sigma^2)$. Define $\theta(\tau):=\frac{1-\tau}{1+\tau}$ where $\tau$ is a non-negative constant. 
Suppose that $m \geq \frac{2}{(\sqrt{1+\theta(\tau)}-1)^2} \left[k+2\ln \,\left(\frac{2(n-k)}{\epsilon}\right) \right]$. Then, for $\Omega \subseteq \Sigma$ such that $|\Omega|=k$,
\begin{align}\label{3cor_thm_g3_3t1}
\mathbb{P} \left( \frac{\underset{i \in \Sigma \setminus \Omega }{\min} \, \sigma^2_{k+1}(\Phi_{\Omega  \cup \{i\}  })}{\underset{j \in \Sigma \setminus \Omega}{\max} \,\sigma^2_{1}(\Phi_{\Omega \cup \{j\}  })} > \tau \right) > 1-\epsilon.
\end{align}
\end{cor}

\begin{proof}[Proof of Corollary \ref{3cor_thm_g3_3}]

The proof is based on \cite[Proposition $6.1$]{lee2012subspace}.
Define a set $E$ such that
 \begin{align}\label{gas_exf_0}
 E:=&\{ \exists \Gamma  \subseteq \Sigma \, | \, 1. \sigma^2_1(\Phi_{\Gamma}) \geq m\sigma^2(1+\theta(\tau))
  \textup{ or } \sigma^2_{k+1}(\Phi_{\Gamma}) \leq m\sigma^2(1-\theta(\tau)), 2.\Omega \subseteq \Gamma, 3. |\Gamma|=k+1\}. 
 \end{align}
Then, by applying the union bound to Proposition \ref{prop_guass} where $g=\theta(\tau)$, it follows that for $\theta(\tau)> \frac{k}{m}+2\sqrt{\frac{k}{m}}$,
\begin{align}\label{gas_exf_1}
\mathbb{P}(E)\leq 2(n-k)\exp \left[-\frac{m}{2} \left(\sqrt{1+\theta(\tau)}-1-\sqrt{\frac{k+1}{m}} \right)^2 \right].
\end{align}
Then 
\begin{align}\label{gas_exf_2}
(\sqrt{1+\theta(\tau)}-1)\sqrt{m}\geq \sqrt{k+1}+\sqrt{2\ln \,\left(\frac{2(n-k)}{\epsilon}\right)}
\end{align}
implies $\theta(\tau)> \frac{k}{m}+2\sqrt{\frac{k}{m}}$ and 
\begin{align}\label{gas_exf_21}
\mathbb{P}(E) \leq\epsilon.
\end{align}
From (\ref{gas_exf_2}) and the fact that 
\begin{align}\label{gas_exf_3}
\sqrt{2}\sqrt{k+1+2\ln \,\left(\frac{2(n-k)}{\epsilon}\right)} \geq \sqrt{k+1}+\sqrt{2\ln \,\left(\frac{2(n-k)}{\epsilon}\right)},
\end{align}
(\ref{gas_exf_21}) is implied by 
\begin{align}\label{gas_exf_4}
m\geq \frac{2}{(\sqrt{1+\theta(\tau)}-1)^2} \left[k+2\ln \,\left(\frac{2(n-k)}{\epsilon}\right) \right].
\end{align}
Since, from the definition of $E$, it follows that
\begin{align}\label{gas_exf_5}
 \mathbb{P} \left(\frac{\underset{i \in \Sigma \setminus \Omega }{\min} \, \sigma^2_{k+1}(\Phi_{\Omega  \cup \{i\}  })}{\underset{j \in \Sigma \setminus \Omega}{\max} \,\sigma^2_{1}(\Phi_{\Omega \cup \{j\}  })} > \tau \right) \geq \mathbb{P}(E^{c}),
\end{align}
therefore if (\ref{gas_exf_4}) holds, (\ref{3cor_thm_g3_3t1}) is satisfied  with a probability higher than $1- \epsilon$. 
\end{proof}

\begin{prop}
\label{prop_dft2}(\cite[Proposition 6.9]{lee2012subspace}) Let $\{c_1,...,c_m\} \subseteq \Sigma$ be a set of indices selected uniformly at ramdom. For $j \in [m]$, let the $j$th row of $A$ be the $c_j$th row of the $n \times n$ DFT matrix divided by $\sqrt{m}$. Suppose that $m \geq \frac{2(3+\tau)(k+1)}{3\tau^2}\,\ln \, (\frac{2(k+1)(n-k)}{\epsilon})$. Then, for a set $\Omega \subseteq \Sigma$ such that $|\Omega|=k$,
\begin{align}\nonumber
\mathbb{P}(\delta_{k+1}(A_\Omega;1) \geq \tau ) \leq \epsilon.
\end{align}
\end{prop}

\begin{prop}
\label{erc_won1} 
Let $X_0 \in \mathbb{K}^{n \times l}$ be row $k$-sparse with $\Omega \subseteq \Sigma$. Let $\bar{S}$ be an $r$-dimensional subspace of $\mathcal{R}(\Phi_{\Omega}X_0^{\Omega})$ such that ${\left \| P_{ \bar S}- P_{ \hat S} \right \|_2}\leq \eta$. 
Let $d$ be $\rank(P^{\perp}_{\mathcal{R}(\Phi_{\Gamma})}\bar{S})$.
Suppose that $\sigma_{|\Omega \cup \Gamma|}(\Phi_{\Omega \cup \Gamma }) >0$. Then submp($\hat S,\Gamma,1$) produces an index in $\Omega \setminus \Gamma$ if
\begin{align}\label{ercprop_3}
\underset{i \in \Sigma - \Omega -\Gamma}{\max}\left \| \dot\Phi^{\dagger}_{\Omega \setminus \Gamma}\dot\phi_{i} \right \|_1 
<1-  \sqrt{\frac{|\Omega \setminus \Gamma|}{d}}\frac{2 \eta}{\sigma_{|\Omega \setminus \Gamma|}(\dot{\Phi}_{\Omega \setminus \Gamma})}.
\end{align} 
\end{prop}

\begin{proof}[Proof of Proposition \ref{erc_won1}]
By Lemma \ref{lem_tnlb}, the right-hand side of (\ref{ercprop_3}) is upper bounded by
\begin{align}\label{erc_g2_nw_2}
1- \frac{2 \eta 
}{\underset{l \in \Omega \setminus \Gamma}{\max}\left \| P_{P^{\perp}_{\mathcal{R}(\Phi_\Gamma)}\bar S}\dot{\phi}_l \right \|_2}.
\end{align}
By applying (\ref{erc_g2_nw_2}) to (\ref{ercprop_3}), (\ref{ercprop_3}) implies 
\begin{align}\label{erc_g2_given2}
\underset{i \in \Sigma - \Omega -\Gamma}{\max}\left \| \dot\Phi^{\dagger}_{\Omega \setminus \Gamma}\dot\phi_{i} \right \|_1 < 1- \frac{2 \eta}{\underset{l \in \Omega \setminus \Gamma}{\max}\left \| P_{P^{\perp}_{\mathcal{R}(\Phi_\Gamma)}\bar S}\dot{\phi}_l \right \|_2}.
\end{align}
From Lemma \ref{lem_ntnub} and the following assumption
\begin{align}\label{erc_g2_nw_1}
\sigma_{|\Omega \cup \Gamma|}(\Phi_{\Omega \cup \Gamma}) >0,
\end{align}
(\ref{erc_g2_given2}) implies 
\begin{align}\label{erc_g2_given}
\underset{i \in \Sigma - \Omega -\Gamma}{\max}\left \| P_{P^{\perp}_{\mathcal{R}(\Phi_\Gamma)}\bar S}\dot{\phi}_i \right \|_2 <  \underset{l \in \Omega \setminus \Gamma}{\max}\left \| P_{P^{\perp}_{\mathcal{R}(\Phi_\Gamma)}\bar S}\dot{\phi}_l \right \|_2 -2 \eta. 
\end{align}
An upper bound of difference between the below two values for $i \in \Sigma$ are obtained as 
\begin{align}\nonumber
&|\left \|  P_{\mathcal{R}(P^{\perp}_{\mathcal{R}(\Gamma)}\hat S)}\dot\phi_i \right \|_2-\left \|  P_{\mathcal{R}(P^{\perp}_{\mathcal{R}(\Gamma)}\bar S)}\dot\phi_i \right \|_2 |\\\nonumber
&\overset{(a)} \leq \left \| ( P_{\mathcal{R}(P^{\perp}_{\mathcal{R}(\Gamma)}\hat S)} -  P_{\mathcal{R}(P^{\perp}_{\mathcal{R}(\Gamma)}\bar S)})\dot\phi_i \right \|_2 \\\nonumber
&\leq \left \|  P_{\mathcal{R}(P^{\perp}_{\mathcal{R}(\Gamma)}\hat S)} -  P_{\mathcal{R}(P^{\perp}_{\mathcal{R}(\Gamma)}\bar S)} \right \|_2\\\label{erc_cond1_1}
&\overset{(b)} \leq \eta,
\end{align}
where (a) follows from the triangle inequality, and (b) follows from Lemma \ref{proj_ineq} and  ${\left \| P_{ \bar S}- P_{ \hat S} \right \|_2}\leq \eta$. 
From (\ref{erc_cond1_1}), (\ref{erc_g2_given}) implies 
\begin{align}\label{erc_g2_nw_3}
\underset{i \in \Sigma - \Omega -\Gamma}{\max}\left \| P_{P^{\perp}_{\mathcal{R}(\Phi_\Gamma)}\hat S}\dot{\phi}_i \right \|_2 <  \underset{l \in \Omega \setminus \Gamma}{\max}\left \| P_{P^{\perp}_{\mathcal{R}(\Phi_\Gamma)}\hat S}\dot{\phi}_l \right \|_2.
\end{align}
The proof is therefore complete since (\ref{erc_g2_nw_3}) is a sufficient condition that submp($\hat S,\Gamma,1$) produces an index in $\Omega \setminus \Gamma$.

\end{proof}

\begin{prop}
\label{thm_g3_3}
Let $X_0 \in \mathbb{K}^{n \times l}$ be row $k$-sparse with $\Omega \subseteq \Sigma$. Let $\eta$ be a constant such that $\rho(\hat S)\leq \eta \leq 0.5$ with an $r$-dimensional space $\hat S$. Suppose that $X_0^{\Omega}$ is row-nondegenerate.
Then, given $\Gamma$ such that $|\Omega \cap \Gamma| \leq k-r$, submp($\hat S,\Gamma,1$) produces an index in $\Omega \setminus \Gamma$ if $\sigma_{|\Omega \cup \Gamma|}(\Phi_{\Omega \cup \Gamma})>0$ and 
\begin{align}\label{thm_g3_given3} 
s_1(\alpha_1,\beta_1,\eta,k,r)>0,
\end{align}
where
\begin{align}\nonumber
&s_1(x_1,...,x_5):=\sqrt{\frac{x_5}{x_4}}x_1 - \sqrt{1-x_2^2}-2x_3\\\nonumber
&\alpha_1:= \sigma_{|\Omega \setminus \Gamma|}(\dot\Phi_{\Omega \setminus \Gamma}) \\\nonumber
&\beta_1:= \underset{i \in \Sigma \setminus (\Omega \cup \Gamma)}{\min}\,\sigma_{|\Omega \setminus \Gamma|+1}(\dot\Phi_{(\{i\} \cup \Omega ) \setminus \Gamma}).
\end{align}
\end{prop}

\begin{proof}[Proof of Propostion \ref{thm_g3_3}]
The condition $\rho(\hat S)\leq \eta$ implies that there exists an $r$-dimensional subspace of $\mathcal{R}(\Phi_{\Omega}X_0^{\Omega})$, denoted by $\bar{S}$ satisfying ${\left \| P_{ \bar S}- P_{ \hat S} \right \|_2}\leq \eta$. Set a constant $d:=\mathcal{R}(P^{\perp}_{\mathcal{R}(\Phi_{\Gamma})} \bar S)$.

It is assumed that 
\begin{align}\label{g3_asp2_new1_1_1}
\sigma_{|\Omega \cup \Gamma|}(\Phi_{\Omega \cup \Gamma})>0.
\end{align}
Then (\ref{thm_g3_given3}) implies
\begin{align}\label{g3_given_new1}
&s_1(\alpha_1,\beta_1,\eta,|\Omega \setminus \Gamma|,d)>0,
\end{align}
since $d$ is equal to $r$ from Lemma \ref{rankbound_rownond} with (\ref{g3_asp2_new1_1_1}) and $s_1(x_1,...,x_5)$ in (\ref{thm_g3_given3}) is monotonically non-increasing for $x_4$ so that the left-hand side of (\ref{thm_g3_given3}) upper bounds $s_1(\alpha_1,\beta_1,\eta,|\Omega \setminus \Gamma|,d)$. 
By Lemmas \ref{lem_2ntnub} and \ref{lem_tnlb}, we obtain the following two conditions:
\begin{align}\label{g3_given_new3}
&\underset{i \in \Sigma \setminus (\Omega  \cup \Gamma)}{\max} {\left \| P_{\mathcal{R}(P^{\perp}_{\mathcal{R}(\Phi_{\Gamma})}\bar{S})}\dot\phi_i \right \|_2} \leq  \sqrt{1-\beta^2_1} \\\label{g3_given_new44}
&\underset{l \in \Omega \setminus \Gamma}{\max} {\left \| P_{\mathcal{R}(P^{\perp}_{\mathcal{R}(\Phi_{\Gamma})}\bar{S})}\dot\phi_l \right \|_2} \geq \sqrt{\frac{d}{|\Omega \setminus \Gamma|}}\alpha_1
\end{align}
By applying (\ref{g3_given_new3}) and (\ref{g3_given_new44}) to (\ref{g3_given_new1}), 
the following condition is obtained which is implied by (\ref{g3_given_new1})
\begin{align}\label{g3_given_new4}
\underset{i \in \Sigma - \Omega -\Gamma}{\max}{\left \| P_{\mathcal{R}(P^{\perp}_{\mathcal{R}(\Phi_{\Gamma})}\bar{S})}\dot\phi_i \right \|_2} <  \underset{l \in \Omega-\Gamma}{\max}{\left \| P_{\mathcal{R}(P^{\perp}_{\mathcal{R}(\Phi_{\Gamma})}\bar{S})}\dot\phi_l \right \|_2} -2 \eta. 
\end{align}
Then, from (\ref{erc_cond1_1}),  
(\ref{g3_given_new4}) implies 
\begin{align}\label{g3_given_new5}
\underset{i \in \Sigma - \Omega -\Gamma}{\max}{\left \| P_{\mathcal{R}(P^{\perp}_{\mathcal{R}(\Phi_{\Gamma})}\hat{S})}\dot\phi_i \right \|_2} <  \underset{l \in \Omega-\Gamma}{\max}{\left \| P_{\mathcal{R}(P^{\perp}_{\mathcal{R}(\Phi_{\Gamma})}\hat{S})}\dot\phi_l \right \|_2}.
\end{align}
Since (\ref{g3_given_new5}) is a sufficient condition that submp($\hat S,\Gamma,1$) recovers an index in $\Omega \setminus \Gamma$, the proof is complete if (\ref{thm_g3_given3}) and (\ref{g3_asp2_new1_1_1}) are satisfied. 

\end{proof}

\begin{cor}
\label{thm_g3_3_cor}
Let $X_0 \in \mathbb{K}^{n \times l}$ be row $k$-sparse with $\Omega \subseteq \Sigma$. Let $\eta$ be a constant such that $\rho(\hat S)\leq \eta \leq 0.5$ with an $r$-dimensional space $\hat S$. Suppose that $X_0^{\Omega}$ is row-nondegenerate. 
Then submp($\hat S,\o,k-r$) produces a set of $k-r$ indices $\Gamma$ such that $\Gamma \subseteq \Omega$ if $\sigma_{k}(\Phi_{\Omega})>0$ and 
\begin{align}\label{thm_g3_given3_cor} 
s_1(\bar \alpha,\bar \beta,\eta,k,r)>0,
\end{align}
where
\begin{align}\nonumber
&s_1(x_1,...,x_5):=\sqrt{\frac{x_5}{x_4}}x_1 - \sqrt{1-x_2^2}-2x_3\\\nonumber
&\bar \alpha:= \min\limits_{\underset{\textup{s.t. }|\Gamma|<k-r}{\Gamma \subseteq \Omega}} \, \sigma_{|\Omega \setminus \Gamma|}(\dot\Phi_{{\Omega} \setminus \Gamma})  \\\nonumber
&\bar \beta:=  \min\limits_{\underset{\textup{s.t. }|\Gamma|<k-r}{\Gamma \subseteq \Omega}} \,  \underset{i \in \Sigma \setminus ({\Omega} \cup \Gamma)}{\min}\,\sigma_{|\Omega \setminus \Gamma|+1}(\dot\Phi_{(\{i\} \cup {\Omega} ) \setminus \Gamma}).
\end{align}
\end{cor}

\begin{proof}[Proof of Corollary \ref{thm_g3_3_cor}]
For $\Gamma  \subseteq \Omega$ such that $|\Gamma|<k-r$, (\ref{thm_g3_given3_cor}) implies (\ref{thm_g3_given3}) (i.e., $s_1(\alpha_1,\beta_1,\eta,k,r)>0$), since $s_1(x_1,...,x_5)$ in (\ref{thm_g3_given3_cor}) is monotonically non-decreasing for $(x_1,x_2)$ so that the left-hand side of (\ref{thm_g3_given3}) lower bounds $s_1(\bar \alpha,\bar \beta,\eta,k,r)$. 
If submp($\hat S,\Gamma,1$) produces an index in $\Omega \setminus \Gamma$ for $\Gamma  \subseteq \Omega$ such that $|\Gamma|<k-r$, it is guaranteed that submp($\hat S,\o,k-r$) produces a set of $k-r$ indices $\Gamma$ such that $\Gamma \subseteq \Omega$. Therefore, the proof is complete.
\end{proof}

\begin{rem}
(\ref{thm_g3_given3_cor}) in Corollary \ref{thm_g3_3_cor} is a much milder condition than any conditions (\ref{ex1_new3prp_3_3})--(\ref{ex1_new3prp_3_2temp1}) in Proposition \ref{new3prp_3}.
\end{rem}

\begin{prop}
\label{thm1n2} 
Let $X_0 \in \mathbb{K}^{n \times l}$ be row $k$-sparse with $\Omega \subseteq \Sigma$. Let $\eta$ be a constant such that $\rho(\hat S)\leq \eta \leq 0.5$ with an $r$-dimensional space $\hat S$. Suppose that $X_0^{\Omega}$ is row-nondegenerate. Then, 
for any $\Gamma$ such that $\sigma_{|\Omega\cup \Gamma|}(\Phi_{\Omega\cup \Gamma})>0$ and $|\Omega \cap \Gamma| \geq k-r$, submp($\hat S,\Gamma,1$) produces an index in $\Omega \setminus \Gamma$ if 
\begin{align}\label{mu_sup_a1}
\underset{i \in \Sigma \setminus \Omega \cup \Gamma}{\min}\sigma^2_{|\Omega \setminus \Gamma|+1} (\dot\Phi_{\Omega \cup \{i\} \setminus \Gamma } )> 4 \eta (1-\eta).
\end{align}
\end{prop}
\begin{proof}[Proof of Proposition \ref{thm1n2}] The proof is based on \cite[Theorem $7.7$]{lee2012subspace}. The condition $\rho(\hat S)\leq \eta$ implies that there exists an $r$-dimensional subspace of $\mathcal{R}(\Phi_{\Omega}X_0^{\Omega})$, denoted by $\bar{S}$, satisfying ${\left \| P_{ \bar S}- P_{ \hat S} \right \|_2}\leq \eta$. Then it follows that 
\begin{align}\nonumber
& |\left \| P_{P^{\perp}_{\mathcal{R}(\Phi_{\Gamma})} \bar S}\dot\phi_i\right \|_2 - \left \|P_{P^{\perp}_{\mathcal{R}(\Phi_{\Gamma})} \hat S}\dot\phi_i \right \|_2|\\\nonumber
&\overset{(a)}\leq \left \| (P_{P^{\perp}_{\mathcal{R}(\Phi_{\Gamma})} \bar S}- P_{P^{\perp}_{\mathcal{R}(\Phi_{\Gamma})} \hat S})\dot\phi_i \right \|_2\\\nonumber
&\leq {\left \| P_{P^{\perp}_{\mathcal{R}(\Phi_{\Gamma})} \bar S}- P_{P^{\perp}_{\mathcal{R}(\Phi_{\Gamma})} \hat S} \right \|_2}\\\label{music_gra1}
&\overset{(b)}\leq \eta,
\end{align}
where $(a)$ follows from the triangle inequality and $(b)$ follows from Lemma \ref{proj_ineq} and ${\left \| P_{ \bar S}- P_{ \hat S} \right \|_2}\leq \eta$. 
By Lemma \ref{lem_tnlb2} and (\ref{music_gra1}), we have for all $i \in \Omega \setminus \Gamma$, 
\begin{align}\nonumber
\left \| P_{P^{\perp}_{\mathcal{R}(\Phi_{\Gamma})} \hat S}\dot\phi_i \right \|_2
\geq \left \| P_{P^{\perp}_{\mathcal{R}(\Phi_{\Gamma})} \bar S}\dot\phi_i \right \|_2 - \eta 
= 1 - \eta. 
\end{align}
Then 
\begin{align}\label{music_gra2}
\underset{i \in \Omega \setminus \Gamma}{\min} \left \| P_{P^{\perp}_{\mathcal{R}(\Phi_{\Gamma})} \hat S}\dot\phi_i \right \|_2
\geq  1 - \eta. 
\end{align}
By Lemma \ref{lem_2ntnub} and (\ref{music_gra1}), we have for all $i \in \Sigma \setminus \Omega \cup \Gamma$,
\begin{align}\nonumber
&\underset{i \in \Sigma \setminus \Omega \cup \Gamma}{\max}\left \| P_{P^{\perp}_{\mathcal{R}(\Phi_{\Gamma})} \hat S}\dot\phi_i \right \|_2 \\\nonumber
&\leq \underset{i \in \Sigma \setminus \Omega \cup \Gamma}{\max}\left \| P_{P^{\perp}_{\mathcal{R}(\Phi_{\Gamma})} \bar S}\dot\phi_i \right \|_2 + \eta\\\label{music_gra3}
&\leq \sqrt{1-\underset{i \in \Sigma \setminus \Omega \cup \Gamma}{\min}\sigma^2_{|\Omega \setminus \Gamma|+1} (\dot\Phi_{\Omega \cup \{i\} \setminus \Gamma } )}+\eta. 
\end{align}
By combing (\ref{music_gra2}) and (\ref{music_gra3}), it follows that 
\begin{align}\label{music_gra5}
1 > \sqrt{1-\underset{i \in \Sigma \setminus \Omega \cup \Gamma}{\min}\sigma^2_{|\Omega \setminus \Gamma|+1} (\dot\Phi_{\Omega \cup \{i\} \setminus \Gamma } )}+2 \eta 
\end{align}
implies 
\begin{align}\label{music_gra4}
\underset{i \in \Omega \setminus \Gamma}{\min} \left \| P_{P^{\perp}_{\mathcal{R}(\Phi_{\Gamma})} \hat S}\dot\phi_i \right \|_2 > \underset{i \in \Sigma \setminus \Omega \cup \Gamma}{\max}\left \| P_{P^{\perp}_{\mathcal{R}(\Phi_{\Gamma})} \hat S}\dot\phi_i \right \|_2.
\end{align}
Note that (\ref{music_gra4}) is a sufficient condition for submp($\hat S,\Gamma,1$) to produce an index in $\Omega \setminus \Gamma$. 
From the definition of singular value, it follows that
\begin{align}
\underset{i \in \Sigma \setminus \Omega \cup \Gamma}{\min}\sigma^2_{|\Omega \setminus \Gamma|+1} (\dot\Phi_{\Omega \cup \{i\} \setminus \Gamma } )\leq 1.
\end{align}
Then, for $\eta \leq 0.5$, (\ref{music_gra5}) is rewritten as 
\begin{align}\nonumber
\underset{i \in \Sigma \setminus \Omega \cup \Gamma}{\min}\sigma^2_{|\Omega \setminus \Gamma|+1} (\dot\Phi_{\Omega \cup \{i\} \setminus \Gamma } )> 4 \eta (1-\eta)
\end{align}
so that the proof is complete.
\end{proof}

\begin{cor}\label{cor_music_c}
Let $\eta$ be a constant such that $\rho(\hat S)\leq \eta \leq 0.5$ with an $r$-dimensional space $\hat S$. Suppose that $X_0^{\Omega}$ is row-nondegenerate and $\krank(\Phi) \geq k+v_1$. Then, for any $\Gamma \in t(k-r,v_1)$,
$\Omega \setminus \Gamma$ belongs to a set of indices produced by submp($\hat S,\Gamma,v_2$) such that $v_2\geq |\Omega \setminus \Gamma|$ if 
\begin{align}\label{mu_sup_c1}
\min\limits_{\Gamma \in t(k-r,v_1)}\underset{i \in \Sigma \setminus (\Omega \cup \Gamma) }{\min}\sigma^2_{|\Omega \setminus \Gamma|+1} (\dot\Phi_{\Omega \cup \{i\} \setminus \Gamma } )> 4 \eta (1-\eta).
\end{align}
For $J \subseteq \Sigma$, $t(a,b)$ is a family of index subsets as follows.
\begin{align}\nonumber
t(a,b) &:= \{\forall J \subseteq \Sigma | \, |J \cap \Omega| \geq a, \,|J\cup \Omega| \leq b+|\Omega| \}
\end{align}
\end{cor} 
\begin{proof}[Proof of Corollary \ref{cor_music_c}] (\ref{mu_sup_a1}) and $\sigma_{|\Omega\cup \Gamma|}(\Phi_{\Omega\cup \Gamma})>0$ from Proposition \ref{thm1n2} are satisfied for any $\Gamma \in t(k-r,v_1)$ if (\ref{mu_sup_c1}) and $\krank(\Phi) \geq k+v_1$ are satisfied. 
The event where submp($\hat S,\bar \Gamma,1$) produces an index in $\Omega \setminus \bar \Gamma$ for any $\bar \Gamma \in t(k-r,v_1)$ includes the event where submp($\hat S,\hat \Gamma,v_2$) such that $v_2\geq |\Omega \setminus \hat \Gamma|$ produces a set of indices, which includes $\Omega \setminus \hat \Gamma$, for any $\hat \Gamma \in t(k-r,v_1)$.
The proof is therefore complete.

\end{proof}

\begin{rem}
(\ref{mu_sup_c1}) in Corollary \ref{cor_music_c} is a much milder condition than any conditions (\ref{ex_msc1})--(\ref{ex_msc4}) in Theorem \ref{propmc_st3}.

\end{rem}

\ifCLASSOPTIONcaptionsoff
  \newpage
\fi

\bibliographystyle{IEEEtran}
\bibliography{IEEEabrv,bibdata}

\end{document}